\documentclass{article}

\usepackage[nonatbib, final]{neurips_2020}

\usepackage{cite}
\usepackage[utf8]{inputenc} %
\usepackage[T1]{fontenc}    %
\usepackage{hyperref}       %
\usepackage{url}            %
\usepackage{booktabs}       %
\usepackage{amsfonts}       %
\usepackage{nicefrac}       %
\usepackage{microtype}      %
\usepackage{multirow}
\usepackage[para,online,flushleft]{threeparttable}
\usepackage{cite}
\usepackage{dsfont}
\usepackage{amsmath,amssymb,amsfonts}
\usepackage{upgreek}
\usepackage{endnotes}

\usepackage{algorithm,algorithmic} 
\usepackage{bbold}
\usepackage[font=small]{caption}
\usepackage[dvips]{graphicx}
\usepackage{colortbl}
\usepackage{array, tabularx}
\usepackage{amsthm}
\usepackage{mathtools}
\usepackage{setspace}
\usepackage{soul}
\usepackage{bm}
\usepackage{bbm}
\usepackage{mathtools}
\usepackage[utf8]{inputenc}
\usepackage[english]{babel}
\usepackage{enumitem}
\usepackage{mathtools}
\usepackage{subfig}
\usepackage{float}

\newtheorem{theorem}{Theorem}
\newtheorem{lemma}{Lemma}

\newtheorem{definition}{Definition}

\newtheorem{prop}{Proposition}
\newtheorem{remark}{Remark}
\newtheorem{assumption}{Assumption}

\newcommand{\maj}{\mathrm{maj}}
\newcommand{\sign}{\mathrm{sign}}

\newcommand{\Th}{\textrm{th}}

\renewcommand{\epsilon}{\varepsilon}

\newcommand\floor[1]{\left\lfloor{#1} \right\rfloor}
\DeclarePairedDelimiter{\ceil}{\lceil}{\rceil}

\newcommand{\mtx}{\mathbf}
\newcommand{\vct}{\bm}
\newcommand\bmt[1]{\textcolor{black}{#1}}

\newcommand\PP[1]{\mathds{P}{(#1)}}

\newcommand{\RR}{\mathds{R}}
\newcommand{\EE}[1]{\mathds{E}\left[{#1}\right]}
\newcommand\ZNorm[1]{\lVert{#1}\rVert_0}

\newcommand{\ept}[1]{\mathds{E}\left[{#1}\right]}
\makeatletter
\newcommand{\@giventhatstar}[2]{\left(#1\;\middle|\;#2\right)}
\newcommand{\@giventhatnostar}[3][]{#1(#2\;#1|\;#3#1)}
\newcommand{\giventhat}{\@ifstar\@giventhatstar\@giventhatnostar}
\makeatother
\makeatletter
\newcommand\footnoteref[1]{\protected@xdef\@thefnmark{\ref{#1}}\@footnotemark}
\makeatother

\newif\ifarxiv
\arxivtrue

\title{Election Coding for Distributed Learning: \\Protecting SignSGD against Byzantine Attacks}

\author{
	Jy-yong Sohn\\
	\texttt{jysohn1108@kaist.ac.kr}\\
	\And
	Dong-Jun Han\\
	\texttt{djhan93@kaist.ac.kr}\\
	\AND
	Beongjun Choi\\
	\texttt{bbzang10@kaist.ac.kr}\\
	\And
	Jaekyun Moon\\
	\texttt{jmoon@kaist.edu}\\
	\AND
	\textnormal{School of Electrical Engineering,}\\
	\textnormal{Korea Advanced Institute of Science and Technology (KAIST)}
}

\begin{document}
	
	\maketitle
	
	\begin{abstract}
		Current distributed learning systems suffer from serious performance degradation under Byzantine attacks. 
		This paper proposes \textsc{Election Coding}, a coding-theoretic framework to guarantee Byzantine-robustness for distributed learning algorithms 
		based on signed stochastic gradient descent (SignSGD) that minimizes the worker-master communication load.
		The suggested framework explores new information-theoretic limits of 
		finding the majority opinion 
		when some workers could be attacked by adversary, and paves the road to implement robust and communication-efficient distributed learning algorithms. Under this framework, we construct two types of codes, random Bernoulli codes and deterministic algebraic codes, that tolerate Byzantine attacks with a controlled amount of computational redundancy and guarantee convergence in general non-convex scenarios. 
		For the Bernoulli codes, we provide an upper bound on the error probability in estimating the signs of the true gradients, which gives useful insights into code design for Byzantine tolerance. The proposed deterministic codes are proven to perfectly tolerate arbitrary Byzantine attacks.
		Experiments on real datasets confirm that the suggested codes provide substantial improvement in Byzantine tolerance of distributed learning systems
		employing SignSGD.
	\end{abstract}
	\vspace{-5mm}
	
	\section{Introduction}
	\vspace{-2mm}
	The modern machine learning paradigm is moving toward parallelization and decentralization \cite{ben2018demystifying, dean2012large, lian2017can} to speed up the training and provide reliable solutions to time-sensitive real-world problems. 
	There has been extensive work on developing distributed learning algorithms \cite{abadi2016tensorflow, paszke2017automatic, seide2016cntk, recht2011hogwild, mcmahan2016communication, konevcny2016federated} to exploit large-scale computing units.  
	These distributed algorithms are usually implemented in parameter-server (PS) framework \cite{li2014scaling}, where a central PS (or master) aggregates the computational results (e.g., gradient vectors minimizing empirical losses) of distributed workers to update the shared model parameters. 
	In recent years, two issues have emerged as major drawbacks that 
	limit the performance of distributed learning: 
	\textit{Byzantine attacks} and \textit{communication burden}.
	
	\vspace{-1mm}
	Nodes affected by Byzantine attacks send arbitrary messages to PS, which would mislead the model updating process and severely degrade learning capability. 
	To counter the threat of Byzantine attacks,  much attention has been focused on robust solutions \cite{alistarh2018byzantine, pmlr-v80-mhamdi18a, elmhamdi2019fast}.
	Motivated by the fact that a naive linear aggregation at PS cannot even tolerate Byzantine attack on a single node, the authors of \cite{chen2017distributed, blanchard2017machine, pmlr-v80-yin18a} considered median-based aggregation methods.
	However, as data volume and the number of workers increase, computing the median involves a large cost \cite{blanchard2017machine} far greater than the cost for batch gradient computations. Thus, recent works  \cite{chen2018draco, rajput2019detox} instead suggested redundant gradient computation that tolerates Byzantine attacks.
	
	\vspace{-1mm}
	Another issue is the high communication burden caused by transmitting gradient vectors between PS and workers for updating network models.
	Regarding this issue, the authors of \cite{wang2018atomo, lin2017deep, pmlr-v80-wu18d, pmlr-v80-bernstein18a, bernstein2018signsgd_ICLR, alistarh2017qsgd, wen2017terngrad, pmlr-v97-karimireddy19a} considered quantization of real-valued gradient vectors.
	The signed stochastic gradient descent method (\textsc{SignSGD}) suggested in \cite{pmlr-v80-bernstein18a} compresses a real-valued gradient vector $\vct{g}$ into a binary
	vector $\sign(\vct{g})$, and updates the model using the 1-bit compressed gradients. This scheme minimizes the communication load from PS to each worker for transmitting the aggregated gradient.
	A further variation called \textsc{SignSGD with Majority Vote (SignSGD-MV)}\cite{pmlr-v80-bernstein18a, bernstein2018signsgd_ICLR} also applies 1-bit quantization on gradients communicated from each worker to PS in achieving minimum master-worker communication in both directions.
	These schemes have been shown to minimize the communication load while maintaining the SGD-level convergence speed in general non-convex problems. A major issue that remains is the lack of 
	Byzantine-robust solutions suitable for such communication-efficient learning algorithms.

	\begin{figure}[!t]
		\vspace{-5mm}
		\centering\	
		\subfloat[][Conventional Setting \textbf{vs} Suggested \textsc{Election Coding}] {\includegraphics[height=27mm]
			{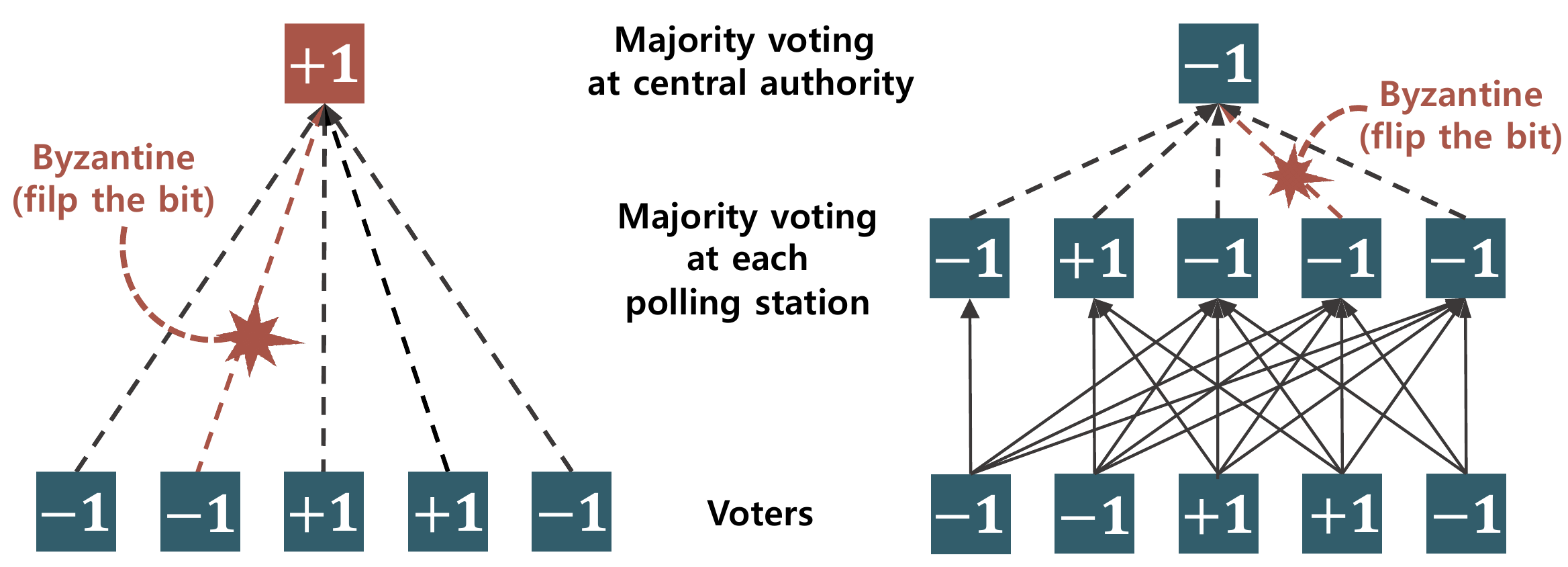}	\label{Fig:HierVote_Comp}}
		\vspace{-2mm} 
		\quad
		\subfloat[][Suggested scheme applied to \textsc{SignSGD}] {\includegraphics[height=23mm]
			{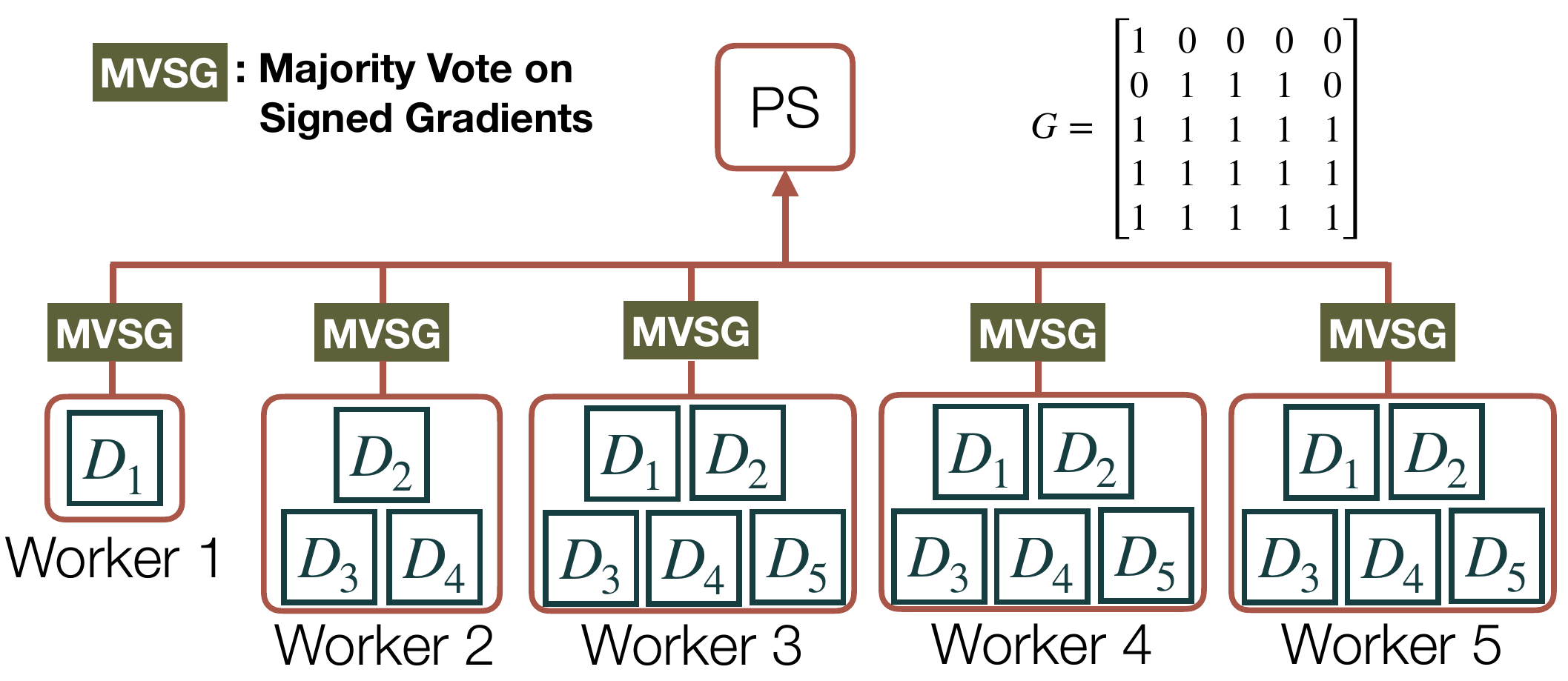}	\label{Fig:HierVote_Sugg}}
		\caption{\bmt{Use of coding to protect the majority vote. (a) The conventional scheme (without coding) is vulnerable to Byzantine attacks, while the suggested scheme is not. In the suggested scheme, each polling station gathers the votes of a subset of voters, and sends the majority value to the master. (b) In the real setup, voters are like data partitions $D_i$ while polling stations represent workers. 
		}}
		\label{Fig:HierVote}
		\vspace{-5mm}
	\end{figure}

	\vspace{-1mm}
	In this paper, we propose \textsc{Election Coding}, a coding-theoretic framework to make 
	\textsc{SignSGD-MV}\cite{pmlr-v80-bernstein18a} highly robust to Byzantine attacks. 
	In particular, we focus on estimating the next step for model update used in~\cite{pmlr-v80-bernstein18a}, i.e., the \textit{majority} voting on the signed gradients extracted from $n$ data partitions,
	under the scenario where $b$ of the $n$ worker nodes are under Byzantine attacks. 
	
	\vspace{-1mm}
	Let us illustrate the concept of the suggested framework as a voting scenario where $n$ people vote for either one candidate ($+1$) or the other  ($-1$). Suppose that each individual must send her vote directly to the central election commission, or a master. Assume that frauds can happen during the transmittal of votes, possibly flipping the result of a closely contested election, \bmt{as a single fraud did in the first example of Fig.~\ref{Fig:HierVote_Comp}}. A simple strategy can effectively combat this type of voting fraud. First set up multiple polling stations and let each person go to multiple stations to cast her ballots. Each station finds the majority vote of its poll and sends it to the central commission. Again a fraud can happen as the transmissions begin. However, \bmt{as seen in the second example of Fig.~\ref{Fig:HierVote_Comp}}, the single fraud was not able to change the election result, thanks to the built-in redundancy in the voting process. In this example, coding amounts to telling each voter to go to which polling stations. In the context of SignSGD, the individual votes are like the locally computed gradient signs that must be sent to the PS, and through some ideal redundant allocation of data partitions we wish to protect the integrity of the gathered gradient computation results under Byzantine attacks on locally computed gradients.

	\vspace{-3mm}
	\paragraph{Main contributions:}
	Under this suggested hierarchical voting framework, we construct two \textsc{Election} coding schemes: random Bernoulli codes and deterministic algebraic codes.
	Regarding the random Bernoulli codes, which are based on arbitrarily assigning data partitions to each node with probability $p$, we obtain an upper bound on the error probability in estimating the sign of the true gradient. Given $p = \Theta (\sqrt{\log (n) /n
	})$, the estimation error vanishes to zero for arbitrary $b$, under the asymptotic regime of large $n$. Moreover, the convergence of Bernoulli coded systems are proven in general non-convex optimization scenarios.
	As for the deterministic codes, we first obtain the necessary and sufficient condition on the data allocation rule, in order to accurately estimate the majority vote under Byzantine attacks. Afterwards, we suggest an explicit coding scheme which achieves perfect Byzantine tolerance for arbitrary $n,b$. %
	\vspace{-1mm}
	Finally, the mathematical results are confirmed by simulations on well-known machine learning architectures. We implement the suggested coded distributed learning algorithms in PyTorch, and deploy them on Amazon EC2 using Python with MPI4py package.
	We trained \textsc{Resnet-18} using \textsc{Cifar-10} dataset as well as a logistic regression model using Amazon Employee Access dataset. The experimental results confirm that the suggested coded algorithm has significant advantages in tolerating Byzantines compared to the conventional uncoded method, under various attack scenarios.

	\vspace{-3mm}
	\paragraph{Related works:}
	The authors of \cite{chen2018draco} suggested a coding-theoretic framework \textsc{Draco} for Byzantine-robustness of distributed learning algorithms.  Compared to the codes in \cite{chen2018draco}, our codes have the following two advantages. First, our codes are more suitable for \textsc{SignSGD} setup (or in general compressed gradient schemes) with limited communication burden. \bmt{The codes proposed in \cite{chen2018draco} were designed for real-valued gradients, while our codes are intended for quantized/compressed gradients.}
	Second, the random Bernoulli codes suggested in this paper can be designed in a more flexible manner.
	The computational redundancy $r=2b+1$ of the codes in \cite{chen2018draco} linearly increases, which is burdensome for large $b$.
	As for Bernoulli codes proposed in this paper, we can control the redundancy by choosing an appropriate connection probability $p$.
	Simulation results show that our codes having a small expected redundancy of $\mathds{E}[r]=2, 3$ enjoy significant gain compared to the uncoded scheme for various $n,b$ settings.
	A recent work  \cite{rajput2019detox} suggested a framework \textsc{Detox} which combines two existing schemes: computing redundant gradients and robust aggregation methods. However, \textsc{Detox} still suffers from a high computational overhead, since it relies on 
	the computation-intensive geometric median aggregator. 
	For communicating 1-bit compressed gradients, a recent work \cite{bernstein2018signsgd_ICLR} analyzed the Byzantine-tolerance of the naive \textsc{SignSGD-MV} scheme. This scheme can only achieve a limited accuracy as $b$ increases, whereas the proposed coding schemes can achieve high accuracy in a wide range of $b$ as shown in the experimental results provided in Section \ref{Sec:Simul}.  Moreover, as proven in Section~\ref{Sec:AlgCode}, the suggested deterministic codes achieve the ideal accuracy of $b=0$ scenario, regardless of the actual number of Byzantines in the system. 
	\vspace{-3mm}
	\paragraph{Notations:}
	{The sum of elements of vector $\vct{v}$ is denoted as $\ZNorm{\vct{v}}$. Similarly, $\ZNorm{\mtx{M}}$ represents the sum of elements of matrix $\mtx{M}$. The sum of the absolute values of the vector elements is denoted as $||{\vct{v}}||_1$. An $n \times n$ identity matrix is denoted as $\mtx{I}_n$.
		The set $\{1, 2, \dots, n\}$ is denoted by $[n]$.
		An $n \times k$ all-ones matrix
		is denoted as $\mathbb{1}_{n \times k}$.}

	\vspace{-2mm}
	\section{Suggested framework for Byzantine-robust distributed learning}
	\vspace{-1mm}
	\subsection{Preliminary: \textsc{SignSGD with Majority Vote (SignSGD-MV)}}
	
	\begin{figure}
		\vspace{-5mm}
		\centering
		\includegraphics[width=0.97\linewidth]{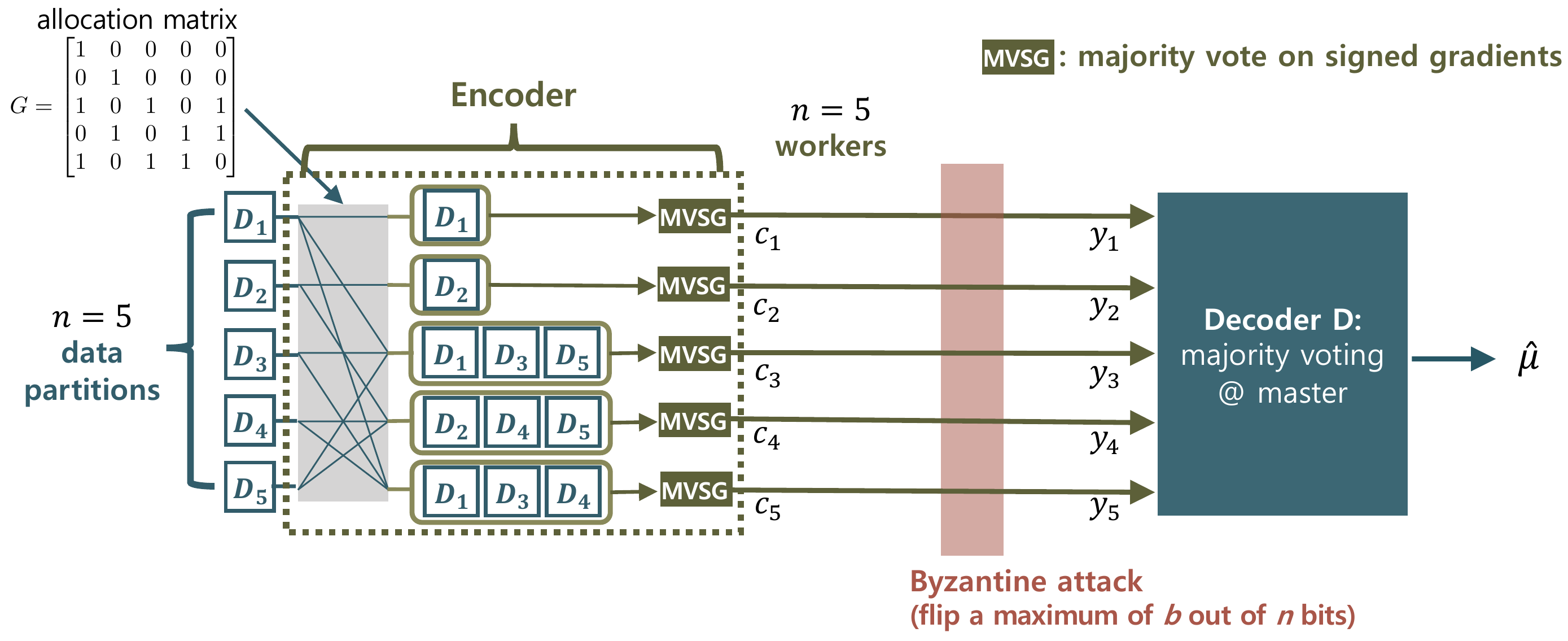}
		\caption{A formal description of the suggested \textsc{Election Coding} framework for estimating the majority opinion $\mu$. This framework is applied for each coordinate of the model parameter $\vct{w} \in \RR^d $ in a parallel manner. 
		}
		\label{Fig:SysMod}
		\vspace{-5mm}
	\end{figure}
	
	\vspace{-1mm}
	Here we review \textsc{SignSGD with Majority Vote} \cite{pmlr-v80-bernstein18a, bernstein2018signsgd_ICLR} applied to distributed learning setup with $n$ workers, where the goal is to optimize model parameter $\vct{w} \subseteq \RR^d$. 
	Assume that each training data is sampled from distribution $\mathcal{D}$.
	We divide the training data into $n$ partitions, denoted as $\{D_j\}_{j \in [n]}$.
	Let $\vct{g} = [g^{(1)},  \cdots, g^{(d)}]$ be the gradient calculated when the whole training data is given. Then, the output of the stochastic gradient oracle for an input data point $x$ is denoted by $\tilde{\vct{g}}(x) = [\tilde{g}^{(1)}(x), \cdots,\tilde{g}^{(d)}(x)]$, which is an estimate of the ground-truth $\vct{g}$. Since processing is identical across $d$ coordinates or dimensions, we shall simply focus on one coordinate $k \in [d]$, with the understanding that encoding/decoding and computation are done in parallel across all coordinates. The superscript will be dropped unless needed for clarity.    
	At each iteration, $B$ data points are selected for each data partition. We denote the set of data points selected for the $j^{\Th}$ partition as $\mathcal{B}_j \subseteq D_j$, satisfying $\lvert \mathcal{B}_j \rvert = B$.
	The output of the stochastic gradient oracle for $j^{\Th}$ partition is expressed as 
	$\tilde{g}_j = \sum_{x \in \mathcal{B}_j} \tilde{g} (x)/B$. 
	For a specific coordinate, the set of gradient elements computed for $n$ data partitions is denoted as 
	$\tilde{\vct{g}}= [\tilde{g}_1,  \cdots, \tilde{g}_n]$.
	The sign of the gradient is represented as
	$\vct{m}  =  [m_1,  \cdots, m_{n}]$
	where $m_j = \sign(\tilde{g}_j) \in \{+1,-1\}$.
	We define the \textit{majority opinion} as 
	$\mu = \maj (\vct{m}),$ 
	where $\maj(\cdot)$ is the \textit{majority} function which outputs the more frequent element in the input argument. %
	At time $t$, \textsc{SignSGD-MV} updates the model as
	$w_{t+1}  = w_t - \gamma \mu, $
	where $\gamma$ is the learning rate.

	\vspace{-1mm}
	\subsection{Proposed \textsc{Election Coding} framework}\label{Sec:SysParam}
	
	\vspace{-1mm}
	The suggested framework for estimating the majority opinion $\mu$ is illustrated in Fig. \ref{Fig:SysMod}. 
	This paper suggests applying codes for allocating data partitions into worker nodes. 
	The data allocation matrix $\mtx{G} \in \{0,1\}^{n \times n}$ is defined as follows: $G_{ij} = 1$ if data partition $j$ is allocated to node $i$, and $G_{ij} = 0$ otherwise. We define $P_i = \{ j : G_{ij}=1 \}$, the set of data partitions assigned to node $i$.
	Given a matrix $\mtx{G}$, the computational redundancy compared to the uncoded scheme is expressed as 
	$r = \ZNorm{\mtx{G}}/n$, the average number of data partitions handled by each node.
	Note that the uncoded scheme corresponds to $\mtx{G}=\mtx{I}_{n}$. 
	Once node $i$ computes $\{m_j\}_{j \in P_i}$ from the assigned data partitions, it generates a bit $c_{i} = E_i ( \vct{m} ; \mtx{G} ) =  \maj(\{m_j\}_{j\in P_i})$ 
	using encoder $E_i$. In other words, node $i$ takes the majority of the signed gradients obtained from partitions $D_j$ observed by the node. We denote the $n$ bits generated from worker nodes by $\vct{c} = [c_{1}, \cdots c_{n}]$.
	After generating $c_{i} \in \{+1,-1\}$, node $i$ transmits
	\begin{equation}\label{Eqn:y_j}
	y_{i} = 
	\begin{cases}
	\mathcal{X}, & \text{ if } \text{ node } i \text{ is a Byzantine\footnotemark node} \\
	c_{i}, & \text{ otherwise}
	\end{cases}
	\end{equation}
	to PS, where $\mathcal{X} \in \{c_{i}, - c_{i} \}$ holds since each node is allowed
	to transmit either $+1$ or $-1$. 
	We denote the number of Byzantine nodes as $b=n \alpha$, where $\alpha$ is the portion corresponding to adversaries.
	PS observes $\vct{y} = [y_{1}, \cdots, y_{n}]$ and estimates 
	$\mu$ using a decoding function $D$.
	Using the output $\hat{\mu}$ of the decoding function, PS updates the model as ${\omega}_{t+1} = {\omega}_{t} - \gamma \hat{\mu}$.
	\footnotetext{We assume in this work that Byzantine attacks take place in the communication links and that the distributed nodes themselves are authenticated and trustful. As such, \textit{Byzantine} nodes in our context are the nodes whose transmitted messages are compromised.
	}

	\vspace{-1mm}
	There are three design parameters which characterize the suggested framework: the data allocation matrix $\mtx{G}$, the encoder function $E_i$ at worker $i \in [n]$, and the decoder function $D$ at the master. 
	In this paper, we propose low-complexity hierarchical voting where both the encoder/decoder are majority voting functions.
	Under this suggested framework, 
	we focus on devising $\mtx{G}$ which tolerates Byzantine attacks on $b$ nodes. Although this paper focuses on applying codes for SignSGD, the proposed idea can easily be extended to multi-bit quantized SGD with a slight modification: instead of finding the majority value at the encoder/decoder, we take the average and then round it to the nearest quantization level.
	
	\vspace{-3mm}
	\section{Random Bernoulli codes }\label{Sec:Bernoulli}
	\vspace{-2mm}
	
	We first suggest random Bernoulli codes, where each node randomly selects each data partition with connection probability $p$ independently. Then, $\{G_{ij}\}$ are independent and identically distributed Bernoulli random variables with $G_{ij} \sim Bern(p)$. 
	The idea of random coding is popular in coding and information theory (see, for example, \cite{mackay2005fountain}) and has also been applied to computing distributed gradients \cite{charles2017approximate}, but random coding for tolerating Byzantines in the context of SignSGD is new and  
	requires unique analysis.
	Note that depending on the given Byzantine attack scenario, flexible code construction is possible by adjusting the connection probability $p$.  
	Before introducing theoretical results on random Bernoulli codes, we clarify the assumptions used for the mathematical analysis. 
	We assume several properties of loss function $f$ and gradient $g = \nabla f$, summarized as below. Assumptions~\ref{Assumption:lower_bound},~\ref{Assumption:smooth},~\ref{Assumption:variance_bound} are commonly used in proving the convergence of general non-convex optimization,
	and Assumption~\ref{Assumption:symmtric_unimodal} is used and justified in previous works on \textsc{SignSGD}~\cite{pmlr-v80-bernstein18a, bernstein2018signsgd_ICLR}. 

	\begin{assumption}\label{Assumption:lower_bound}
		For arbitrary $x$, the loss value is bounded as $f(x) \geq f^{\star}$ for some constant $f^{\star}$. 
	\end{assumption}
	
	\begin{assumption}\label{Assumption:smooth}
		For arbitrary $x,y$, there exists a vector of Lipschitz constants $\vct{L} = [L^{(1)}, \cdots, L^{(d)}]$ satisfying
		$|f(y)-[f(x)+g(x)^{T}(y-x)]| \leq \frac{1}{2} \sum_{i} L^{(i)}(y^{(i)}-x^{(i)})^{2}$.
	\end{assumption}

	\begin{assumption}\label{Assumption:variance_bound}
		The output of the stochastic gradient oracle is the unbiased estimate on the ground-truth gradient, where the variance of the estimate is bounded at each coordinate. In other words, for arbitrary data point $x$, we have
		$\mathds{E}[\tilde{g}(x)] = g$  and $\mathds{E}[\left(\tilde{g}(x) - g \right) ^2 ] \leq \sigma^2.$
	\end{assumption}

	\begin{assumption}
		\label{Assumption:symmtric_unimodal}
		The component of the gradient noise for data point $x$, denoted by  $ \tilde{g} (x) - g$, has a unimodal distribution that is symmetric about the zero mean, for all coordinates and all $x \sim \mathcal{D}$.
	\end{assumption}

	\vspace{-2mm}
	\subsection{Estimation error bound}
	\vspace{-2mm}
	In this section, we measure the probability of the PS correctly estimating the sign of the true gradient, 
	under the scenario of applying suggested random Bernoulli codes. To be specific, we compare $\sign(g)$, the sign of the true gradient, and $\hat{\mu}$, the output of the decoder at the master. Before stating the first main theorem which provides an upper bound on the estimation error $\PP{ \hat{\mu} \neq \sign(g)} $, we start from finding the statistics of the mini-batch gradient $\tilde{g}_j$ obtained from data partition $D_j$, the proof of which is in
	Supplementary Material.
	\begin{prop}[Distribution of batch gradient]\label{Lemma:symm_unimodal_partition}
		The mini-batch gradient $\tilde{g}_j$ for data partition $D_j$ follows a unimodal distribution that is symmetric around the mean $g$. The mean and variance of  $\tilde{g}_j$ are given as 
		$\mathds{E} [ \tilde{g}_j ]= g,$ and $
		\textrm{var}(\tilde{g}_j) = \mathds{E} [ ( \tilde{g}_j - g)^2 ]  \leq \bar{\sigma}^2 \coloneqq  \sigma^2/B.$
	\end{prop}

	\begin{definition}\label{Def:S_k}
		Define  $S^{(k)} = \left|g^{(k)}\right|/\bar{\sigma}^{(k)}$ as the signal-to-noise ratio (SNR) of the stochastic gradient observed at each mini-batch, for coordinate $k \in [d]$.
	\end{definition}

	Now, we measure the estimation error of random Bernoulli codes. Recall that we consider a hierarchical voting system, where each node makes a local decision on the sign of gradient,
	and the PS makes a global decision by aggregating the local decisions of distributed nodes. 
	We first find the local estimation error $q$ of a Byzantine-free (or intact) node $i$ as below, which is proven in %
	Supplementary Material.
	\begin{lemma}[Impact of local majority voting]\label{Thm:q_k_updated}
		Suppose $p \geq p^{\star} = 2  \sqrt{\frac{C \log(n)}{n} }$ for some $C>0$. %
		Then, the estimation error probability of a Byzantine-free node $i$ on the sign of gradient along a given coordinate is bounded:
		\vspace{-1mm}
		\begin{align}\label{Eqn:q_k}
		q = \PP{c_{i} \neq \sign(g)}   
		& \leq q^{\star} \coloneqq  2 \cdot \max  \big\{   \frac{2}{n^{2C}},  {\rm e}^{- \sqrt{ C n \log(n)} \frac{S^2}{2(S^2+4)}} \big\},
		\end{align}
	\end{lemma}
	\vspace{-2mm}
	where the superscripts on $g$ and $S$ that point to the coordinate are dropped to ease the nontation.	
	As the connection probability $p$ increases, we can take a larger $C$, thus the \emph{local} estimation error bound decreases. This makes sense because each node tends to make a correct decision as the number of observed data partitions increases.

	Now, consider the scenario with $b=n\alpha$ Byzantine nodes and $n(1-\alpha)$ intact nodes. 
	From \eqref{Eqn:y_j}, an intact node $i$ sends $y_{i} = c_{i}$ for a given coordinate. Suppose all Byzantine nodes send the reverse of the true gradient, i.e., $y_{i} = - \sign(g)$, corresponding to the worst-case scenario which maximizes the global estimation error probability $P_{\textrm{global}} = 	\PP{ \hat{\mu} \neq \sign(g)} $.
	Under this setting, the global error probability at PS is bounded as below,
	the proof of which is given in
	Supplementary Material.
	
	\begin{theorem}[Estimation error at master]\label{Thm:global_error_majority}
		Consider the scenario with connection probability $p = \Theta(\sqrt{\log(n)/n})$ as in Lemma~\ref{Thm:q_k_updated}. Suppose the portion of Byzantine-free nodes satisfies 
		\begin{align}\label{Eqn:condition_portion}
		1-\alpha > 
		\frac{(\sqrt{\log(\Delta)/n} + \sqrt{\log(\Delta)/n + 4 u_{\textrm{min}}^{\star}})^2}{8 (u_{\textrm{min}}^{\star})^2} 
		\end{align}
		for some $\Delta > 2$,
		where $u_{\textrm{min}}^{\star} = 1 - \max\limits_{k } \{ q^{(k)\star}\}$ is a lower bound on the local estimation success probability.
		Then, for each coordinate, the global estimation error at PS is bounded as
		\begin{align*}
		P_{\textrm{global}} = 	\PP{ \hat{\mu} \neq \sign(g)}   
		< 1/\Delta. %
		\end{align*}
	\end{theorem}
	This theorem specifies the sufficient condition on the portion of Byzantine-free nodes that allows the global estimation error smaller than $1/\Delta$.
	Suppose $n$ and $p$ are given, thus $u_{\textrm{min}}^{\star}$ is fixed. As $\Delta$ increases, the right-hand-side of~\eqref{Eqn:condition_portion} also increases. 
	This implies that to have a smaller error bound (i.e., larger $\Delta$), the portion of Byzantine-free nodes needs to be increased, which makes sense.%

	\begin{remark}\label{Rmk}
		In the asymptotic regime of large $n$, the condition \eqref{Eqn:condition_portion} reduces to $1-\alpha > \frac{1}{2u_{\textrm{min}}^{\star}} \rightarrow \frac{1}{2}$ for arbitrary $\Delta$, since $q^{\star} \rightarrow 0$ as shown in~\eqref{Eqn:q_k}. 
		This implies that the global estimation error vanishes to zero even in the extreme case of having maximum Byzantine nodes $b=n/2$, provided that the number of nodes $n$ is in the asymptotic regime and the connection probability satisfies $p = \Theta(\sqrt{\log(n)/n})$. 
	\end{remark}
	
	Remark~\ref{Rmk} states the behavior of global error bound for asymptotically large $n$. In a more practical setup with, say, a Byzantines portion of $\alpha=0.2$, the batch size of $B=128$, the connection factor of $C=1$, and the SNR of $\lvert g \rvert / \sigma = 1$, the global error is bounded as $P_{\textrm{global}} < 0.01$ given $n \geq 40$.

	\subsection{Convergence analysis}
	The convergence of the suggested scheme can be formally stated as follows:
	
	\begin{theorem}[Convergence of the Bernoulli-coded \textsc{SignSGD-MV}]\label{Thm:Convergence}
		Suppose that Assumptions~\ref{Assumption:lower_bound},~\ref{Assumption:smooth},~\ref{Assumption:variance_bound},~\ref{Assumption:symmtric_unimodal} hold and the portion of Byzantine-free nodes satisfies~\eqref{Eqn:condition_portion} for $\Delta>2$.
		Apply the random Bernoulli codes with connection probability $p = \Theta(\sqrt{\log(n)/n})$ on SignSGD-MV, and run SignSGD-MV for $T$ steps with an initial model $\vct{w}_0$. Define the learning rate as $\upgamma(T)=\sqrt{\frac{f(\vct{w}_{0})-f^{*}}{\|\vct{L}\|_{1} T}}$.
		Then, the suggested scheme converges as $T$ increases in the sense
		\vspace{-2mm}
		\begin{align*}
		\frac{1}{T} \sum_{t=0}^{T-1} \EE{\left\|g(\vct{w}_{t})\right\|_{1}} \leq 
		\frac{3 \lVert \vct{L} \rVert_1}{2 (1-2/\Delta)} \upgamma(T) \rightarrow 0 \textrm{ \quad as \  } T \rightarrow \infty.
		\end{align*}
	\end{theorem}
	\vspace{-2mm}
	A full proof is given in
	Supplementary Material. Theorem~\ref{Thm:Convergence} shows that the average of the absolute value of the gradient $g(\vct{w}_t)$ converges to zero, meaning that the gradient itself becomes zero eventually. This implies that the algorithm converges to a stationary point.
	
	\subsection{\bmt{Few remarks regarding random Bernoulli codes}}

	\begin{remark}
		\bmt{We have assumed that the number of data partitions $k$ is equal to the number of workers $n$. 	
			What if $n \neq k$? In this case, the required connection probability for Byzantine tolerance in Lemma 1 behaves as $p \sim \sqrt{\log(k)/k}$. Since the computational load at each worker is proportional to $p$, the load can be lowered by increasing $k$. However, the local majority vote becomes less accurate with increasing $k$ since each worker uses fewer data points, i.e., parameter $S$ in Definition~\ref{Def:S_k} deteriorates. 
		} 
	\end{remark}

	\begin{remark}
		\bmt{What is the minimal code redundancy while providing (asymptotically) perfect Byzantine tolerance?
			From the proof of Lemma~\ref{Lemma:symm_unimodal_partition}, we see that
			if the connection probability (which is proportional to redundancy) can be expressed as $p = f(k)/\sqrt{k}$ with $f(k)$ being a growing function of $k$,
			the number of data partition, then the random codes achieve
			full Byzantine protection (in the asymptotes of $k$).
			Our choice $p \sim \sqrt{\log(k)/k}$ meets this, although it is not proven whether this represents minimal redundancy.
		}
	\end{remark}
	
	\begin{remark}
		\bmt{The worker loads of random Bernoulli codes are non-uniform, but the load distribution tends to uniform as $n$ grows, due to concentration property of binomial distribution.}%
	\end{remark}

	\begin{algorithm}[t]
		\small
		\caption{\small Data allocation matrix $\mtx{G}$ satisfying perfect $b-$Byzantine tolerance
			($0 < b <\floor{n/2}$)
		}
		\label{Alg:DetCode}
		\begin{algorithmic}
			\STATE \textbf{Input:} Number of nodes $n$, number of Byzantine nodes $b$.  \STATE \textbf{Output:} Data allocation matrix $\mtx{G} \in \{0,1\}^{n \times n}$ that achieves perfect $b-$Byzantine tolerance.
			\STATE \textbf{Initialize:} Define $s = \frac{n-1}{2} -b$ and $L = \lfloor \frac{n-(2b+1)}{2(b+1)} \rfloor +1$. Initialize $\mtx{G}$ as the all-zero matrix.
			\STATE \textbf{Step 1:} Set the top left $s$-by-$s$ submatrix of $\mtx{G}$ as identity matrix, i.e., $\mtx{G}(1:s, 1:s) = \mtx{I}_s$.
			\STATE \textbf{Step 2:} Set the bottom $(n-s-L)$ rows as the all-one matrix, i.e., $\mtx{G}(s+L+1:n, :) = \mathbb{1}_{(n-s-L) \times n}$.
			\STATE \textbf{Step 3:} Fill in the matrix $\mtx{A} = \mtx{G}(s+1:s+L, s+1:n)$ as follows: Insert $2b+1$ ones on each row by shifting the location by $b+1$, i.e., $\mtx{A}(l, (l-1)(b+1) + (1:2b+1)) = \mathbb{1}_{1 \times (2b+1)}$ for $l=1,\cdots, L$.
		\end{algorithmic}
	\end{algorithm}

	\vspace{-3mm}
	\section{Deterministic codes for perfect Byzantine tolerance}\label{Sec:AlgCode}
	\vspace{-3mm}
	
	We now construct codes that guarantee perfect Byzantine tolerance. We say the system is \emph{perfect} $b$-Byzantine tolerant if $\PP{\hat{\mu} \neq \mu } = 0 $ holds on all coordinates\footnote{According to the assumption on the distribution of $\tilde{g}_j$, no codes satisfy $\PP{\hat{\mu} \neq \sign(g)} = 0$. Thus, we instead aim at designing codes that always satisfy $\hat{\mu} = \mu$; when such codes are applied to the system, the learning algorithm is guaranteed to achieve the ideal performance of SignSGD-MV~\cite{bernstein2018signsgd_ICLR} with no Byzantines.} under the attack of $b$ Byzantines. We devise codes that enable such systems. 
	Trivially, no probabilistic codes achieve this condition, and thus we focus on \emph{deterministic} codes where the entries of the allocation matrix $\mtx{G}$ are fixed.
	We use the notation $\mtx{G}(i,:)$ to represent $i^{th}$ row of $\mtx{G}$.
	We assume that the number of data partitions
	$\ZNorm{\mtx{G}(i,:)}$ assigned to an arbitrary node $i$ is an odd number, to avoid ambiguity of the majority function at the encoder. 
	We further map the sign values $\{+1, -1\}$ to binary values $\{1, 0\}$ for ease of presentation. 
	Define
	\begin{align}
	S_v (\vct{m}) \hspace{-.7mm} = \hspace{-.7mm} \{ i\in [n]& \hspace{-.4mm}: \hspace{-.4mm} \vct{m}^T \mtx{G}(i,:) \hspace{-.4mm} \geq \hspace{-.4mm} v,  
	\ZNorm{\mtx{G}(i,:)} = 2v-1   \}, \label{Eqn:J_v}
	\end{align}
	which is the set of nodes having at least $v$ partitions with $m_j=1$, out of $2v - 1$ allocated data partitions. Note that
	we have $c_i =  \maj( \{m_j\}_{j \in P_i} ) = 1$ iff $i \in S_v(\vct{m})$ holds for some $v$.

	Before providing the explicit code construction rule (i.e., data allocation matrix $\mtx{G}$), we state the necessary and sufficient condition on $\mtx{G}$ to achieve \emph{perfect} $b$-Byzantine tolerance.
	
	\begin{lemma}\label{Lemma:NSCondition}
		Consider using a data allocation matrix $\mtx{G}$. The system is \textit{perfect} $b-$Byzantine tolerant if and only if  $	\sum_{v=1}^{\floor{n/2}} \lvert S_v (\vct{m})
		\rvert  \leq \floor{\frac{n}{2}}-b$ for all vectors $\vct{m} \in \{0,1\}^n$ having weight $\ZNorm{\vct{m}}= \floor{n/2}$.
	\end{lemma}
	\vspace{-3mm}
	\begin{proof}
		The formal proof is in Supplementary Material, and here we just give an intuitive sketch. 
		Recall that the majority opinion is $\mu=0$ when $\vct{m}$ has weight $\ZNorm{\vct{m}} = \floor{n/2}$. Moreover, in the worst case attacks from $b$ Byzantines, the output $y_i$ and the computational result $c_i$ of node $i$ satisfy $n_0 \coloneqq \lvert \{ i: y_i = 1 \}\rvert = \lvert \{ i: c_i = 1 \}\rvert + b$. Since the estimate on the majority opinion is $\hat{\mu} = \maj\{y_1, \cdots, y_n \}$,  the sufficient and necessary condition for perfect Byzantine tolerance (i.e., $\hat{\mu} = \mu$) is $n_0 \leq \floor{n/2}$, or equivalently, $\lvert \{ i: c_i = 1 \}\rvert = \sum_{v=1}^{\floor{n/2}} \lvert S_v (\vct{m})
		\rvert  \leq \floor{n/2}-b$. 
	\end{proof}
	
	\vspace{-2mm}
	Based on Lemma~\ref{Lemma:NSCondition}, we can construct explicit matrices $\mtx{G}$ 
	that guarantee perfect $b-$Byzantine tolerance, under arbitrary $n,b$ settings. 
	The detailed code construction rule is given in Algorithm \ref{Alg:DetCode}, and the structure of the suggested allocation matrix $\mtx{G}$ is depicted in Fig. \ref{Fig:Generator_matrix}. %
	The following theorem states the main property of our code, which is proven in
	Supplementary Material.
	
	\begin{theorem}\label{Thm:upper_bound}
		The deterministic code given in Algorithm \ref{Alg:DetCode} satisfies perfect $b-$Byzantine tolerance for $0 < b < \floor{n/2}$, by utilizing a computational redundancy of 
		\begin{align} \label{Eqn:upper_bound}
		r &= \frac{n+(2b+1)}{2} - \left(\left\lfloor \frac{n-(2b+1)}{2(b+1)} \right\rfloor + \frac{1}{2}\right) \frac{n-(2b+1)}{n}.
		\end{align}
	\end{theorem}

	\begin{remark}[Performance/convergence guarantee]
		The code in Algorithm~\ref{Alg:DetCode} satisfies $\hat{\mu} = \mu$ for any realization of stochastic gradient %
		obtained from $n$ data partitions. Thus, our scheme achieves the ideal performance of SignSGD-MV with no Byzantines, regardless of the actual number of Byzantines in the system.
		Moreover, the convergence of our scheme is guaranteed from Theorem 2 of~\cite{bernstein2018signsgd_ICLR}.
	\end{remark}

	\begin{figure}[t]
		\centering	
		\includegraphics[height=21mm]{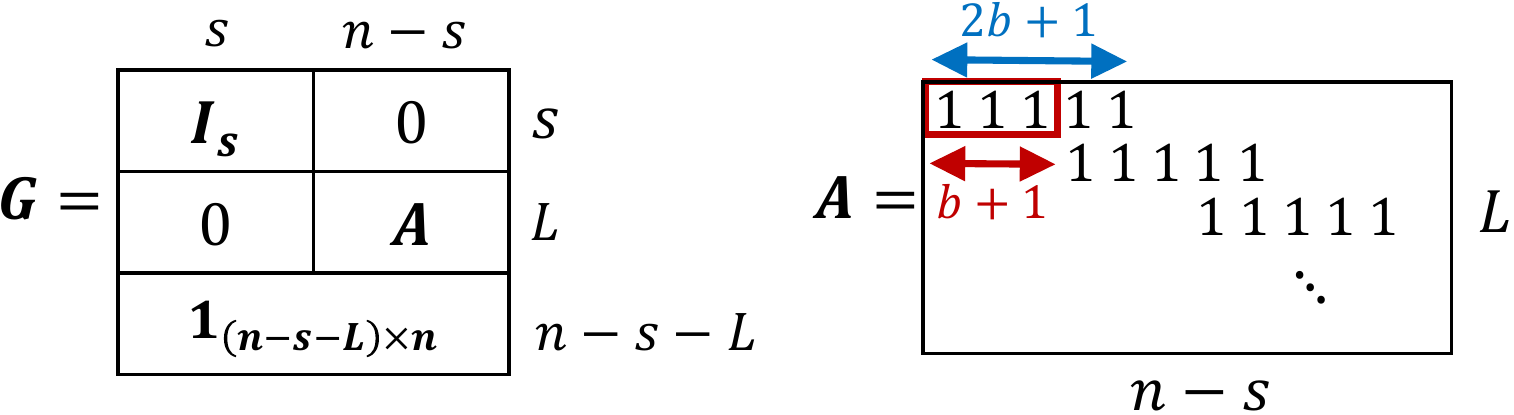}
		\caption{The structure of generator matrix $\mtx{G}$ devised in Algorithm \ref{Alg:DetCode}}
		\label{Fig:Generator_matrix}
		\vspace{-2mm}
	\end{figure}

	\vspace{-3mm}
	\section{Experiments on Amazon EC2}\label{Sec:Simul}

	Experimental results are obtained on Amazon EC2. 
	Considering a distributed learning setup, we used message passing interface  MPI4py\cite{DALCIN20111124}.%

	\paragraph{Compared schemes.} We compared the suggested coding schemes with the conventional uncoded scheme of \textsc{SignSGD with Majority Vote}. Similar to the simulation settings in the previous works \cite{pmlr-v80-bernstein18a, bernstein2018signsgd_ICLR}, we used the momentum counterpart $\textsc{Signum}$ instead of $\textsc{SignSGD}$ for fast convergence, and used a learning rate of $\gamma = 0.0001$ 
	and a momentum term of $\eta = 0.9$. 
	We simulated deterministic codes given in Algorithm \ref{Alg:DetCode}, and Bernoulli codes suggested in Section \ref{Sec:Bernoulli} with connection probability of $p$. Thus, the probabilistic code have expected computational redundancy of $\ept{r}=np$.

	\vspace{-3mm}
	\paragraph{Byzantine attack model.} We consider two attack models used in related works \cite{bernstein2018signsgd_ICLR, chen2018draco}: 1) the \textit{reverse} attack where a Byzantine node flips the sign of the gradient estimate, and 2) the \textit{directional} attack where a Byzantine node guides the model parameter in a certain direction. Here, we set the direction as an all-one vector.
	For each experiment, $b$ Byzantine nodes are selected arbitrarily.

	\vspace{-2mm}
	\subsection{Experiments on deep neural network models}

	\begin{figure}[!t]
		\vspace{-5mm}
		\centering
		\subfloat[][$n=5, b=1$ (epoch)] {\includegraphics[width=45mm]
			{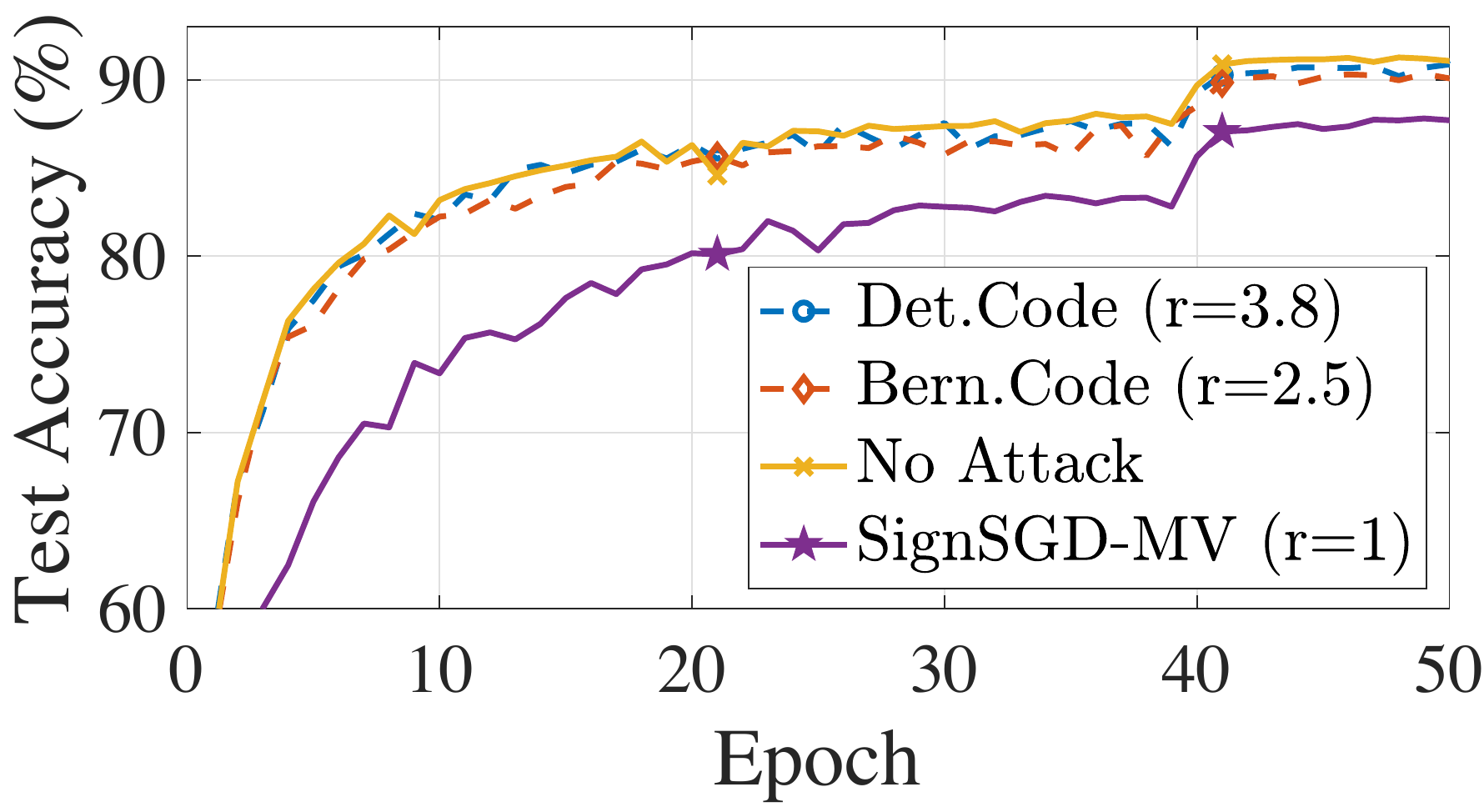}	\label{Fig:Simul_n5_b1_epoch}}
		\
		\subfloat[][$n=9, b=2$ (epoch) ]{\includegraphics[width=45mm]
			{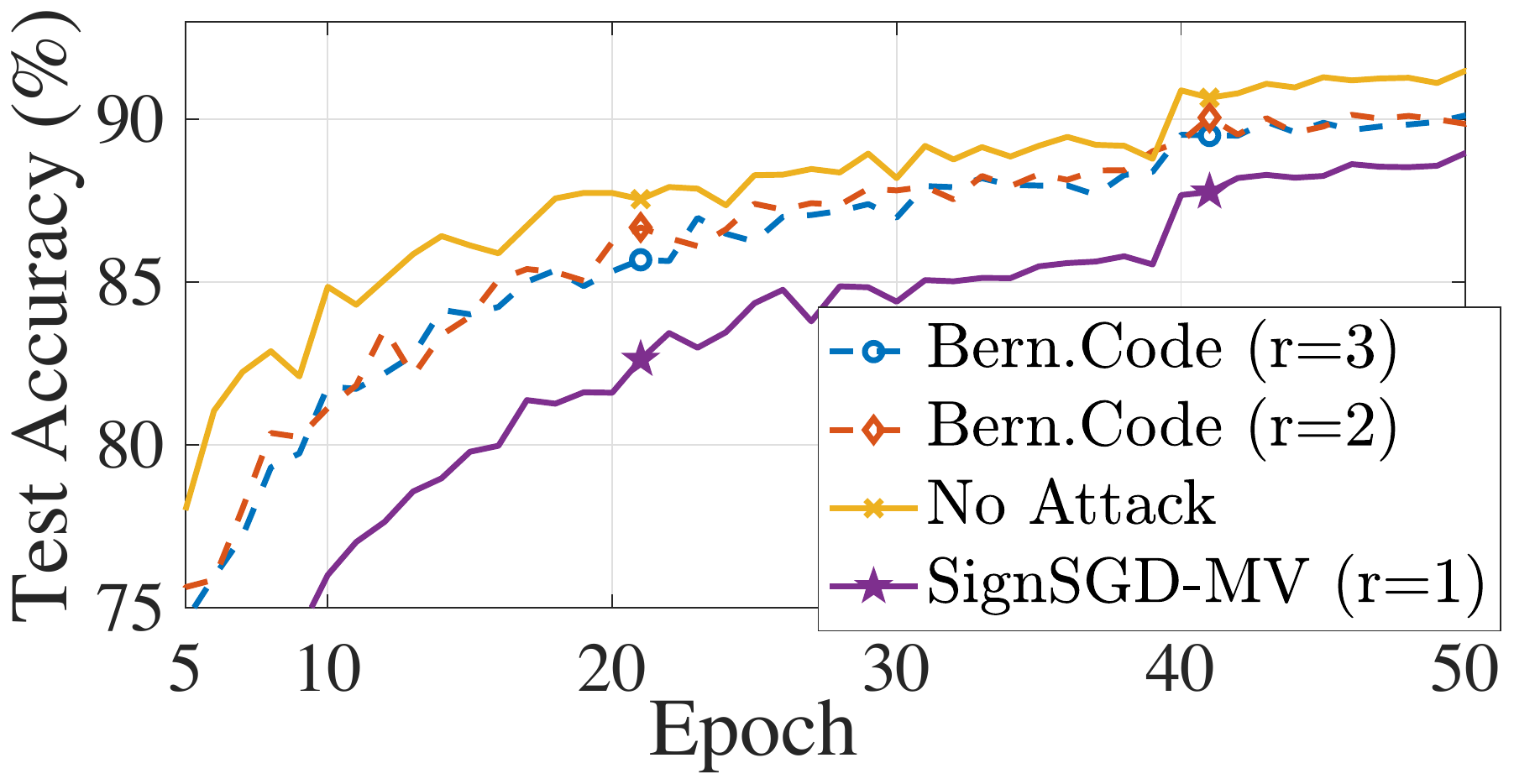}
			\label{Fig:Simul_n9_b2_epoch}}
		\
		\subfloat[][$n=15, b=3$ (epoch) ]{\includegraphics[width=45mm]
			{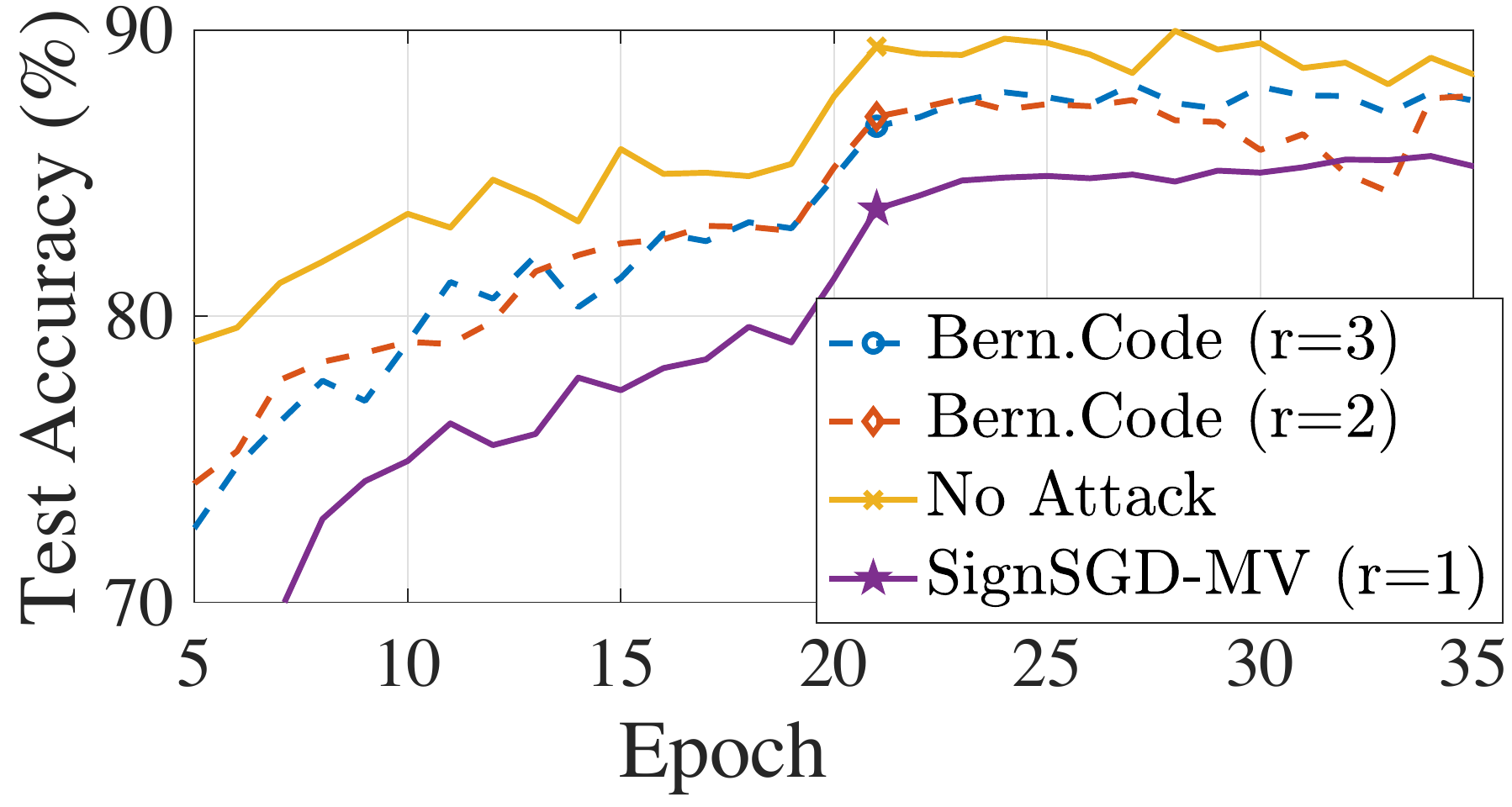}
			\label{Fig:Simul_n15_b3_epoch}}
		\vspace{-4mm}	
		\\	
		\subfloat[][$n=5, b=1$ (time)] {\includegraphics[width=45mm]
			{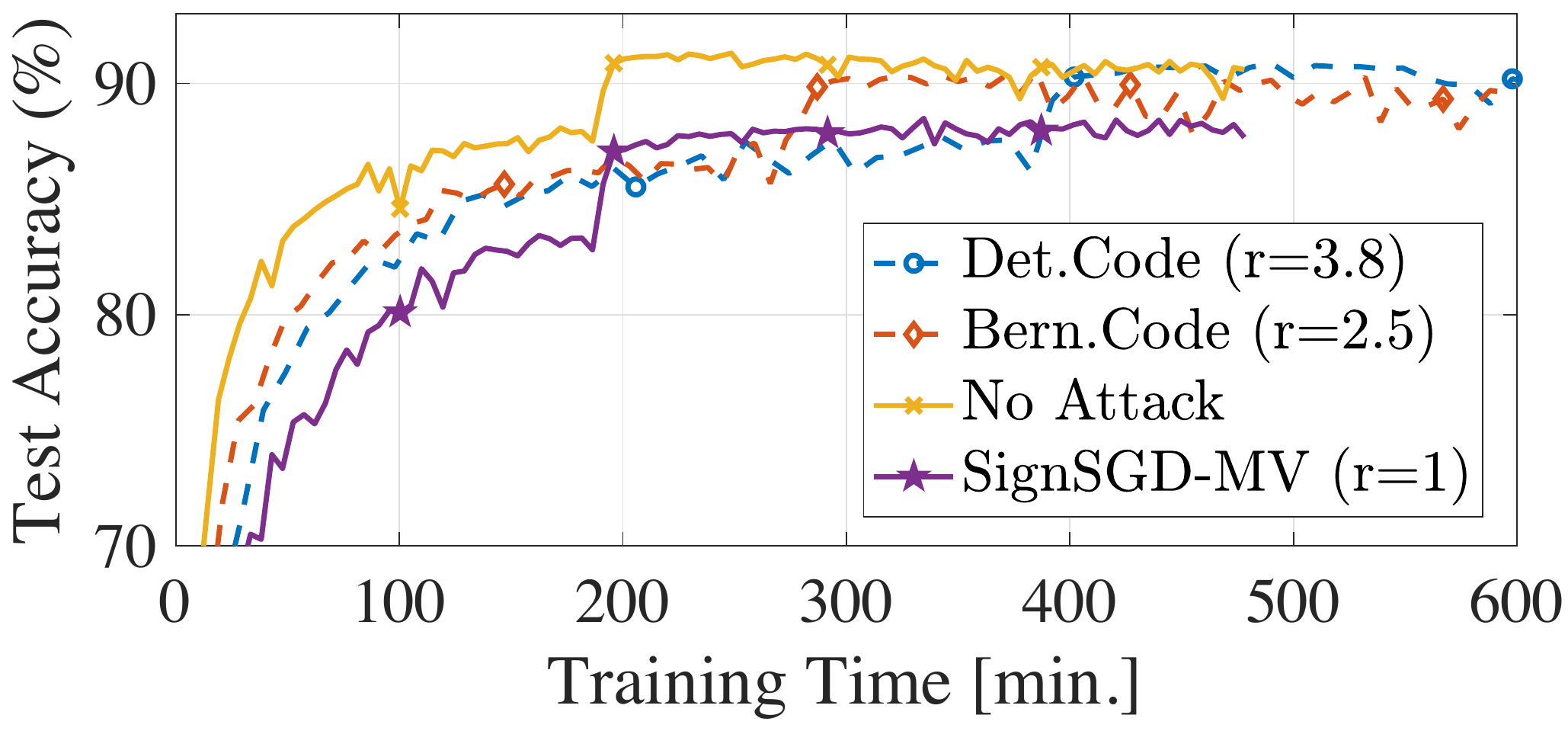}	\label{Fig:Simul_n5_b1_time}}
		\
		\subfloat[][$n=9, b=2$ (time)]{\includegraphics[width=45mm]
			{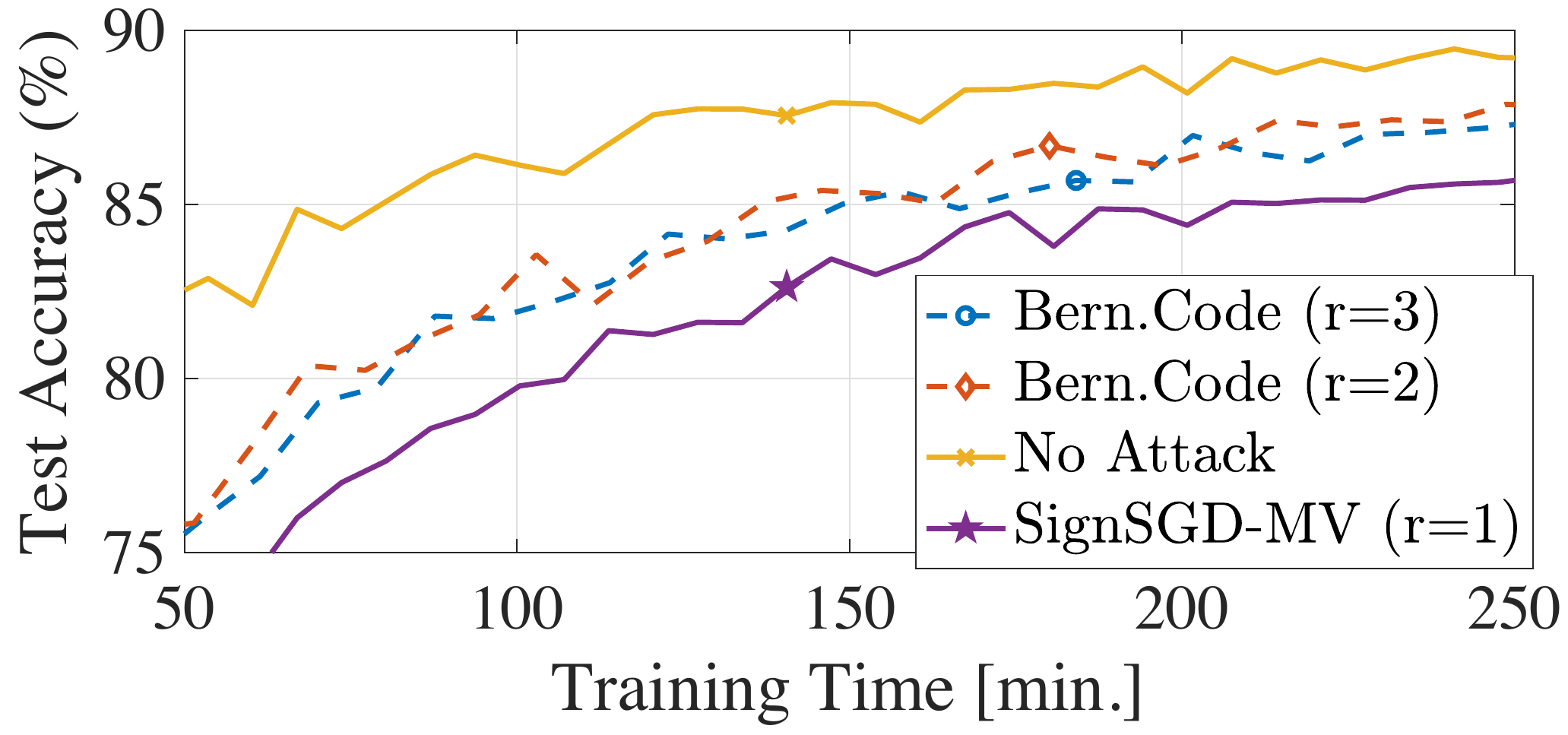}
			\label{Fig:Simul_n9_b2_time}} 
		\
		\subfloat[][$n=15, b=3$ (time)]{\includegraphics[width=45mm]{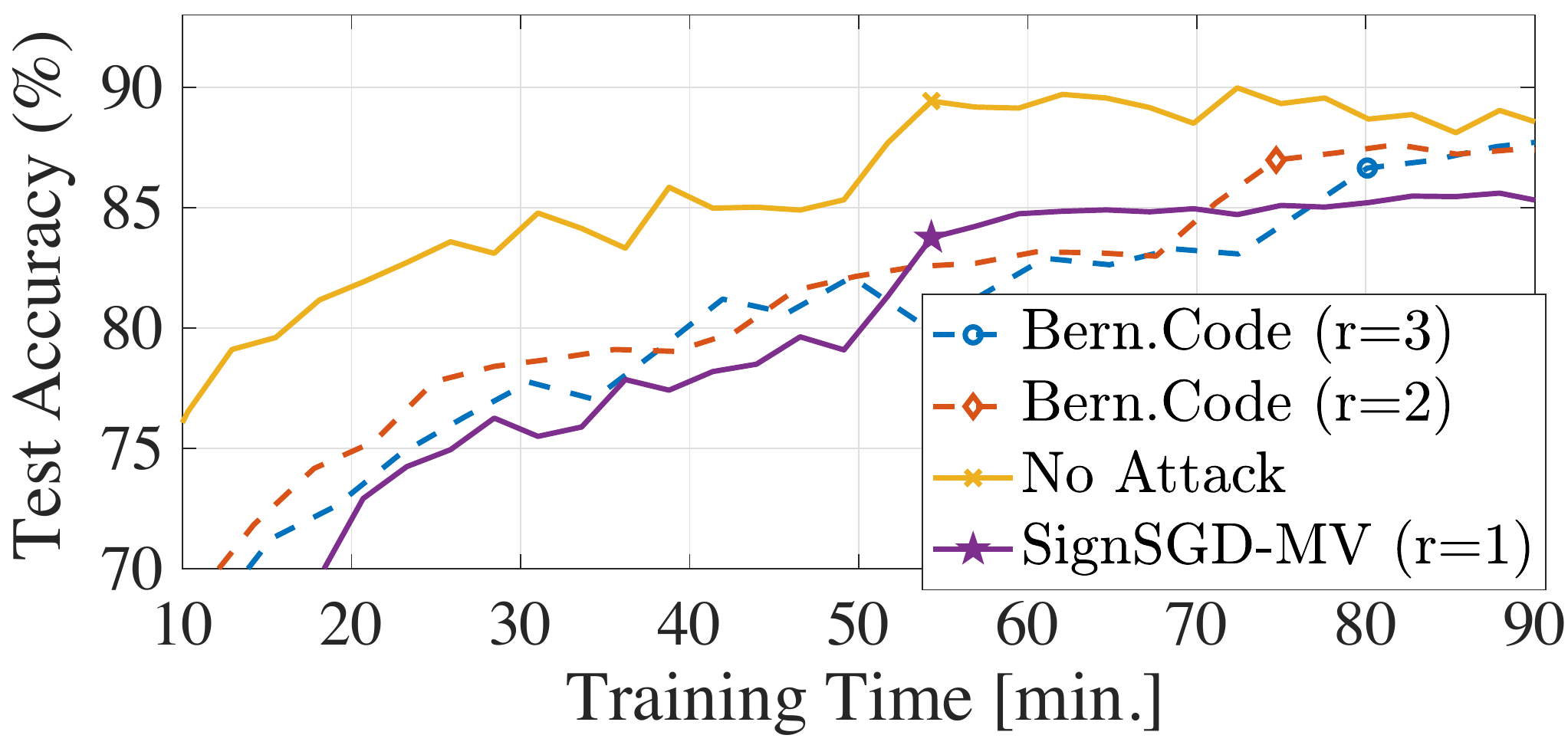}
			\label{Fig:Simul_n15_b3_time}} 
		\vspace{-2mm} 
		\caption{Simulation results with a small portion of Byzantines, under the \emph{reverse} attack on \textsc{Resnet-18} training \textsc{Cifar-10}. Deterministic/Bernoulli codes are abbreviated as "Det. code" and "Bern. code", respectively.}
		\label{Fig:Simul_Resnet}
		\vspace{-5mm}
	\end{figure}
	
	\vspace{-1mm}
	We trained a \textsc{Resnet-18} model on \textsc{Cifar-10} dataset. 
	Under this setting, the model dimension is $d=11,173,962$, and the number of training/test samples is set to $n_{\text{train}}=50000$ and $n_{\text{test}}=10000$, respectively. We used stochastic mini-batch gradient descent with batch size $B$, and our experiments are simulated on 
	\texttt{g4dn.xlarge} instances (having a GPU) for both workers and the master. 
	
	\vspace{-1mm}
	Fig. \ref{Fig:Simul_Resnet} illustrates the test accuracy performances of coded/uncoded schemes, when a small portion of nodes are adversaries. Each curve represents an average over $3$ independent train runs.
	Each column corresponds to different settings of $n$ and $b$: the first scenario is when $n=5, b=1$, the second one has $n=9, b=2$, and the last setup is $n=15,b=3$.
	For each scenario, two types of plots are given: one having
	horizontal axis of \textit{training epoch}, and the other with horizontal axis of \textit{training time}. We plotted the case with no attack as a reference to an ideal scenario.
	As in figures at the top row of Fig.~\ref{Fig:Simul_Resnet}, both deterministic and Bernoulli codes nearly achieve the ideal performance at each epoch, for all three scenarios. Overall, the suggested schemes enjoy $5-10 \%$ accuracy gains compared to the uncoded scheme in~\cite{bernstein2018signsgd_ICLR}. Moreover, the figures at the bottom row of Fig.~\ref{Fig:Simul_Resnet} show that the suggested schemes achieve a given level of test accuracy with less training time, compared to the uncoded scheme. 
	Interestingly, an expected redundancy as small as $\mathds{E}[r]=2, 3$ is enough to guarantee a meaningful performance gain compared to the conventional scheme~\cite{bernstein2018signsgd_ICLR} with $r=1$.
	
	\vspace{-1mm}
	Now, when the Byzantine portion of nodes is large, the results are even more telling. We plotted the simulation results in Fig. \ref{Fig:Simul_Resnet_mismatch}. Again, each curve reflects an average over $3$ independent runs. For various $n, b$ settings, it is clear that the uncoded scheme is highly vulnerable to Byzantine attacks, and the suggested codes with reasonable redundancy (as small as $\mathds{E}[r]=2, 3$) help maintaining high accuracy under the attack. When the model converges, the suggested scheme enjoys $15-30\%$ accuracy gains compared to the uncoded scheme; our codes provide remarkable training time reductions to achieve a given level of test accuracy.

	\begin{figure}[!t]
		\vspace{-3mm}
		\centering
		\subfloat[][$n=5, b=2$] {\includegraphics[width=44mm]
			{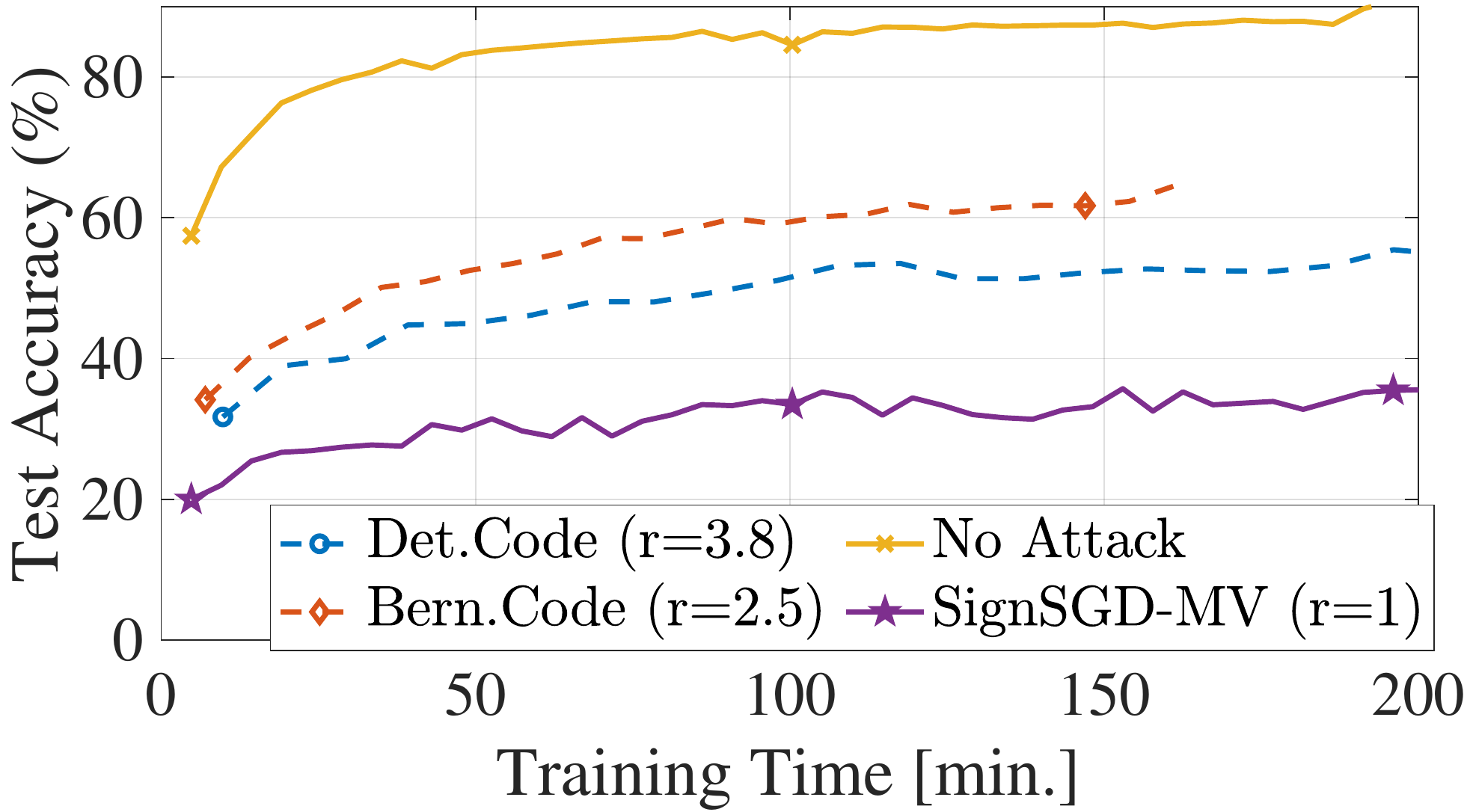}	\label{Fig:Simul_n5_b2_time}}
		\  
		\subfloat[][$n=9, b=3$]{\includegraphics[width=45mm]
			{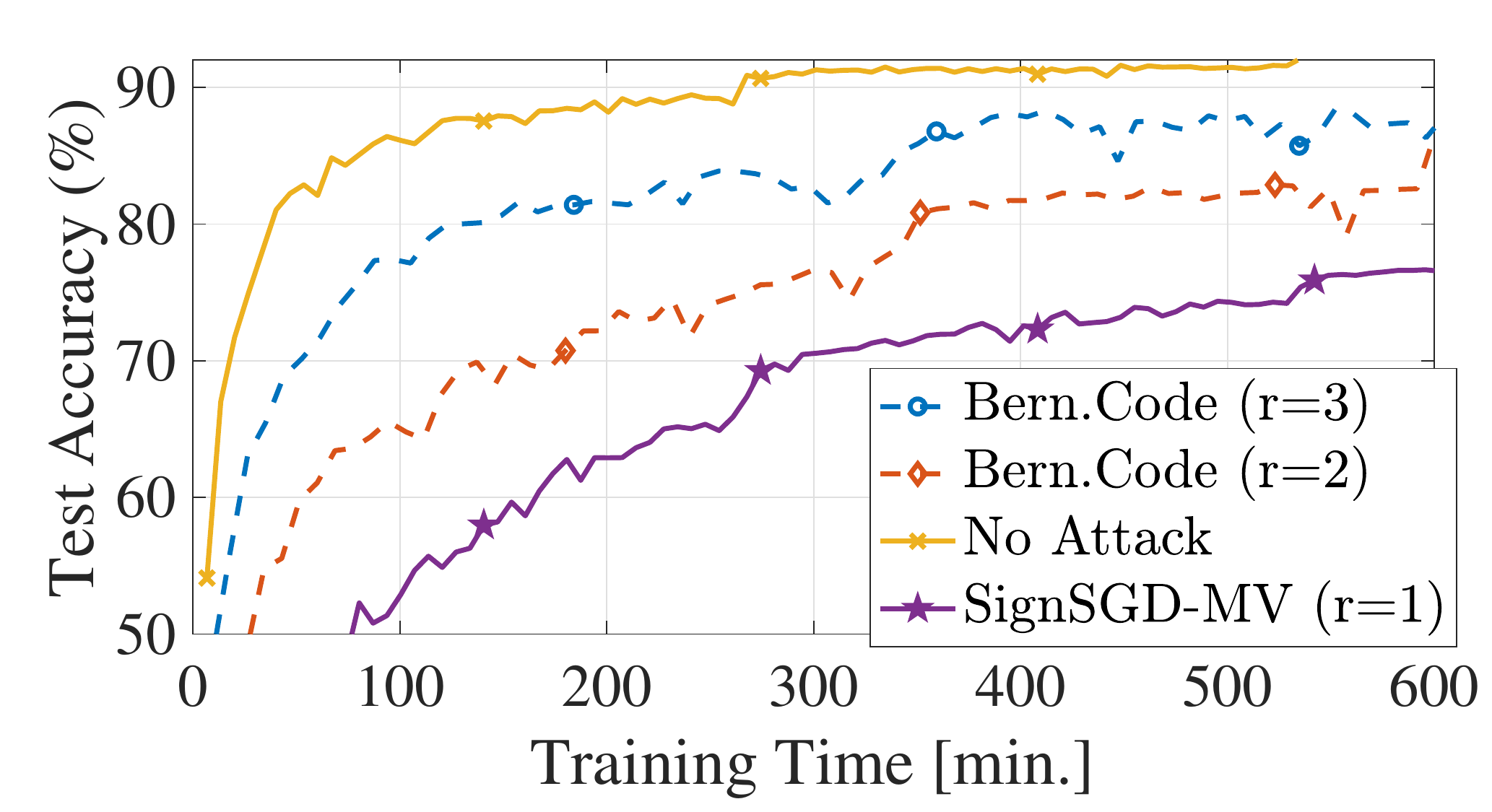}	\label{Fig:Simul_n9_b3_time}} 
		\ 
		\subfloat[][$n=15, b=6$ ]{\includegraphics[width=45mm]{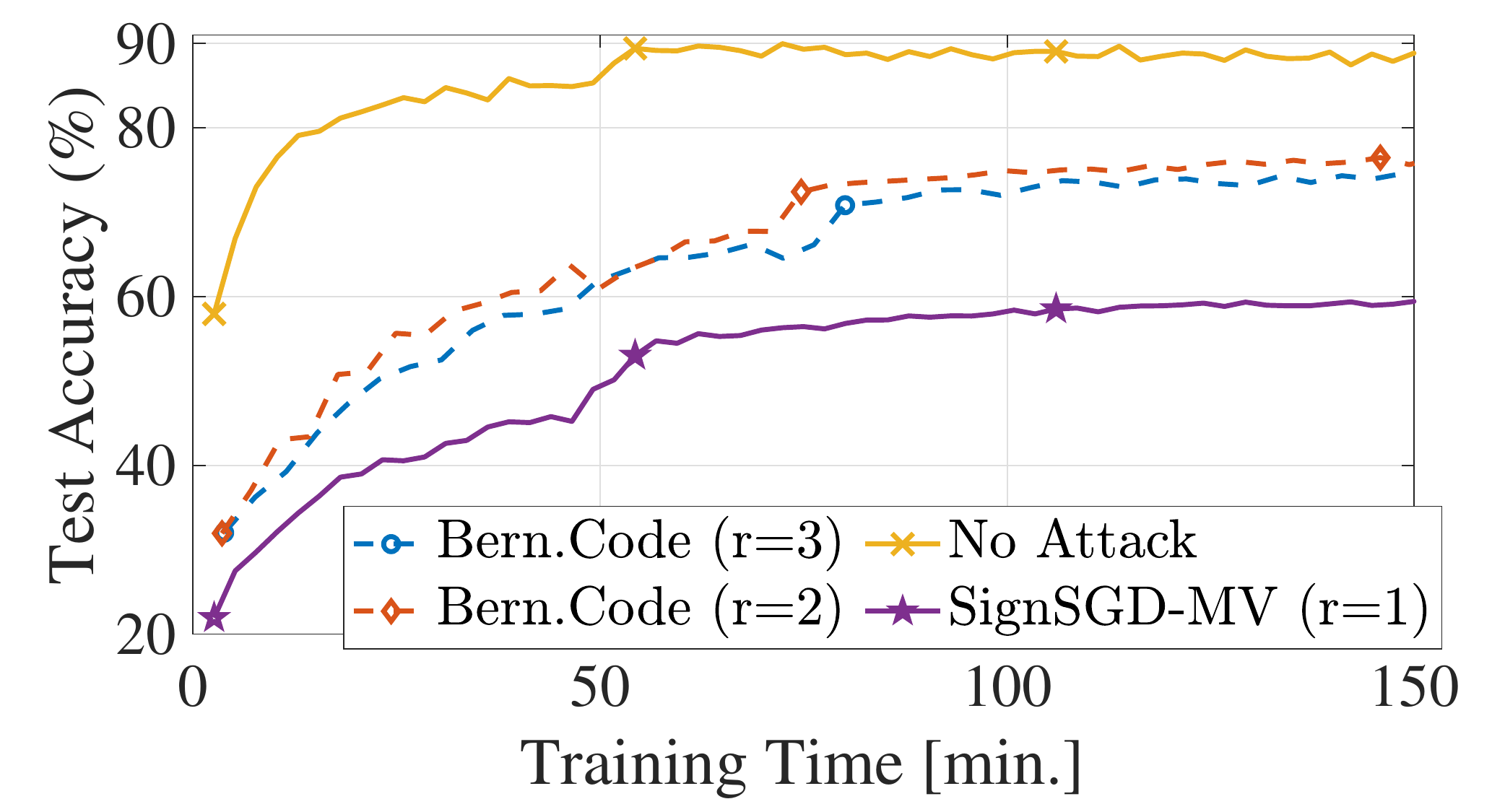}
			\label{Fig:Simul_n15_b6_time}} 
		\caption{Simulation results with a large portion of Byzantines, under the \emph{reverse} attack on \textsc{Resnet-18} training \textsc{Cifar-10}. Deterministic/Bernoulli codes are abbreviated as "Det. code" and "Bern. code", respectively.}
		\label{Fig:Simul_Resnet_mismatch}
		\vspace{-5mm}
	\end{figure}

	\vspace{-2mm}
	\subsection{Experiments on logistic regression models}
	\vspace{-1mm}
	We trained a logistic regression model on the Amazon Employee Access data set from Kaggle\footnote{https://www.kaggle.com/c/amazon-employee-access-challenge}, which is used in \cite{pmlr-v80-ye18a, tandon2017gradient} on coded gradient computation schemes.
	The model dimension is set to $d=263500$ after applying one-hot encoding with interaction terms.  
	We used \texttt{c4.large} instances for $n$ \textit{workers} that compute batch gradients, and a single \texttt{c4.2xlarge} instance for the \textit{master} that aggregates the gradients from workers and determines the model updating rule. 

	\begin{figure}[!t]
		\vspace{0mm}
		\centering
		\subfloat[][$n=15, b=5$]{\includegraphics[height=25mm]
			{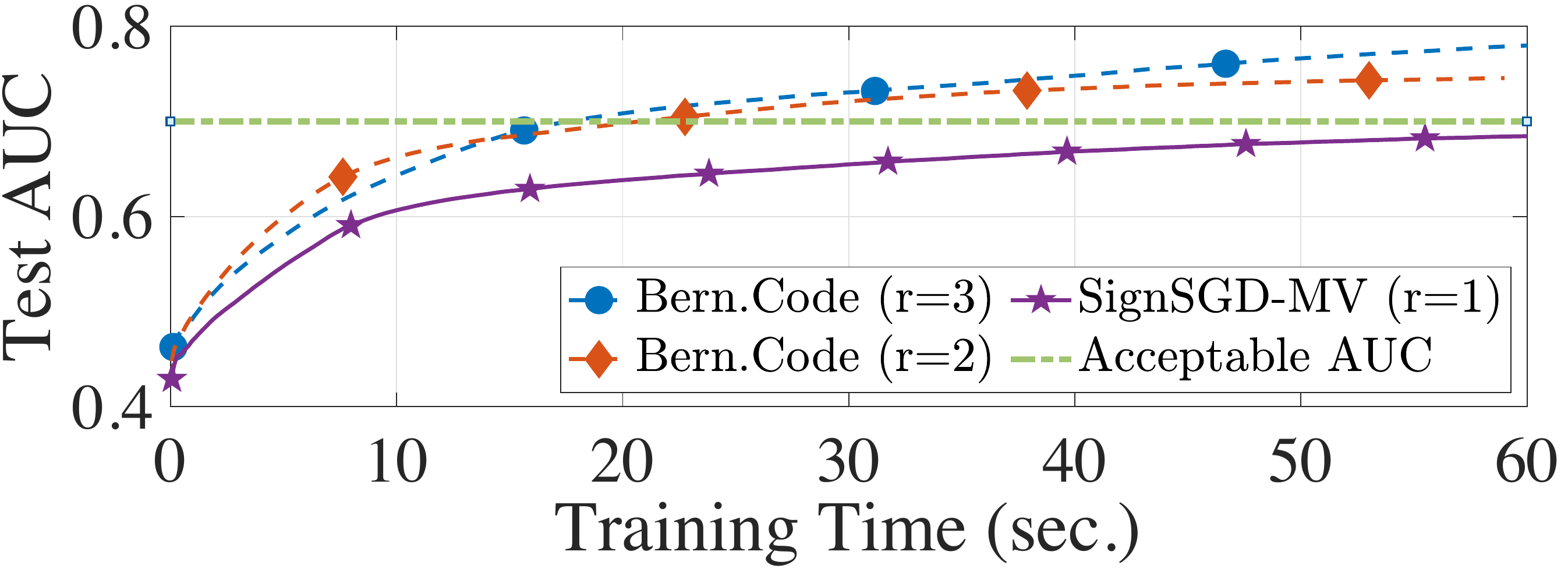}	\label{Fig:Simul_n15_b5_time}} 
		\subfloat[][$n=49, b=5$]{\includegraphics[height=25mm]
			{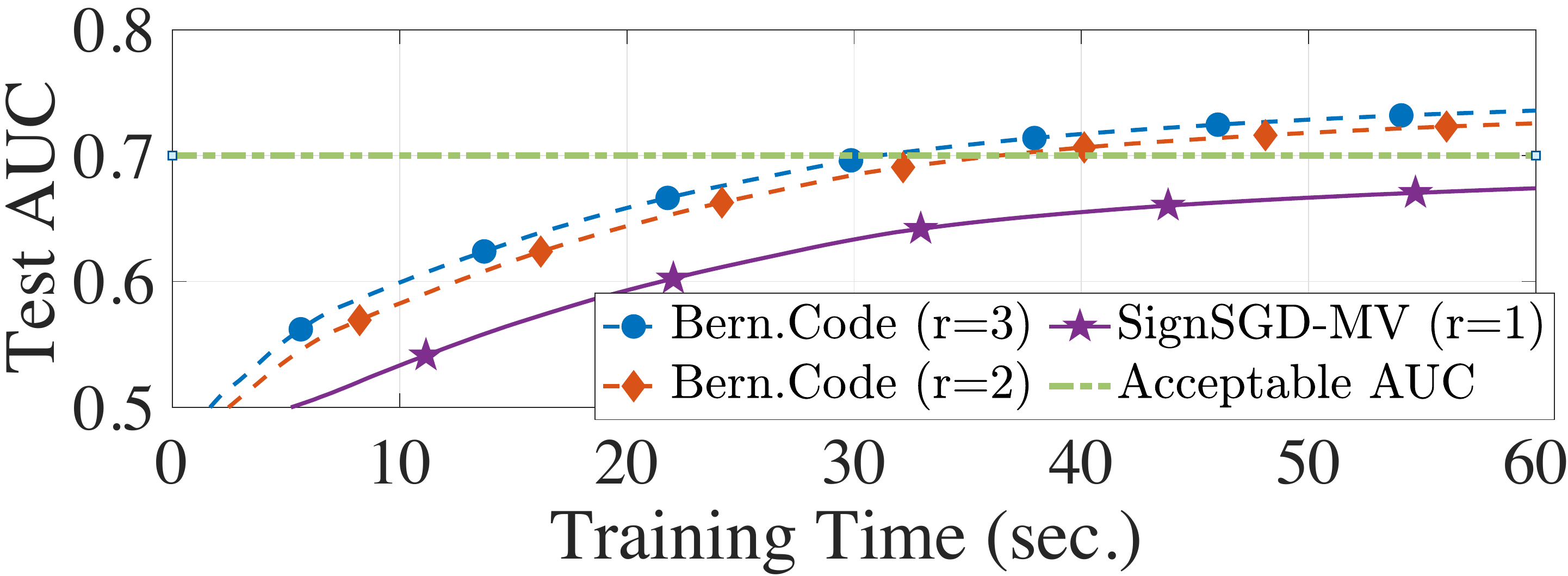}	\label{Fig:Simul_n49_b5_directional}}
		\caption{Test AUC performances of coded/uncoded schemes with various redundancy factors $r$, under \textit{directional} attack of $b$ Byzantine nodes. Random Bernoulli codes are abbreviated as "Bern. code".}
		\label{Fig:Simul_n15}
		\vspace{-5mm}
	\end{figure}
	
	\vspace{-1mm}
	Fig. \ref{Fig:Simul_n15} illustrates the generalization area under curve (AUC) performance of coded/uncoded schemes, under the \textit{directional} attack of Byzantine nodes.
	Fig. \ref{Fig:Simul_n15_b5_time} illustrates the performances when $n=15$ and $b=5$.
	Here we set the batch size $B=15$ and the number of training data $q=26325$. 
	Fig. \ref{Fig:Simul_n49_b5_directional} compares the performances of coded/uncoded schemes when $n=49$ and $b=5$.
	In this case, we set $B=5$ and $q=28665$. 
	In both scenarios, there exist clear performance gaps between the suggested Bernoulli codes and the conventional uncoded scheme.
	Moreover, our scheme with an expected redundancy as small as $\mathds{E}[r]=2$ gives a large training time reduction relative to the uncoded system in achieving a given level of accuracy (e.g. AUC=0.7, regarded as an acceptable performance).

	\vspace{-2mm}
	\section{Conclusions} %
	\vspace{-2mm}
	In this paper, we proposed \textsc{Election Coding}, a coding-theoretic framework that provides Byzantine tolerance of distributed learning 
	employing minimum worker-master communication.
	This framework tolerates arbitrary attacks corrupting the gradient computed in the training phase, by exploiting redundant gradient computations with appropriate allocation mapping between the individual workers and data  partitions.
	Making use of majority-rule-based encoding as well as decoding functions, we suggested two types of codes that tolerate Byzantine attacks with a controlled amount of redundancy, namely, random Bernoulli codes and deterministic codes. The Byzantine tolerance and the convergence of these coded distributed learning algorithms are 
	proved via mathematical analysis as well as through deep learning and logistic regression simulations on Amazon EC2.

	\section*{Broader Impact}
	Since our scheme uses minimum communication burden across distributed nodes, it is useful for numerous time-sensitive applications including smart traffic systems and anomaly detection in stock markets. 
	Moreover, our work is beneficial for various safety-critical applications that require the highest level of reliability and robustness, including autonomous driving, smart home systems, and healthcare services.
	
	In general, the failure in robustifying machine learning systems may cause some serious problems including traffic accidents or identity fraud. Fortunately, the robustness of the suggested scheme is mathematically proved, so that applying our scheme in safety-critical applications would not lead to such problems.

\newpage

\appendix
\renewcommand\thefigure{\thesection.\arabic{figure}}
\renewcommand\thetable{\thesection.\arabic{table}}
\setcounter{figure}{0}
\numberwithin{equation}{section}
\numberwithin{lemma}{section}
\numberwithin{theorem}{section}

\section{Additional experimental results}

\subsection{Speedup of \textsc{Election Coding} over SignSGD-MV}

We observe the speedup of \textsc{Election Coding} compared with the conventional uncoded scheme~\cite{bernstein2018signsgd_ICLR}. Table~\ref{Table:Speedup} summarizes the required training time to achieve the target test accuracy, for various $n,b,r$ settings, where $n$ is the number of nodes, $b$ is the number of Byzantines, and $r$ is the computational redundancy.
 As the portion of Byzantines $\frac{b}{n}$ increases, the suggested code with $r>1$ has a much higher gain in speedup, compared to the existing uncoded scheme with $r=1$. Note that this significant gain is achievable by applying codes with a reasonable redundancy of $r = 2,3$. 
\bmt{As an example, our Bernoulli code with mean redundancy of $r=2$ (two partitions per worker) achieves an 88\% accuracy at unit time of 609 while uncoded SignSGD-MV reaches a 79\% accuracy at a slower time of 750 and fails altogether to reach the 88\% level with $b=3$ out of $n=9$ workers compromised.}

\begin{table}[h]
	\caption{Training time (minutes) to reach the test accuracy for the suggested \textsc{Election Coding} with $r>1$ and the uncoded SignSGD-MV with $r=1$, for experiments using ResNet-18 to train CIFAR-10. Here, $\infty$ means that it is impossible for the uncoded scheme to reach the target test accuracy. For every target test accuracy, the suggested code with $r>1$ requires less training time than the conventional uncoded scheme with $r=1$.} 
	\centering
	\footnotesize
	\label{Table:Speedup}
	\vspace{2mm}
	\begin{tabular}{cccccccll}
		\cline{1-7}
		\multicolumn{3}{|c|}{Test accuracy}                                          & \multicolumn{1}{c|}{75\%}  & \multicolumn{1}{c|}{80\%}  & \multicolumn{1}{c|}{85\%}  & \multicolumn{1}{c|}{89\%}                  &                      &                      \\ \cline{1-7} &
		&     \vspace{-2mm}                       &                            &                            &                            &                                            &                      &                      \\ \cline{1-7}
		\multicolumn{1}{|c|}{\multirow{3}{*}{n=5, b=1}} & \multicolumn{1}{c|}{\multirow{2}{*}{\begin{tabular}[c]{@{}c@{}}Suggested\\ Election Coding\end{tabular}}} &  \multicolumn{1}{c|}{r=3.8} & \multicolumn{1}{c|}{\textbf{39.2}} & \multicolumn{1}{c|}{\textbf{68.6}} & \multicolumn{1}{c|}{\textbf{137.2}} & \multicolumn{1}{c|}{\textbf{392.1}}                 &                      &                      \\ \cline{3-7}
		\multicolumn{1}{|c|}{}                        & \multicolumn{1}{c|}{}    & \multicolumn{1}{c|}{r=2.5} & \multicolumn{1}{c|}{\textbf{28.0}}    & \multicolumn{1}{c|}{\textbf{56.0}}    & \multicolumn{1}{c|}{\textbf{119.0}}   & \multicolumn{1}{c|}{\textbf{287.0}}                   &                      &                      \\ \cline{2-7}
		\multicolumn{1}{|c|}{}                   & \multicolumn{1}{c|}{SignSGD-MV}         & \multicolumn{1}{c|}{r=1}   & \multicolumn{1}{c|}{\textbf{52.6}} & \multicolumn{1}{c|}{\textbf{95.6}} & \multicolumn{1}{c|}{\textbf{191.2}} & \multicolumn{1}{c|}{\bm{$\infty$}} &                      &                      \\ \cline{1-7}
		&
		&     \vspace{-2mm}                       &                            &                            &                            &                                            &                      &                      \\ \cline{1-7}
		\multicolumn{1}{|c|}{\multirow{3}{*}{n=9, b=2}} & \multicolumn{1}{c|}{\multirow{2}{*}{\begin{tabular}[c]{@{}c@{}}Suggested\\ Election  Coding\end{tabular}}} & \multicolumn{1}{c|}{r=3}  & \multicolumn{1}{c|}{\textbf{52.6}} & \multicolumn{1}{c|}{\textbf{87.7}} & \multicolumn{1}{c|}{\textbf{149.0}}   & \multicolumn{1}{c|}{\textbf{350.6}}                 &                      &                      \\ \cline{3-7}
	\multicolumn{1}{|c|}{}                          & \multicolumn{1}{c|}{}                                                                                       & \multicolumn{1}{c|}{r=2}   & \multicolumn{1}{c|}{\textbf{51.4}} & \multicolumn{1}{c|}{\textbf{68.6}} & \multicolumn{1}{c|}{\textbf{137.2}} & \multicolumn{1}{c|}{\textbf{334.4}}                 & \multicolumn{1}{c}{} &                      \\ \cline{2-7}
			\multicolumn{1}{|c|}{}                          & \multicolumn{1}{c|}{SignSGD-MV}                                                                             & \multicolumn{1}{c|}{r=1}   & \multicolumn{1}{c|}{\textbf{67.0}} & \multicolumn{1}{c|}{\textbf{113.8}} & \multicolumn{1}{c|}{\textbf{207.5}} & \multicolumn{1}{c|}{\textbf{381.5}}                   & \multicolumn{1}{c}{} & \multicolumn{1}{c}{} \\ \cline{1-7}
		\multicolumn{1}{l}{}                            &                            &                            &                            &                            &                                            & \multicolumn{1}{c}{} & \multicolumn{1}{c}{}
	\end{tabular}
\begin{tabular}{cccccccll}
	\cline{1-7}
	\multicolumn{3}{|c|}{Test accuracy}                                          & \multicolumn{1}{c|}{30\%} & \multicolumn{1}{c|}{60\%}                  & \multicolumn{1}{c|}{79\%}                  & \multicolumn{1}{c|}{88\%}                  &                      &                      \\ \cline{1-7}
	&          \vspace{-2mm}                  &                           &                                            &                                            &                                            &                      &                      \\ \cline{1-7}
	\multicolumn{1}{|c|}{\multirow{3}{*}{n=5, b=2}} & \multicolumn{1}{c|}{\multirow{2}{*}{\begin{tabular}[c]{@{}c@{}}Suggested\\ Election Coding\end{tabular}}} & \multicolumn{1}{c|}{r=3.8} & \multicolumn{1}{c|}{\textbf{9.8}}  & \multicolumn{1}{c|}{\textbf{19.6}}                  & \multicolumn{1}{c|}{\textbf{58.8}}                  & \multicolumn{1}{c|}{\textbf{245.0}}                 &                      &                      \\ \cline{3-7}
	\multicolumn{1}{|c|}{}                          & \multicolumn{1}{c|}{}    & \multicolumn{1}{c|}{r=2.5}  & \multicolumn{1}{c|}{\textbf{7.0}}  & \multicolumn{1}{c|}{\textbf{14.0}}                  & \multicolumn{1}{c|}{\textbf{56.0}}                  & \multicolumn{1}{c|}{\textbf{245.0}}                 &                      &                      \\ \cline{2-7}
	\multicolumn{1}{|c|}{}                          & \multicolumn{1}{c|}{SignSGD-MV}         & \multicolumn{1}{c|}{r=1}  & \multicolumn{1}{c|}{\textbf{43.0}} & \multicolumn{1}{c|}{\bm{$\infty$}} & \multicolumn{1}{c|}{\bm{$\infty$}} & \multicolumn{1}{c|}{\bm{$\infty$}} &                      &                      \\ \cline{1-7}
	&
	&     \vspace{-2mm}                       &                            &                            &                            &                                            &                      &                      \\ \cline{1-7}
	\multicolumn{1}{|c|}{\multirow{3}{*}{n=9, b=3}} & \multicolumn{1}{c|}{\multirow{2}{*}{\begin{tabular}[c]{@{}c@{}}Suggested\\ Election  Coding\end{tabular}}} & \multicolumn{1}{c|}{r=3}  & \multicolumn{1}{c|}{\textbf{8.8}}  & \multicolumn{1}{c|}{\textbf{26.3}}                  & \multicolumn{1}{c|}{\textbf{122.7}}                 & \multicolumn{1}{c|}{\textbf{394.4}}                 &                      &                      \\ \cline{3-7}
	\multicolumn{1}{|c|}{}                          & \multicolumn{1}{c|}{}                                                                                       & \multicolumn{1}{c|}{r=2}   & \multicolumn{1}{c|}{\textbf{8.6}}  & \multicolumn{1}{c|}{\textbf{60.0}}                  & \multicolumn{1}{c|}{\textbf{351.5}}                 & \multicolumn{1}{c|}{\textbf{608.7}}                 & \multicolumn{1}{c}{} &                      \\ \cline{2-7}
	\multicolumn{1}{|c|}{}                           & \multicolumn{1}{c|}{SignSGD-MV}                                                                             & \multicolumn{1}{c|}{r=1}      & \multicolumn{1}{c|}{\textbf{13.4}} & \multicolumn{1}{c|}{\textbf{167.3}}                 & \multicolumn{1}{c|}{\textbf{749.5}}                 & \multicolumn{1}{c|}{\bm{$\infty$}} & \multicolumn{1}{c}{} & \multicolumn{1}{c}{} \\ \cline{1-7}
	\multicolumn{1}{l}{}                            &                            &                           &                                            &                                            &                                            & \multicolumn{1}{c}{} & \multicolumn{1}{c}{}
\end{tabular}
\end{table}

\subsection{Performances for Byzantine mismatch scenario}

Suppose the deterministic code suggested in this paper is designed for $\hat{b}<b$ Byzantine nodes, where $b$ is the actual number of Byzantines in the system. 
When the number of nodes is $n=5$ and the number of Byzantines is $b=2$, Fig.~\ref{Fig:Simul_Resnet_Byz_mismatch} shows the performance of the deteministic code (with redundancy $r=3.8$) designed for $\hat{b} = 1$. 
Even in this underestimated Byzantine setup, the suggested code maintains its tolerance to Byzantine attacks, and the performance gap between the suggested code and the conventional SignSGD-MV is over $20 \%$.

\subsection{Performances for extreme Byzantine attack scenario}

In Fig.~\ref{Fig:Simul_Resnet_max_Byz}, we compare the performances of \textsc{Election Coding} and the conventional \textsc{SignSGD-MV}, under the maximum number of Byzantine nodes, i.e., $b = (n-1)/2$, when $n= 9$ and $b=4$. The suggested Bernoulli code enjoy over $20\%$ performance gap compared to the conventional SignSGD-MV. This shows that the suggested \textsc{Election Coding} is highly robust to Byzantine attacks even under the maximum Byzantine setup, while the conventional SignSGD-MV is vulnerable to the extreme Byzantine attack scenario.

\begin{figure}[]
	\vspace{-3mm}
	\centering
	\includegraphics[height=25mm]
		{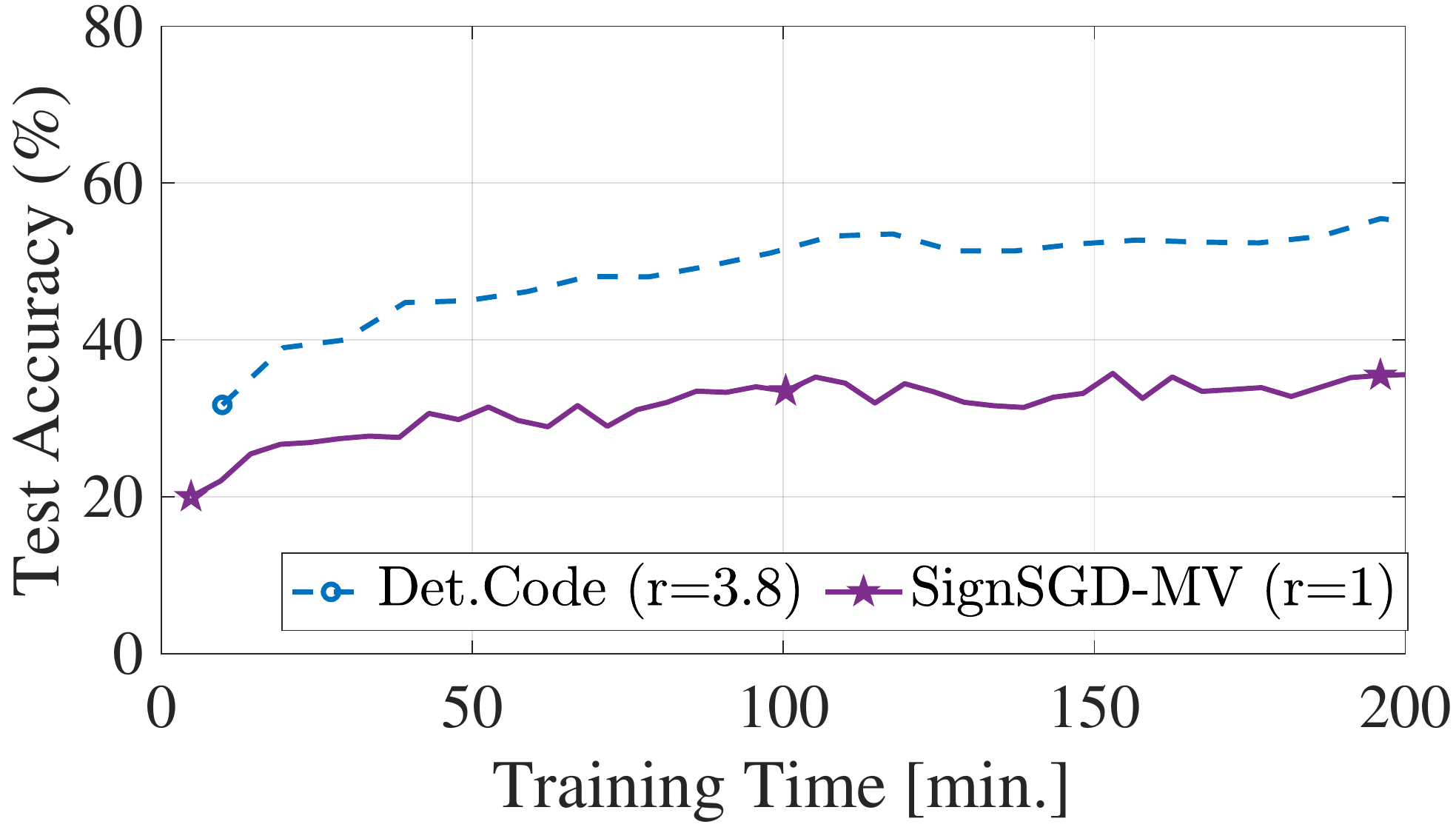}	\label{Fig:Simul_n5_b2_time_mismatch}
	\caption{Result for Byzantine mismatch scenario, i.e., $\hat{b} \neq b$, when $n=5$ and $b=2$. The suggested deterministic code (with $r=3.8$) is designed for $\hat{b}=1$, which is smaller than the actual number of Byzantines $b=2$. Even in this underestimated Byzantine setup, the suggested code is highly tolerant compared to the conventional SignSGD-MV. }
	\label{Fig:Simul_Resnet_Byz_mismatch}
	\vspace{-5mm}
\end{figure}

\begin{figure}[]
	\vspace{-3mm}
	\centering
		\includegraphics[height=25mm]
		{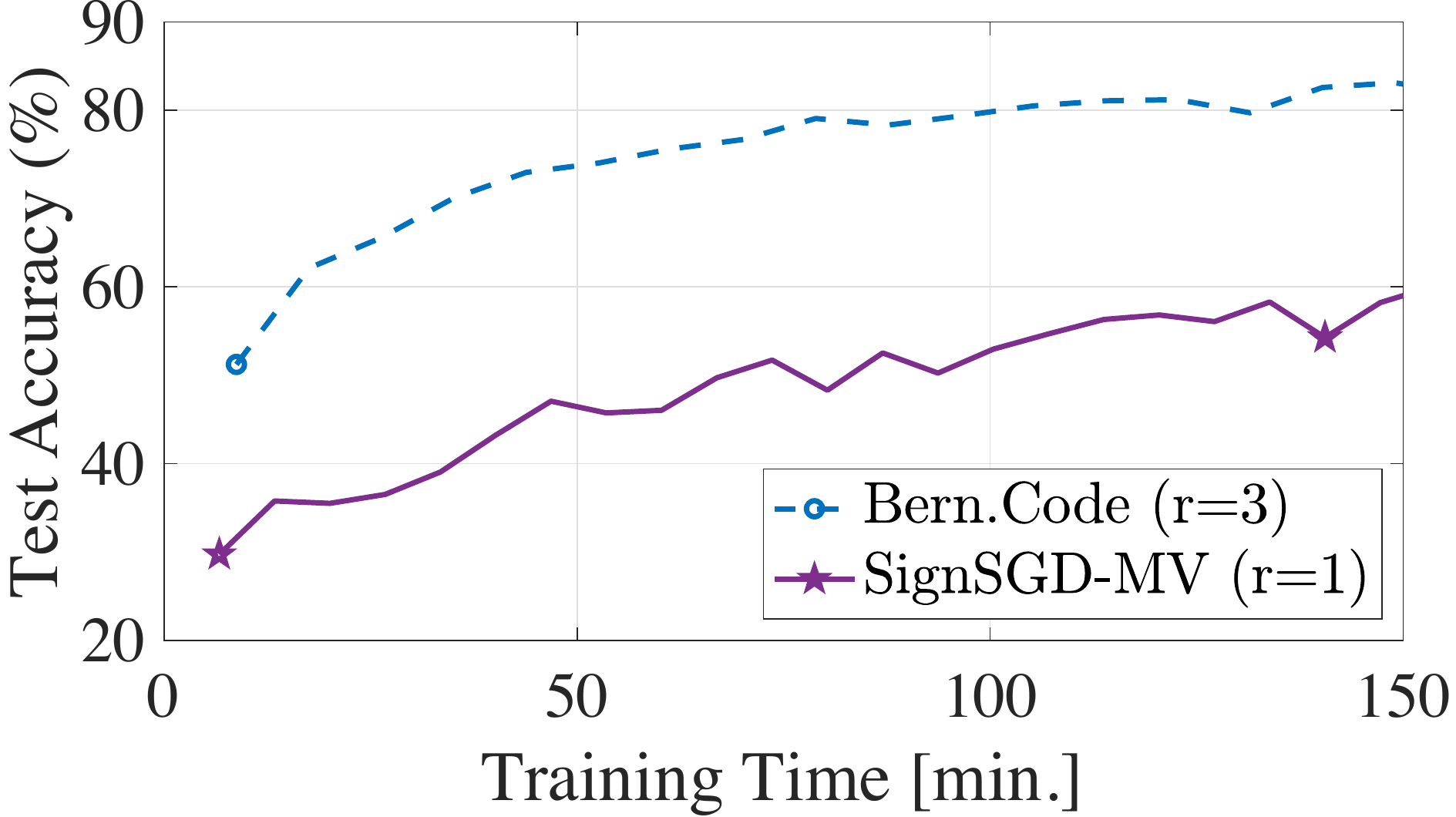}	%
	\caption{Results for maximum Byzantine scenario, i.e., $b= (n-1)/2$, when $n=9$ and $b=4$, under the \emph{reverse} attack on \textsc{Resnet-18} training \textsc{Cifar-10}. Bernoulli codes are abbreviated as "Bern. code".}
	\label{Fig:Simul_Resnet_max_Byz}
	\vspace{-5mm}
\end{figure}

\begin{figure}[!t]
	\vspace{-3mm}
	\centering
	\subfloat[][$n=5, b=2$] {\includegraphics[height=24mm]
		{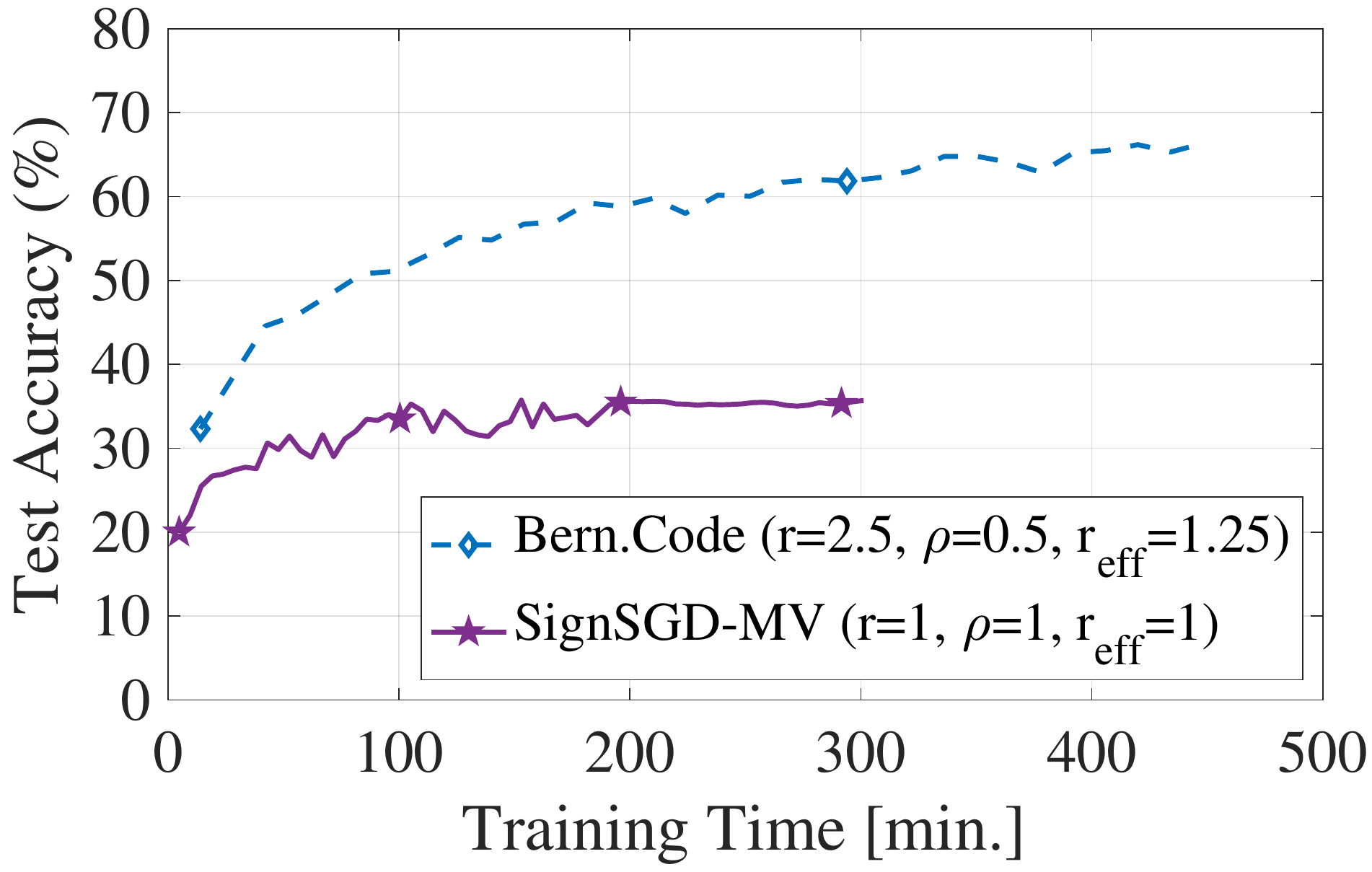}	\label{Fig:Simul_n5_b2_time_r_eff}}
	\  
	\subfloat[][$n=15, b=6$ ]{\includegraphics[height=24mm]{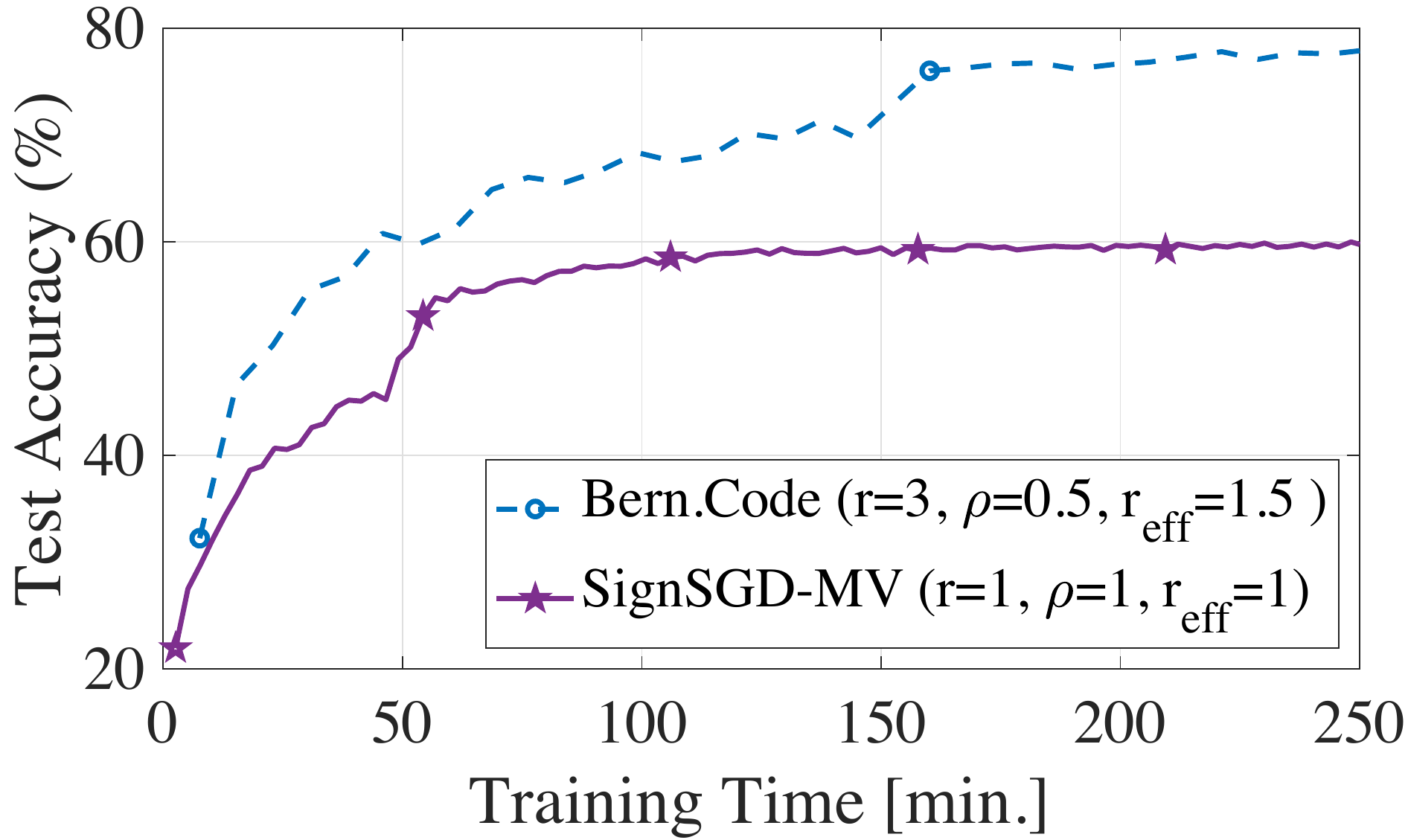}
		\label{Fig:Simul_n15_b6_time_r_eff}} 
	\caption{Impact of the reduced effective redundancy $r_{\textrm{eff}}$, under the \emph{reverse} attack on \textsc{Resnet-18} training \textsc{Cifar-10}. The suggested Bernoulli codes with effective redundancy $r_{\textrm{eff}}= 1.25, 1.5$ has a high performance gain compared to the conventional SignSGD-MV.}
	\label{Fig:Simul_Resnet_r_eff}
	\vspace{-5mm}
\end{figure}

\subsection{Reducing the amount of redundant computation}

We remark that there is a simple way to lower the computational burden of \textsc{Election Coding} by reducing the mini-batch size to $\rho B$ for some reduction factor $\rho < 1$. In this case, the effective computational redundancy can be represented as $r_{\textrm{eff}} = \rho r$. Fig.~\ref{Fig:Simul_Resnet_r_eff} shows the performance of \textsc{Election Coding} with reduced effective redundancy. For experiments on $n=5$ or $n=15$, we tested the Bernoulli code with redundancy $r = 2.5$ and $r=3$, respectively, and reduced the batch size by the ratio $\rho = 0.5$. This results in the effective redundancy of $r_{\textrm{eff}} = 1.25$ and $r_{\textrm{eff}} = 1.5$, respectively, as in Fig.~\ref{Fig:Simul_Resnet_r_eff} . Comparing with the SignSGD with effective redundancy $r_{\textrm{eff}}=1$, the suggested codes have $20\sim30 \%$ performance gap. This shows that \textsc{Election Coding} can be used as a highly practical tool for tolerating Byzantines, with effective redundancy well below 2.

\subsection{Performances for larger networks / noisy computations}
\bmt{We ran experiments for a larger network, ResNet-50, as seen in Fig.~\ref{Fig:ResNet50}. Our scheme with redundancy $r_{\text{eff}}=1.5$ is still effective here.
We also considered the effect of noisy gradient computation on the performance of the suggested scheme. We added independent Gaussian noise to all gradients corresponding to individual partitions before the signs are taken (in addition to Byzantine attacks on local majority signs). The proposed Bernoulli code can tolerate noisy gradient computation well, while uncoded signSGD-MV cannot, as shown in Fig.~\ref{Fig:noisy_computation} for added noise with variance $\sigma^2 = 1e^{-8}$.}

\subsection{Comparison with median-based approach}
\bmt{We compared our scheme with full gradient + median (FGM). As in Fig.~\ref{Fig:geometric_median}, our scheme with redundancy as low as $r_{\text{eff}}=1.5$ outperforms FGM. While FGM requires no computational redundancy, it needs $32$x more communication bandwidth than ours while performing worse. For the 20\% Byzantine scenario ($b=3$), FGM performs relatively better, but still falls well below ours. For the severe 40\% Byzantine case ($b=6$), it is harder to achieve near-perfect protection but our schemes clearly perform  better.}

\begin{figure}
	\vspace{-5mm}
	\centering
	\subfloat[][ResNet-50]{\includegraphics[height=27mm]
		{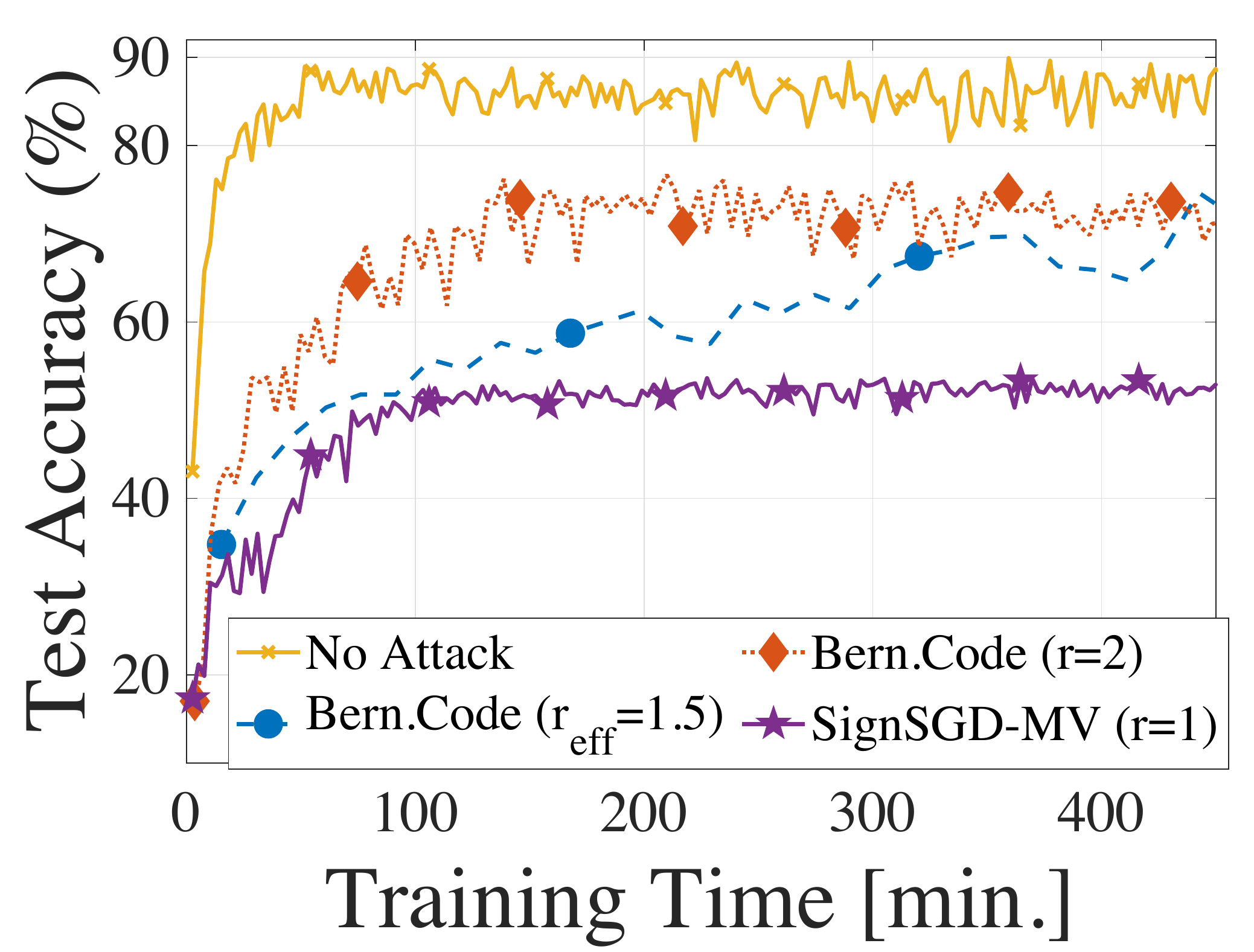}	\label{Fig:ResNet50}} 
	\quad
	\subfloat[][\centering{Noisy computation}]{\includegraphics[height=27mm]
		{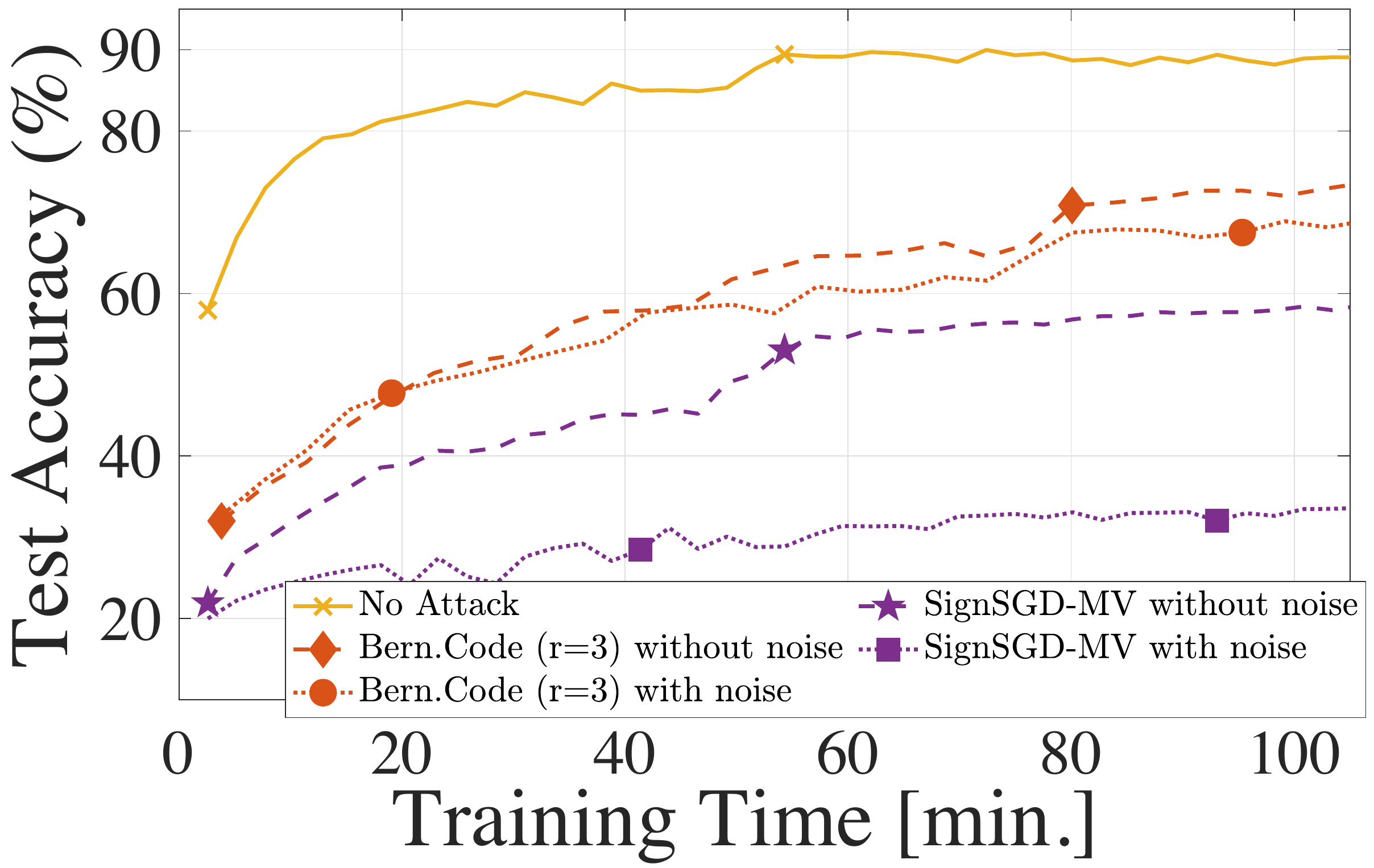}	\label{Fig:noisy_computation}}
	\quad
	\subfloat[][\centering{Compare with median}]{\includegraphics[height=27mm]
		{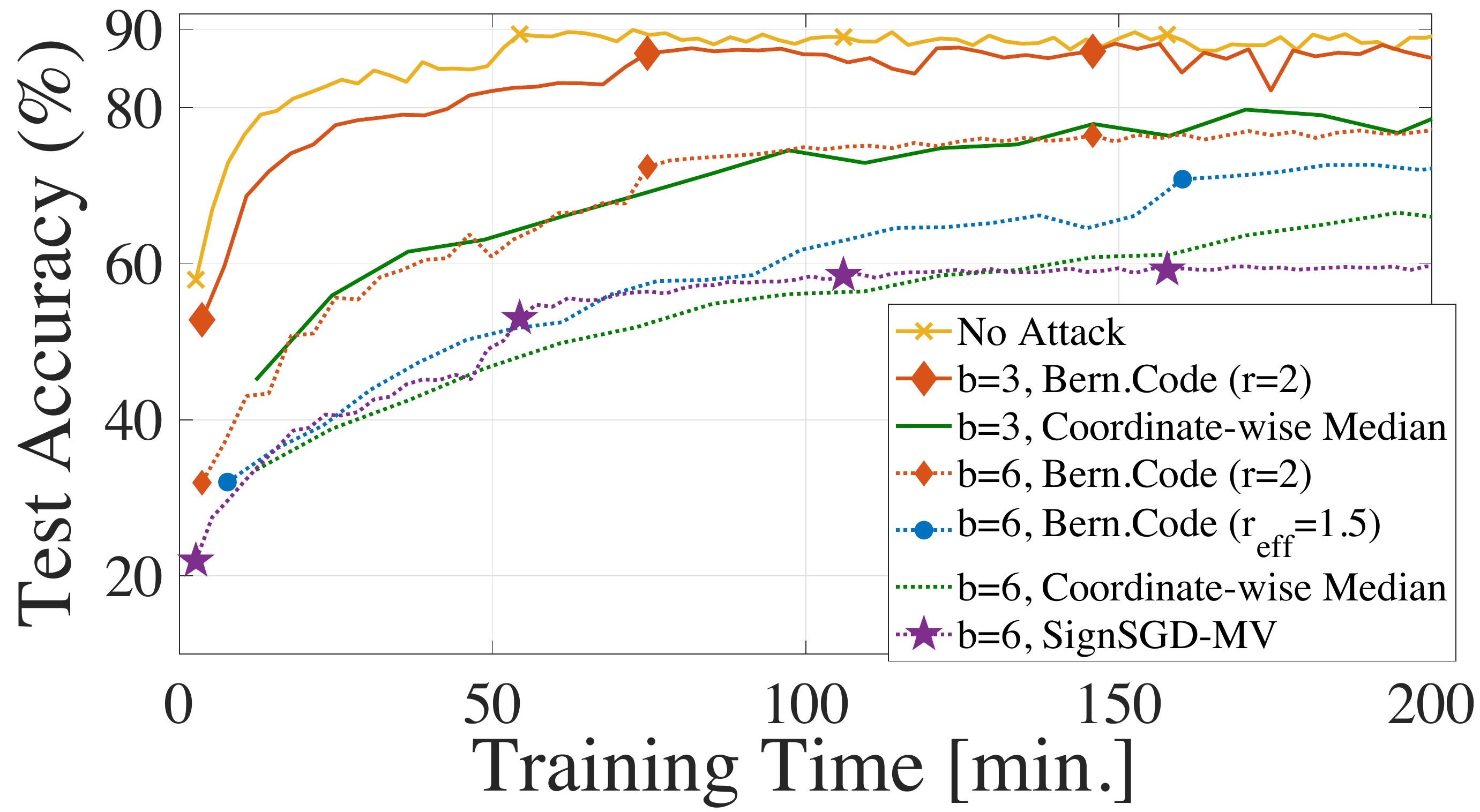}	\label{Fig:geometric_median}}
	\vspace{-0.2cm}
	\caption{\bmt{Experimental results on CIFAR-10 dataset, for $n=15$ and $b=6$ (unless stated otherwise)}}%
	\label{Fig:Simul}
	\vspace{-5mm}
\end{figure}

\section{Hyperparmeter setting in experiments for CIFAR-10 dataset}

Experiments for CIFAR-10 dataset on Resnet-18 use the hyperparameters summarized in Table~\ref{Table:Hyperparam}.  For the experiments on Resnet-50, the batch size is set to $B=64$. 

\begin{table}[h]
	\caption{Hyperparameters used in experiments for CIFAR-10 dataset on Resnet-18}
	\centering
	\small
	\label{Table:Hyperparam}
	\begin{tabular}{|c|c|c|c|c|c|c|c|c|}
		\hline
		$(n,b)$   & (5,1)   & (5,2)   & (9,2) & (9,3) & (9,4) & (15,3)            & (15,6)     & (15,7)       \\ \hline
		\begin{tabular}[c]{@{}c@{}}Batch size $B$ \\ (per data partition)\end{tabular} & 24 & 48 & 64   & \multicolumn{2}{c|}{14}    & \multicolumn{2}{c|}{16}   & 64             \\ \hline
		\begin{tabular}[c]{@{}c@{}}Learning rate\\ decaying epochs $E$ \end{tabular}        & \multicolumn{5}{c|}{{[}40, 80{]}} & \multicolumn{3}{c|}{{[}20, 40, 80{]}} \\ \hline
		Initial learning rate	$\gamma$   &   \multicolumn{8}{c|}{$10^{-4}$}           \\ \hline
		Momentum	$\eta$   &   \multicolumn{8}{c|}{$0.9$}           \\ \hline
		
	\end{tabular}
\end{table}

\section{Notations and preliminaries}

We define notations used for proving main mathematical results. 
For a given set $S$, the identification function $\mathds{1}_{\{x \in S\}}$ outputs one if $x \in S$, and outputs zero otherwise.
We denote the mapping between a message vector $\vct{m}$ and a coded vector $\vct{c}$ as $\phi (\cdot)$:
\begin{align*}
\phi (\vct{m}) = \vct{c} = [c_1, \cdots, c_n] = [E_1(\vct{m}; \mtx{G}), \cdots, E_n(\vct{m}; \mtx{G}) ].
\end{align*}
We also define the attack vector $\bm{\beta}=[\beta_1, \beta_2, \cdots, \beta_n]$, where $\beta_j = 1$ if node $j$ is a Byzantine and $\beta_j = 0$ otherwise. 
The set of attack vectors with a given support $b$ is denoted as 
$B_b = \{ \bm{\beta} \in \{0,1\}^n: \ZNorm{\bm{\beta}} = b \}.$
For a given attack vector $\vct{\beta}$,  we define an attack function 
$f_{\vct{\beta}} : \vct{c} \mapsto \vct{y}$
to represent the behavior of Byzantine nodes. 
According to the definition of $y_j$ in the main manuscript, the set of valid attack functions can be expressed as 
$\mathcal{F}_{\vct{\beta}}  \coloneqq \{ f_{\vct{\beta}} \in \mathcal{F} :  y_j = c_j \quad \forall j\in [n] \text{ with } \beta_j = 0 \},$
where
$\mathcal{F}  = \{ f :  \{0,1\}^{n} \rightarrow \{0,1\}^{n}  \}
$ is the set of all possible mappings.
Moreover, the set of message vectors $\vct{m}$ with weight $t$ is defined as
\begin{align}\label{Eqn:M_t}
M_t \coloneqq \{ \vct{m} \in \{0,1\}^n : \ZNorm{\vct{m}} = t \}.
\end{align}

Now we define several sets:
\begin{align*}
M^+ & \coloneqq \{\vct{m} \in \{ 0,1 \}^n : \ZNorm{\vct{m}} > \floor{\frac{n}{2}} \}, \quad \quad
M^-  \coloneqq \{\vct{m} \in \{ 0,1 \}^n : \ZNorm{\vct{m}} \leq \floor{\frac{n}{2}} \}, \\ 
Y^+ & \coloneqq \{\vct{y} \in \{ 0,1 \}^n : \ZNorm{\vct{y}} > \floor{\frac{n}{2}} \}, \quad \quad
Y^-  \coloneqq \{\vct{y} \in \{ 0,1 \}^n : \ZNorm{\vct{y}} \leq \floor{\frac{n}{2}} \}. 
\end{align*}
Using these definitions, Fig. \ref{Fig:mTomu} provides a description on the mapping from $\vct{m}$ to $\hat{\mu}$. 
Since decoder $D(\cdot)$ is a majority vote function, we have
$\hat{\mu} =  \mathds{1}_{ \{ \vct{y} \in Y^+  \}  }$.
Moreover, we have
$\mu = \mathds{1}_{ \{ \vct{m} \in M^+  \}  }$.

Before starting the proofs, we state several 
preliminary results. We begin with a property, which can be directly obtained from the definition of $\vct{y} = [y_1, \cdots, y_n]$. %

\begin{lemma}\label{Prop:ByzantineAttack}
	Assume that there are $b$ Byzantine nodes, i.e., the attack vector satisfies $\vct{\beta} \in B_b$. For a given vector $\vct{c}$, the output $\vct{y} = f_{\vct{\beta}}(\vct{c})$ of an arbitrary attack function $ f_{\vct{\beta}} \in \mathcal{F}_{\vct{\beta}}$ satisfies  
	$\ZNorm{\vct{y} \oplus \vct{c} } \leq b.$
	In other words, $\vct{y}$ and $\vct{c}$ differ at most $b$ positions.
\end{lemma} 

\begin{figure}
	\vspace{0mm}
	\centering
	\includegraphics[width=0.5\linewidth]{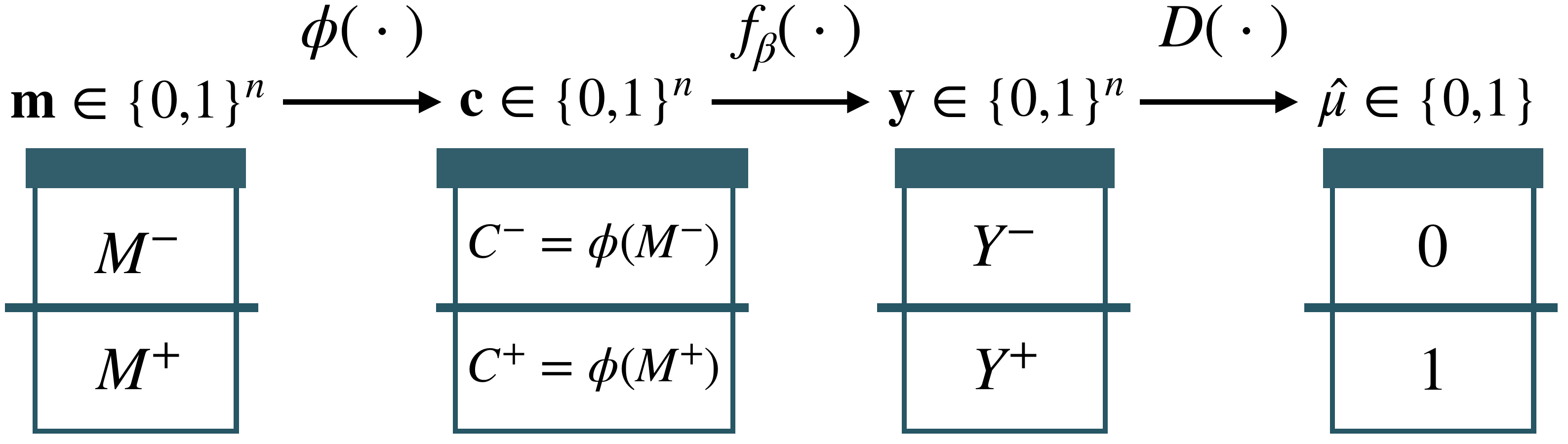}
	\caption{Mapping from $\vct{m} \in \{0,1\}^n$ to $\hat{\mu} \in \{0,1\}$. For all attack vectors $\vct{\beta} \in B_b$ and attack functions $f_{\vct{\beta}} \in \mathcal{F}_{\vct{\beta}}$, we want the overall mapping to satisfy $\hat{\mu}=1$ for all $\vct{m} \in M^-$ and $\hat{\mu}=0$ for all $\vct{m} \in M^+$. }
	\label{Fig:mTomu}
	\vspace{-5mm}
\end{figure}

Now we state four mathematical results which are useful for proving the theorems in this paper.

\begin{lemma}\label{Lemma:conv_independent}
	Consider $X = \sum_{i=1}^n X_i$ where $\{X_i\}_{i \in [n]}$ is the set of independent random variables. Then, the probability density function of $X$ is 
	\begin{align*}
	f_X = f_{X_1} * \cdots * f_{X_n} = \textrm{conv }\{f_{X_i} \}_{i \in [n]}.
	\end{align*}
\end{lemma}

\begin{lemma}[Theorem 2.1,~\cite{purkayastha1998simple}]\label{Lemma:conv_unimodal}
	Consider $f$ and $g$, two arbitrary unimodal distributions that are symmetric around zero. Then, the convolution $f * g$ is also a symmetric unimodal distribution with zero mean. 
\end{lemma}

\begin{lemma}[Lemma D.1,~\cite{pmlr-v80-bernstein18a}]\label{Lemma:node_failure_ICML18}
	Let \(\tilde{g}_{k}\) be an unbiased stochastic approximation to the $k^{th}$gradient component \(g_{k},\) with variance bounded by \(\sigma_{k}^{2} .\) Further assume that the noise distribution
	is unimodal and symmetric. Define the signal-to-noise ratio \(S_k :=\frac{\left|g_{k}\right|}{\sigma_{k}} .\) Then we have 
	$$
	\mathbb{P}\left[\operatorname{sign}\left(\tilde{g}_{k}\right) \neq \operatorname{sign}\left(g_{k}\right)\right] \leq\left\{\begin{array}{ll}
	\frac{2}{9} \frac{1}{S_k^{2}} & \text { if } S_i >\frac{2}{\sqrt{3}} \\
	\frac{1}{2}-\frac{S_k}{2 \sqrt{3}} & \text { otherwise }
	\end{array}\right.
	$$
	which is in all cases less than or equal to \(\frac{1}{2}\).
\end{lemma}

\begin{lemma}[Hoeffding's inequality for Binomial distribution]\label{Lemma:Hoeffding_binomial}
	Let $X = \sum_{i=1}^n X_i$ be the sum of i.i.d. Bernoulli random variables $X_i \sim \textrm{Bern}(p)$. For arbitrary $\epsilon > 0$, the tail probability of $X$ is bounded as 
	\begin{align*}
	&\PP{X - np  \geq n \epsilon  } \leq {\rm e}^{-2 \epsilon^2 n},\\
	&\PP{X - np  \leq -n \epsilon  } \leq {\rm e}^{-2 \epsilon^2 n}.
	\end{align*}
	As an example, when $\epsilon = \sqrt{\frac{\log n}{n}}$, the upper bound is $\frac{2}{n^2}$, which vanishes as $n$ increases.
\end{lemma}

\section{Proof of Theorems}

\subsection{Proof of Theorem~\ref{Thm:global_error_majority}}\label{Sec:Proof_of_Theorem:global_error_majority}
	Let $q_{\textrm{max}} = \max_{k} q_k$, where $q_k$ is defined in~\eqref{Eqn:q_k}. Moreover, define $u_k = 1 - q_k$ and $u_{\textrm{min}} = 1 - q_{\textrm{max}}$. Then, $u_{\textrm{min}}^{\star} = 1 - \max\limits_{k} q_k^{\star} \leq 1 - \max\limits_{k} q_k = u_{\textrm{min}}$ holds. 
	Now, we find the global estimation error probability $\PP{ \hat{\mu}_k \neq \sign(g_k)}$ for arbitrary $k$ as below. In the worst case scenario that maximizes the global error, a Byzantine node $i$ always sends the wrong sign bit, i.e., $y_{i,k} \neq \sign(g_k)$. Let $X_{i,k}$ be the random variable which represents whether the $i^{\textrm{th}}$ Byzantine-free node transmits the correct sign for coordinate $k$:
	\begin{align*}
	X_{i,k} = \begin{cases}
	1, & c_{i,k} = \sign(g_k)  \\
	0, & \textrm{otherwise}
	\end{cases}
	\end{align*}
	Then, $X_{i,k}  \sim \textrm{Bern}(p_{k})$.
	Thus, the number of nodes sending the correct sign bit is $X_{\textrm{global,k}} = \sum_{i=1}^{n(1-\alpha)} X_{i,k}$, the sum of $X_{i,k}$'s for $n-b = n(1-\alpha)$ Byzantine-free nodes,
	which follows the binomial distribution
	$X_{\textrm{global,k}} \sim \mathcal{B}(n(1-\alpha), p_{k})$.
	From the Hoeffding's inequality (Lemma~\ref{Lemma:Hoeffding_binomial}), we have
	\begin{align}\label{Eqn:global_error_Hoeffding}
	\PP{X_{\textrm{global,k}} - n(1-\alpha)u_{k} \leq - n (1-\alpha) \epsilon_k } \leq {\rm e}^{-2 \epsilon_k^2 n (1-\alpha)}
	\end{align}
	for arbitrary $\epsilon_k > 0$.
	Set $\epsilon_k = u_{k} - \frac{1}{2(1-\alpha)}$ and define $\epsilon_{\textrm{min}}= \min\limits_k \epsilon_k$. Then, 
	we have
	\begin{align}\label{Eqn:p_min_lower}
	u_{\textrm{min}}^{\star} - \frac{1}{2(1-\alpha)} > \sqrt{\frac{1}{2(1-\alpha)}} \sqrt{\frac{\log(\Delta)}{n}},
	\end{align}
	which is proven as below.
	Let $Y = \sqrt{2(1-\alpha)}$. Then,~\eqref{Eqn:p_min_lower} is equivalent to 
	\begin{align}\label{Eqn:p_min^star}
	u_{\textrm{min}}^{\star} Y^2 - \sqrt{\frac{\log(\Delta)}{n}} Y - 1 > 0,
	\end{align}
	which is all we need to prove.
	Note that~\eqref{Eqn:p_min_lower} implies that
	\begin{align*}
	\frac{Y^2}{2} > \frac{( \sqrt{\log(\Delta)/n} + \sqrt{ \log(\Delta)/n + 4 u_{\textrm{min}}^{\star} }   )^2}{8(u_{\textrm{min}}^{\star})^2},
	\end{align*}
	which is equivalent to 
	\begin{align}\label{Eqn:Y_bound}
	Y > \frac{ \sqrt{\log(\Delta)/n} }{2 u_{\textrm{min}}^{\star}} + \frac{ \sqrt{ \log(\Delta)/n + 4 u_{\textrm{min}}^{\star} }   }{2 u_{\textrm{min}}^{\star}}. 
	\end{align}
	Then,
	\begin{align*}
	u_{\textrm{min}}^{\star} Y^2 - \sqrt{\frac{\log(\Delta)}{n}} Y - 1 &= u_{\textrm{min}}^{\star} \left(  Y - \frac{1}{2 u_{\textrm{min}}^{\star} } \sqrt{\frac{\log(\Delta)}{n}} \right)^2 - \frac{1}{4 u_{\textrm{min}}^{\star}} \frac{\log(\Delta)}{n} - 1 > 0,
	\end{align*}
	which proves~\eqref{Eqn:p_min^star} and thus~\eqref{Eqn:p_min_lower}.
	Thus, we have $\epsilon_k \geq \epsilon_{\textrm{min}} \geq u_{\textrm{min}}^{\star} - \frac{1}{2(1-\alpha)} > \sqrt{\frac{1}{2(1-\alpha)}} \sqrt{\frac{\log(\Delta)}{n}} > 0$.
	Then,~\eqref{Eqn:global_error_Hoeffding} reduces to
	\begin{align*}
	P_{\textrm{global,k}} = 	\PP{ \hat{\mu}_k \neq \sign(g_k)}   = 	\PP{X_{\textrm{global,k}} \leq n/2} \leq {\rm e}^{-2 \epsilon_k^2 n (1-\alpha)} < 1/\Delta,
	\end{align*}
	which completes the proof.

\subsection{Proof of Theorem \ref{Thm:Convergence}}\label{Sec:Proof_of_Theorem:Convergence}

	Here we basically follow the proof of~\cite{bernstein2018signsgd_ICLR} with a slight modification, reflecting the result of Theorem~\ref{Thm:global_error_majority}.
	Let $\vct{w}_1, \cdots, \vct{w}_T$ be the sequence of updated models at each step. 
	Then, we have
	\begin{align*}
	f(\vct{w}_{t+1}) - f(\vct{w}_{t}) & \stackrel{(a)}{\leq} g(\vct{w}_{t})^T (\vct{w}_{t+1} - \vct{w}_t) + \sum_{k=1}^d \frac{L_k}{2} (\vct{w}_{t+1} - \vct{w}_t)_k^2 \stackrel{(b)}{=} - \gamma g(\vct{w}_t)^T \hat{\vct{\mu}} + \gamma^2 \sum_{k=1}^d \frac{L_k}{2} \\
	&= -\gamma \lVert g(\vct{w}_t) \rVert_1 + 2 \gamma \sum_{i=1}^d \lvert (g(\vct{w}_t))_k \rvert \cdot \mathds{1}_{ \{\sign ((g(\vct{w}_t))_k) \neq \hat{\mu}_k \}} + \gamma^2 \sum_{k=1}^d \frac{L_k}{2},
	\end{align*}
	where (a) is from Assumption~\ref{Assumption:smooth}, and (b) is obtained from $\vct{w}_{t+1} = \vct{w}_t  - \gamma \vct{\hat{\mu}} $ and $\lvert \hat{\mu}_k \rvert = 1$. Let $g_k$ be a simplified notation of $(g(\vct{w}_t))_k$. Then, taking the expectation of the equation above, we have
	\begin{align}\label{Eqn:expectation_convergence}
	\EE{f(\vct{w}_{t+1}) - f(\vct{w}_{t}) } &\leq -\gamma \lVert g(\vct{w}_t) \rVert_1 + \frac{\gamma^2}{2} \lVert \vct{L} \rVert_1  + 2 \gamma \sum_{k=1}^d \lvert g_k \rvert \  \PP{ \hat{\mu}_k \neq  \sign (g_k)  } .
	\end{align}
	Inserting the result of Theorem~\ref{Thm:global_error_majority} to \eqref{Eqn:expectation_convergence}, we have
	\begin{align*}
	\EE{f(\vct{w}_{t+1}) - f(\vct{w}_{t}) } &\leq -\gamma \left(1 - \frac{2}{\Delta}\right)  \lVert g(\vct{w}_t) \rVert_1 + \frac{\gamma^2}{2} \lVert \vct{L} \rVert_1   \\
	&\stackrel{(a)}{=} -\sqrt{\frac{f(\vct{w}_{0})-f^{*}}{\|\vct{L}\|_{1} T}} \left(1 - \frac{2}{\Delta} \right)  \left\|g(\vct{w}_t)\right\|_{1}+\frac{f(\vct{w}_{0})-f^{*}}{2 T}
	\end{align*}
	where (a) is from the parameter settings of $\gamma$ in the statement of Theorem~\ref{Thm:Convergence}.
	Thus, we have 
	\begin{align*}
	f(\vct{w}_{0})-f^{*} & \geq f(\vct{w}_{0})-\EE{f(\vct{w}_{T})} =\sum_{t=0}^{T-1} \EE{f(\vct{w}_{t})-f(\vct{w}_{t+1})} \\ & \geq \sum_{t=0}^{T-1} \EE{\sqrt{\frac{f(\vct{w}_{0})-f^{*}}{\| \vct{L} \|_{1} T}} \left( 1- \frac{2}{\Delta} \right)  \left\|g(\vct{w}_{t})\right\|_{1}-\frac{f(\vct{w}_{0})-f^{*}}{2 T}} \\ &=\sqrt{\frac{T\left(f(\vct{w}_{0})-f^{*}\right)}{\|\vct{L}\|_{1}}} \EE{\frac{1}{T} \sum_{t=0}^{T-1}\left\|g(\vct{w}_{t})\right\|_{1}} \left( 1- \frac{2}{\Delta} \right) 
	-\frac{f(\vct{w}_{0})-f^{*}}{2}.    
	\end{align*}
	The expected gradient (averaged out over $T$ iterations) is expressed as
	\begin{align*}
	\frac{1}{T} \sum_{t=0}^{T-1}  \EE{ \left\|g(\vct{w}_{t})\right\|_{1}} \leq \frac{1}{\sqrt{T}} \frac{1}{1-\frac{2}{\Delta}}  \frac{3}{2} \sqrt{\lVert \vct{L} \rVert_1 (f(\vct{w}_0) - f^*) } 	\rightarrow 0 \quad \textrm { as   } T \rightarrow \infty 
	\end{align*}
	Thus, the gradient becomes zero eventually, which completes the convergence proof.

\subsection{Proof of Theorem \ref{Thm:upper_bound}}\label{Sec:Proof_of_Thm:upper_bound}

Recall that according to Lemma~\ref{Lemma:NSCondition} and the definition of $M_t$ in \eqref{Eqn:M_t}, the system using the allocation matrix $\mtx{G}$ is perfect $b-$Byzantine tolerable if and only if
	\begin{equation}\label{Eqn:Byz_Lemma_supp}
	\sum_{v=1}^{(n-1)/2} \lvert S_v (\vct{m})
	\rvert  \leq \frac{n-1}{2} -b
	\end{equation}
	holds for arbitrary message vector $\vct{m} \in M_{\frac{n-1}{2}}$, where $S_v(\vct{m})$ is defined in~\eqref{Eqn:J_v}.
	Note that we have
	\begin{align}\label{Eqn:G_j_norm}
	\ZNorm{\mtx{G}(j,:)} =
	\begin{cases}
	1, &  1 \leq j \leq s \\
	2b+1, & s+1 \leq j \leq s+L \\
	n, & s+L+1 \leq j \leq n.
	\end{cases}
	\end{align}
	from Fig. \ref{Fig:Generator_matrix}. Thus, the condition in \eqref{Eqn:Byz_Lemma_supp} reduces to 
	\begin{align}\label{Eqn:Byz_Lemma_supp_2}
	\lvert S_1 (\vct{m}) \rvert + \lvert S_{b+1} (\vct{m}) \rvert 
	\leq s.
	\end{align}
	Now all that remains is to show that \eqref{Eqn:Byz_Lemma_supp_2} holds for arbitrary message vector
	$\vct{m} \in M_{\frac{n-1}{2}}$.

	Consider a message vector $\vct{m}\in M_{\frac{n-1}{2}}$ denoted as $\vct{m} = [m_1, m_2, \cdots, m_n]$. Here, we note that 
	\begin{align}
	S_1(\vct{m}) & \subseteq \{ 1, 2, \cdots, s \}, 	\quad S_{b+1}(\vct{m}) \subseteq \{ s+1, s+2, \cdots, s+L \} \label{Eqn:J_1_subset}
	\end{align}
	hold from Fig. \ref{Fig:Generator_matrix}. Define 
	\begin{align}
	v(\vct{m}) \coloneqq \lvert \{ i\in \{s+1, s+2, \cdots, n  \} : m_i = 1    \}   \rvert,
	\end{align}
	which is the number of $1$'s in the last $(n-s)$ coordinates of message vector $\vct{m}$.
	Since $\vct{m}\in M_{\frac{n-1}{2}}$, we have 
	\begin{align} \label{Eqn:v_complement}
	\lvert \{ i\in \{1,2,\cdots, s \} : m_i = 1 \}   \rvert = \frac{n-1}{2} - v(\vct{m}).
	\end{align}
	Note that since
	$\mtx{G}(1:s, :) = [\ \mtx{I}_s \ | \ \mathbf{0}_{s \times (n-s)}]$, we have 
	\begin{align}\label{Eqn:inner_product_J_1}
	\vct{m}^T \mtx{G}(j,:) = \mathds{1}_{\{m_j = 1\}} , \quad  \ZNorm{\mtx{G}(j,:)} = 1, \quad \forall j \in [s].
	\end{align}
	
	Combining \eqref{Eqn:J_v}, \eqref{Eqn:J_1_subset}, \eqref{Eqn:v_complement}, and \eqref{Eqn:inner_product_J_1}, we have
	$\lvert S_1(\vct{m})  \rvert = \frac{n-1}{2}-v(\vct{m}).$
	Now, in order to obtain \eqref{Eqn:Byz_Lemma_supp_2}, all we need to prove is to show
	\begin{align}\label{Eqn:J_b+1}
	\lvert S_{b+1}(\vct{m}) \rvert \leq s - \left(\frac{n-1}{2} - v(\vct{m})\right) \overset{\mathrm{(a)}}{=} v(\vct{m}) - b
	\end{align}
	where $(a)$ is from the definition of $s$ in Algorithm \ref{Alg:DetCode}. We alternatively prove that\footnote{Note that \eqref{Eqn:J_b+1_alternative} implies \eqref{Eqn:J_b+1}, when the condition part is restricted to $\lvert S_{b+1}(\vct{m}) \rvert = q$.}  
	\begin{align}\label{Eqn:J_b+1_alternative}
	\text{if } \lvert S_{b+1}(\vct{m}) \rvert \geq q \text{ for some }  q \in \{0,1,\cdots, L\}, \text{ then } v(\vct{m}) \geq b+q.
	\end{align}
	Using the definition
	$M^{(q)} \coloneqq \{ \vct{m} \in M_{\frac{n-1}{2}} : \lvert S_{b+1}(\vct{m}) \rvert \geq q  \},$
	the statement in \eqref{Eqn:J_b+1_alternative} is proved as follows: for arbitrary $q \in \{0,1,\cdots, L\}$, we first find the minimum $v(\vct{m})$ among $\vct{m} \in M^{(q)}$, i.e., we obtain a closed-form expression for 
	\begin{align}\label{Eqn:optimal_v_def}
	v^*_q \coloneqq &\min_{\vct{m} \in M^{(q)}} v(\vct{m}).
	\end{align}
	Second, we show that $v^*_q \geq b+ q$ holds for all $q \in \{0,1,\cdots, L\}$, which completes the proof.
	\ifarxiv
	\begin{figure}
		[!t]
		\centering
		\includegraphics[width=0.5\linewidth]{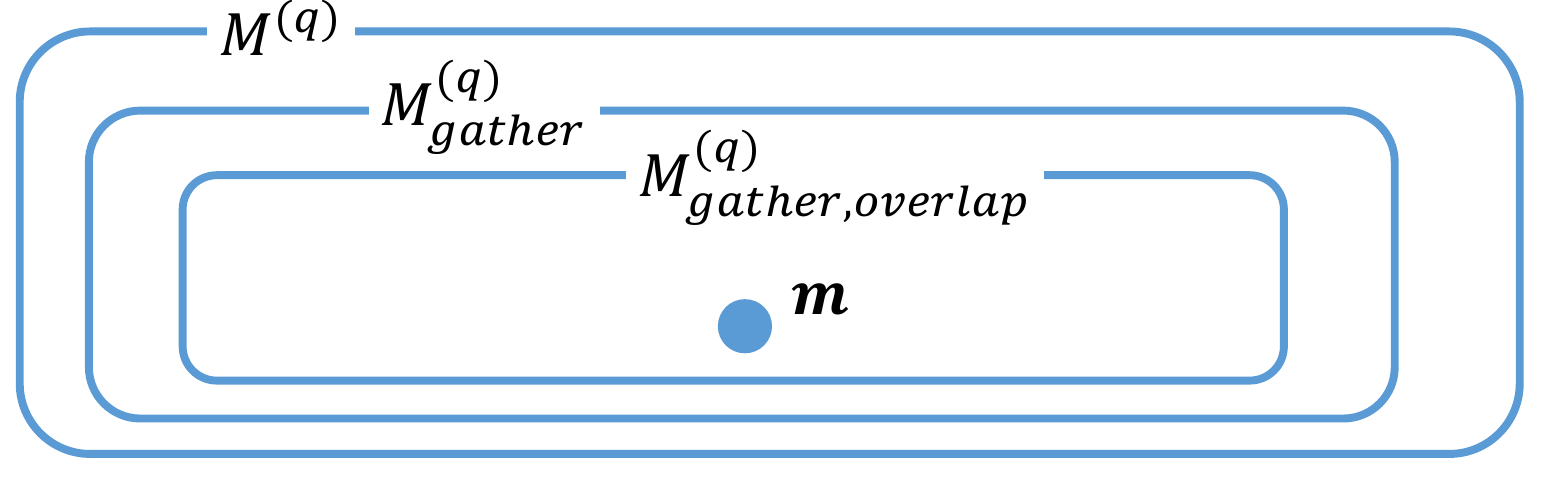}
		\caption{%
			Sets of message vectors used in proving Lemmas \ref{Lemma:gather} and \ref{Lemma:gather_overlap}}
		\label{Fig:M_set}
	\end{figure}

	The expression for $v^*_q$ can be obtained as follows. Fig. \ref{Fig:M_set} supports the explanation. First, define 
	\begin{align}
	M_{gather}^{(q)} \coloneqq \{ \vct{m} \in M^{(q)} &: 
	\text{ if } j, j+2 \in S_{b+1}(\vct{m}), \text{ then } j+1 \in S_{b+1}(\vct{m})  \},
	\end{align} 
	the set of message vectors $\vct{m}$ which satisfy that $S_{b+1}(\vct{m})$ is consisted of consecutive integers. 
	We now provide a lemma which states that within $M_{gather}^{(q)}$, there exists a minimizer of the optimization problem \eqref{Eqn:optimal_v_def}.
	\begin{lemma}\label{Lemma:gather}
		For arbitrary $q \in \{0,1,\cdots, L\}$, 
		we have
		\begin{align*}
		v_q^* = \min_{\vct{m} \in M_{gather}^{(q)}} v(\vct{m}).
		\end{align*}
	\end{lemma}
	\begin{proof}
		From Fig. \ref{Fig:M_set} and the definition of $v^{*}_{q}$, all we need to prove is the following statement: for all $\vct{m} \in M^{(q)} \cap (M_{gather}^{(q)})^c$, we can assign another message vector $\vct{m}^* \in M_{gather}^{(q)}$ such that $v(\vct{m}^*) \leq v(\vct{m})$ holds. Consider arbitrary $\vct{m} \in M^{(q)} \cap (M_{gather}^{(q)})^c$, denoted as $\vct{m} = [m_1, m_2, \cdots, m_n]$. Then, there exist integers $j \in \{1, \cdots, L \}$ and $\delta \in \{2, 3, \cdots, L-j \}$ such that $s+j, s+j+\delta \in S_{b+1}(\vct{m})$ and $s+j+1, \cdots, s+j+\delta-1 \notin S_{b+1}(\vct{m})$ hold. Select the smallest $j$ which satisfies the condition. 
		Consider $\vct{m}' =[m_1', \cdots, m_n']$ generated as the following rule:
		\begin{enumerate}
			\item The first $s+j(b+1)$ elements (which affect the first $j$ rows of $\mtx{A}$ in Figure \ref{Fig:Generator_matrix})  of $\vct{m}'$ is identical to that of $\vct{m}$. 
			\item The last $n-(j+\delta-1)(b+1)-s$ elements of $\vct{m}$ are shifted to the left by $(\delta-1)(b+1)$, and inserted to $\vct{m}'$. In the shifting process, we have $b$ locations where the original $m_i$ and the shifted $m_{i+(\delta-1)(b+1)}$ overlap. In such locations, $m_i'$ is set to the maximum of two elements; if either one is 1, we set $m_i'=1$, and otherwise we set $m_i'=0$.
		\end{enumerate}

		This can be mathematically expressed as below:
		\begin{align}\label{Eqn:m_i'}
		m_i' =
		\begin{cases}
		m_i, & 1 \leq i \leq s+j(b+1) \\
		\max \{ m_i, m_{i+(\delta-1) (b+1)}  \}, & s+j(b+1)+1 \leq i \leq s+(j+1)(b+1) \\
		m_{i+(\delta-1) (b+1)} , & s+(j+1)(b+1)+1 \leq i \leq n-(\delta-1)(b+1) \\
		0, & n-(\delta-1)(b+1)+1 \leq i \leq n
		\end{cases}
		\end{align}
		Note that we have
		\begin{align}\label{Eqn:m_i_sum}
		\sum_{i=1}^n m_i = \frac{n-1}{2}
		\end{align}
		since $\vct{m} \in M_{(n-1)/2}$. Moreover, \eqref{Eqn:m_i'} implies
		\begin{align}
		\sum_{i=1}^n m_i' &= \sum_{i=1}^{s+j(b+1)} m_i' +  \sum_{i=s+j(b+1)+1}^{s+(j+1)(b+1)} m_i' +  \sum_{i=s+(j+1)(b+1)+1}^{n-(\delta-1)(b+1)} m_i'  \nonumber\\
		&= \sum_{i=1}^{s+j(b+1)} m_i +  \sum_{i=s+j(b+1)+1}^{s+(j+1)(b+1)} m_i' +  \sum_{i=s+(j+\delta)(b+1)+1}^{n} m_i \nonumber\\
		&\overset{\eqref{Eqn:m_i_sum}}{=} \frac{n-1}{2} - \left( \sum_{i=s+j(b+1)+1}^{s+(j+\delta)(b+1)} m_i - \sum_{i=s+j(b+1)+1}^{s+(j+1)(b+1)} m_i' \right)  \overset{\mathrm{(a)}}{\leq} \frac{n-1}{2} \label{Eqn:m_i'_sum}
		\end{align}
		where Eq.$(a)$ is from 
		\begin{align*}
		\sum_{i=s+j(b+1)+1}^{s+(j+\delta)(b+1)} m_i &\overset{\mathrm{(b)}}{\geq} \sum_{i=s+j(b+1)+1}^{s+(j+1)(b+1)} (m_i + m_{i+(\delta-1)(b+1)}) \geq \sum_{i=s+j(b+1)+1}^{s+(j+1)(b+1)} \max \{ m_i, m_{i+(\delta-1)(b+1)} \} \\
		&\overset{\eqref{Eqn:m_i'}}{=} \sum_{i=s+j(b+1)+1}^{s+(j+1)(b+1)} m_i'
		\end{align*}
		and Eq.$(b)$ is from $\delta \geq 2$. 
		Note that 
		\begin{align}\label{Eqn:v_message_case0}
		v(\vct{m}') = v(\vct{m}) - \epsilon
		\end{align}
		holds for 
		\begin{align}\label{Eqn:weight_gap}
		\epsilon \coloneqq \frac{n-1}{2} - \sum_{i=1}^n m_i',
		\end{align}
		which is a non-negative integer from \eqref{Eqn:m_i'_sum}.
		Now, we show the behavior of $S_{b+1}(\vct{m})$ as follows. Recall that for $j_0 \in \{s+1, \cdots, s+L\}$,
		\begin{equation}\label{Eqn:A_one_sequence}
		\mtx{G} (j_0, i_0) = 
		\begin{cases}
		1, & \text{ if } s+(j_0-s-1)(b+1)+ 1 \leq i_0 \leq s+(j_0-s-1)(b+1)+2b+1\\
		0, & \text{ otherwise}
		\end{cases}
		\end{equation}
		holds from Algorithm \ref{Alg:DetCode} and Fig. \ref{Fig:Generator_matrix}. Define 
		\begin{align*}
		S_+  & \coloneqq \{ j' \in \{ s+j+\delta, \cdots, s+L \} : j' \in S_{b+1}(\vct{m})  \}, \\
		S_- & \coloneqq \{ j' \in \{ s+1, \cdots, s+j \} : j' \in S_{b+1}(\vct{m})  \}.
		\end{align*}
		From \eqref{Eqn:m_i'} and \eqref{Eqn:A_one_sequence}, we have
		\begin{align*}
		\begin{cases}
		S_- \subseteq S_{b+1}(\vct{m}'), \\
		\text{ if } j' \in S_+, \text{ then } j'-(\delta-1) \in S_{b+1}(\vct{m}').
		\end{cases}
		\end{align*}
		Thus, we have 
		\begin{align}\label{Eqn:J_b+1_bound_m'}
		\lvert S_{b+1} (\vct{m}')  \rvert  \geq \lvert S_- \rvert +\lvert S_+ \rvert = \lvert S_{b+1}(\vct{m}) \rvert \overset{\mathrm{(a)}}{\geq}  q
		\end{align}
		where Eq.$(a)$ is from $\vct{m} \in M^{(q)}$.

		Now we construct $\vct{m}'' \in M^{(q)}$ which satisfies $v(\vct{m}'') \leq v(\vct{m})$.
		Define $S_0 \coloneqq \{ i\in [s]: m_i' = 0 \}$ and $\epsilon_0 \coloneqq \lvert S_0 \rvert $. The message vector $\vct{m}'' =[m_1'', \cdots, m_n'']$ is defined as follows. 
		
		\textbf{Case I} (when $\epsilon \leq \epsilon_0$):
		Set $m_i'' = m_i'$ for $i \in \{s+1, s+2, \cdots, n\}$. Randomly select $\epsilon$ elements in $S_0$, denoted as
		$\{i_1, \cdots, i_{\epsilon}\} = S_0^{(\epsilon)} \subseteq S_0.$
		Set $m_i'' = 1$ for $i \in S_0^{(\epsilon)}$, and set $m_i'' = m_i'$ for $i \in S_0 \setminus S_0^{(\epsilon)}$.
		Note that this results in 
		\begin{align}
		v(\vct{m}'') = v(\vct{m}').  \label{Eqn:v_message_case1}
		\end{align}
		
		\textbf{Case II} (when $\epsilon > \epsilon_0$):	
		Set $m_i'' = 1$ for $i \in [s]$. Define $S_1 \coloneqq \{ i\in \{ s+1, \cdots, n\}: m_i' = 0 \}$. Randomly select $\epsilon-\epsilon_0$ elements in $S_1$, denoted as
		$\{i_1', \cdots, i_{\epsilon-\epsilon_0}'\} = S_1^{(\epsilon)} \subseteq S_1.$
		Set $m_i'' = 1$ for $i \in S_1^{(\epsilon)}$, and set $m_i'' = m_i'$ for $i \in \{s+1, \cdots, n\} \setminus S_1^{(\epsilon)}$.
		Note that this results in 
		\begin{align}
		v(\vct{m}'') = v(\vct{m}') + (\epsilon-\epsilon_0).  \label{Eqn:v_message_case2}
		\end{align}
		
		Note that in both cases, the weight of $\vct{m}''$ is 
		\begin{align}\label{Eqn:m''_weight}
		\ZNorm{\vct{m}''} = \sum_{i=1}^n m_i'' = \sum_{i=1}^n m_i'  + \epsilon \overset{\eqref{Eqn:weight_gap}}{=} \frac{n-1}{2}.
		\end{align}
		Moreover,
		\begin{align}\label{Eqn:J_b+1_bound_m''}
		\lvert S_{b+1}(\vct{m}'') \rvert \overset{\textrm{(a)}}{\geq} \lvert S_{b+1}(\vct{m}') \rvert  \overset{\eqref{Eqn:J_b+1_bound_m'}}{\geq} q
		\end{align}
		holds where Eq.$(a)$ is from the fact that all elements of 
		$\vct{m}''-\vct{m}$ are non-negative. Finally, %
		\begin{align}\label{Eqn:v_gap}
		v(\vct{m}'') = v(\vct{m}) - \min \{ \epsilon, \epsilon_0 \} \leq v(\vct{m})
		\end{align}
		holds from \eqref{Eqn:v_message_case0}, \eqref{Eqn:v_message_case1}, and \eqref{Eqn:v_message_case2}. Combining \eqref{Eqn:m''_weight}, \eqref{Eqn:J_b+1_bound_m''} and \eqref{Eqn:v_gap}, one can confirm that $\vct{m}'' \in M^{(q)}$ and $v(\vct{m}'') \leq v(\vct{m})$ hold; this \textit{gathering} process\footnote{In Fig. \ref{Fig:gathering_process}, one can confirm that $S_{b+1}(\vct{m})$ is not consisted of consecutive integers (i.e., there's a gap), while $S_{b+1}(\vct{m}'')$ has no gap. Thus, we call this process as \textit{gathering} process.} maintains the weight of a message vector and does not increase the $v$ value. Let $\vct{m}^*$ be the message vector generated by applying this gathering process to $\vct{m}$ sequentially until $S_{b+1}(\vct{m}^*)$ is consisted of consecutive integers.
		Then, $\vct{m}^*$ satisfies the followings:
		\begin{enumerate}
			\item $S_{b+1}(\vct{m}^*)$ contains more than $q$ elements. Moreover, since $S_{b+1}(\vct{m}^*)$ is consisted of consecutive integers, we have $ \vct{m}^* \in M_{gather}^{(q)}$. 
			\item $v(\vct{m}^*) \leq  v(\vct{m}'') \leq v(\vct{m})$ holds.
		\end{enumerate}
		Since the above process of generating $\vct{m}^* \in M_{gather}^{(q)}$ is valid for arbitrary message vector $\vct{m} \in M^{(q)} \cap (M_{gather}^{(q)})^c$, this completes the proof.

	\begin{figure}
		\centering
		\includegraphics[width=0.85\linewidth]{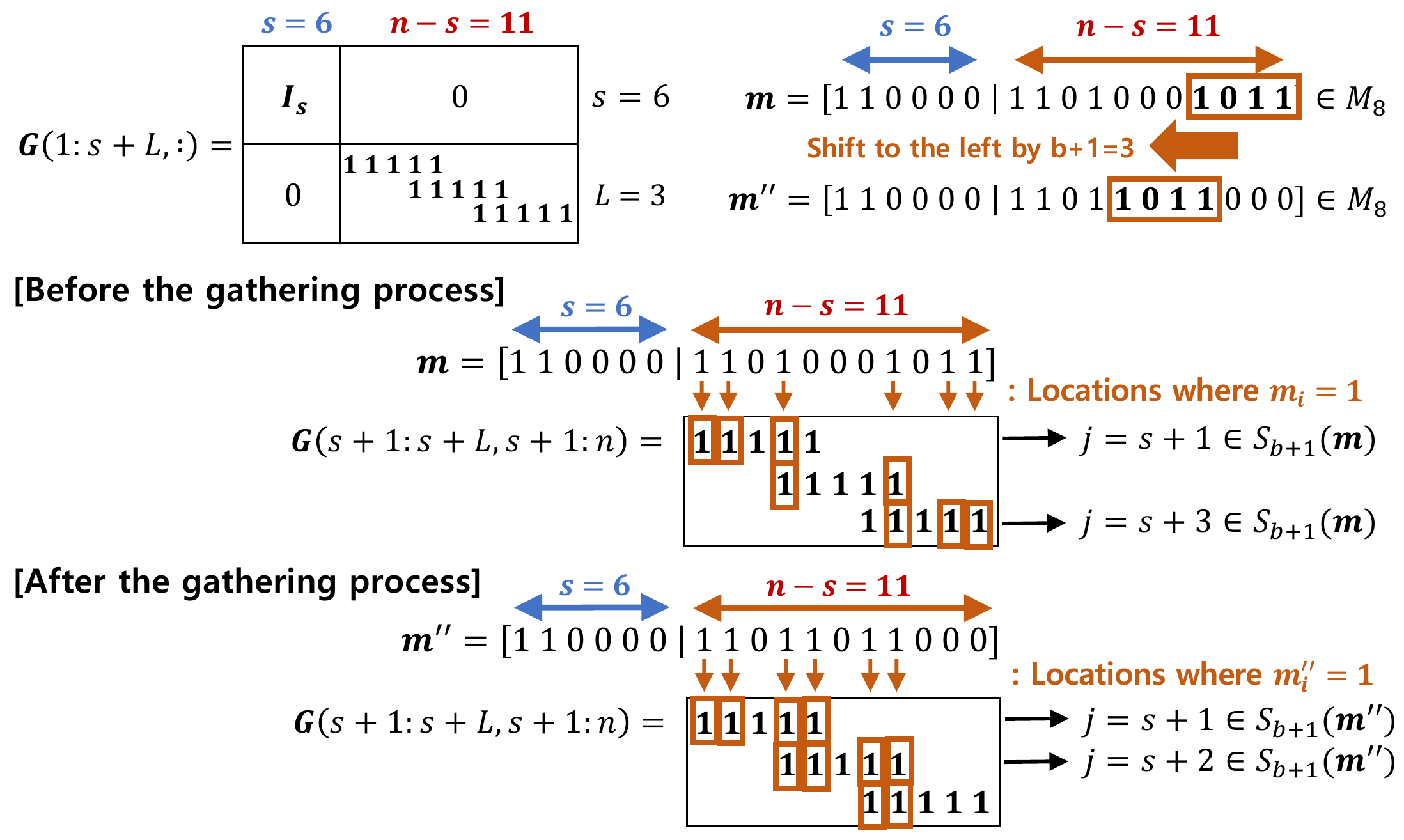}
		\caption{The gathering process illustrated in the proof of Lemma \ref{Lemma:gather}, when $n=17, b=2$. Under this setting, we have $S_{b+1}(\vct{m}) = S_{3}(\vct{m}) = \{ j \in \{s+1, \cdots, s+L\} : \vct{m}^T G(j,:) \geq 3  \}$. Before the gathering process, we have $S_{b+1}(\vct{m})  = \{s+1, s+3\}$, while $S_{b+1}(\vct{m}'')  = \{s+1, s+2\}$ holds after the process.}
		\label{Fig:gathering_process}
	\end{figure}

	Now consider arbitrary message vectors satisfying $\vct{m} \in M_{gather}^{(q)}$. 
	Then, we have
	\begin{align}\label{Eqn:J_b+1_overlap}
	S_{b+1}(\vct{m}) = \{j, j+1, \cdots, j+\delta-1 \} 
	\end{align}
	for some $j \in \{ s+1, \cdots, s+L-\delta+1\}$ and $\delta \geq q$.
	Here, we define 
	\begin{align}\label{Eqn:M_gather_overlap}
	M_{gather,overlap}^{(q)} = \left\{ \vct{m} \in M_{gather}^{(q)} : m_{s+(j_0-s)(b+1)}= 0 \text{ for } j_0 \in \{j, \cdots, j  + \delta - 1 \} \right\}
	\end{align}
	We here prove that arbitrary message vector $\vct{m} \in M_{gather}^{(q)}$ can be mapped into another message vector $\vct{m}' \in M_{gather, overlap}^{(q)}$ without increasing the corresponding $v$ value, i.e., $v(\vct{m}') \leq v(\vct{m})$. 
	Given a message vector $\vct{m} \in M_{gather}^{(q)}$ , define $\vct{m}' = [m_1', m_2', \cdots, m_n']$  as in Algorithm \ref{Alg:message_vector_overlapping}.
	In line 9 of this algorithm, we can always find $l \in [s]$ that satisfies $m_l = 0$, due to the following reason. 
	Note that 
	\begin{align}
	\sum_{i=s+1}^n m_i \overset{\mathrm{(a)}}{\geq} \vct{m}^T \mtx{G}(j,:) \overset{\eqref{Eqn:J_b+1_overlap}}{\geq} b+1
	\end{align}
	holds where $(a)$ is from the fact that $G(j,i)=0$ for $i\in[s]$ as in \eqref{Eqn:A_one_sequence}.  Thus, we have
	\begin{align}\label{Eqn:existence_l}
	\sum_{i=1}^s m_i \overset{\eqref{Eqn:J_b+1_overlap}}{\leq} \sum_{i=1}^n m_i - (b+1) = \frac{n-1}{2} - (b+1) = s-1.
	\end{align}
	Therefore, we have
	\begin{align}\label{Eqn:existence_m_l_0}
	\exists l \in [s] \text{ such that } m_l = 0. 
	\end{align}  
	
	The vector $\vct{m}'$ generated from Algorithm \ref{Alg:message_vector_overlapping} satisfies the following four properties:
	\begin{enumerate}
		\item $\vct{m}' \in M_{(n-1)/2}$,
		\item $S_{b+1}(\vct{m}') = S_{b+1}(\vct{m}) = \{j, j+1, \cdots, j+\delta- 1\}$,
		\item $\vct{m}'  \in M_{gather, overlap}^{(q)} $,
		\item $v(\vct{m}') \leq v(\vct{m})$.
	\end{enumerate}
	The first property is from the fact that lines 7 and 10 of the algorithm maintains the weight of the message vector to be $\ZNorm{\vct{m}} = (n-1)/2$. 
	The second property is from the fact that 
	\begin{align*}
	(\vct{m}')^T \mtx{G}(j_0, :) &\overset{\mathrm{\eqref{Eqn:A_one_sequence}}}{=} \sum_{i=1}^{2b+1} m_{s+(j_0-s-1)(b+1)+i}' \\
	&\overset{\mathrm{\textrm{(a)}}}{=} \begin{cases}
	\sum_{i=1}^{2b+1} m_{s+(j_0-s-1)(b+1)+i} ,& \text{ if line 6 of Algorithm \ref{Alg:message_vector_overlapping} is satisfied} \\
	2b, &  \text{ otherwise}
	\end{cases} 
	\overset{\mathrm{\eqref{Eqn:J_b+1_overlap}}}{\geq} b+1
	\end{align*}
	for $j_0 \in \{ j, j+1, \cdots, j+\delta-1  \}$, where $(a)$ is from the fact that $\sum_{i=1}^{2b+1} m_{s+(j_0-s-1)(b+1)+i} = 2b+1$ holds if line 6 of Algorithm \ref{Alg:message_vector_overlapping} is not satisfied. The third property is from the first two properties and the definition of $M_{gather, overlap}^{(q)}$ in \eqref{Eqn:M_gather_overlap}.
	The last property is from the fact that 1) each execution of line 7 in the algorithm maintains $v(\vct{m}')=v(\vct{m})$, and 2) each execution of line 10 in the algorithm results in $v(\vct{m}') = v(\vct{m})-1$.
	Thus, 
	combining with Lemma \ref{Lemma:gather}, we have the following lemma: 
	\begin{lemma}\label{Lemma:gather_overlap}
		For arbitrary $q \in \{0,1,\cdots, L\}$, we have
		\begin{align*}
		v_q^* = \min_{\vct{m} \in M_{gather, overlap}^{(q)}} v(\vct{m}).
		\end{align*}
	\end{lemma}

	\begin{algorithm}[t]
		\small
		\caption{\small Defining $\vct{m}' \in M_{gather,overlap}^{(q)}$ from arbitrary $\vct{m} \in M_{gather}^{(q)}$.}
		\label{Alg:message_vector_overlapping}
		\begin{algorithmic}[1]
			\STATE \textbf{Input:} message vector $\vct{m}=[m_1, m_2, \cdots, m_n]$ having $S_{b+1}(\vct{m}) = \{j, j+1, \cdots, j+\delta-1\}$
			\STATE \textbf{Output:} message vector $\vct{m}' =[m_1', m_2', \cdots, m_n']$
			\STATE \textbf{Initialize:} $\vct{m}' = \vct{m}$
			\FOR{$j_0=j$ to $j+\delta-1$}
			\IF{$m_{s+(j_0-s)(b+1)}=1$}
			\IF{$\exists i \in [2b+1]$ such that $m_{s+(j_0-s-1)(b+1)+i}=0$}
			\STATE $m_{s+(j_0-s-1)(b+1)+i}' \leftarrow 1$, \quad  $m_{s+(j_0-s)(b+1)}' \leftarrow 0$. 
			\ELSE
			\STATE Find $l \in [s]$ such that $m_l = 0$ \text{ (The existence of such $l$ is proven in \eqref{Eqn:existence_m_l_0}.) }
			\STATE $m_{l}' \leftarrow 1$, \quad $m_{s+(j_0-s)(b+1)}' \leftarrow 0$. 
			\ENDIF
			\ENDIF
			\ENDFOR
		\end{algorithmic}
	\end{algorithm}

	\begin{figure} [!t]
		\centering
		\includegraphics[width=0.75\linewidth]{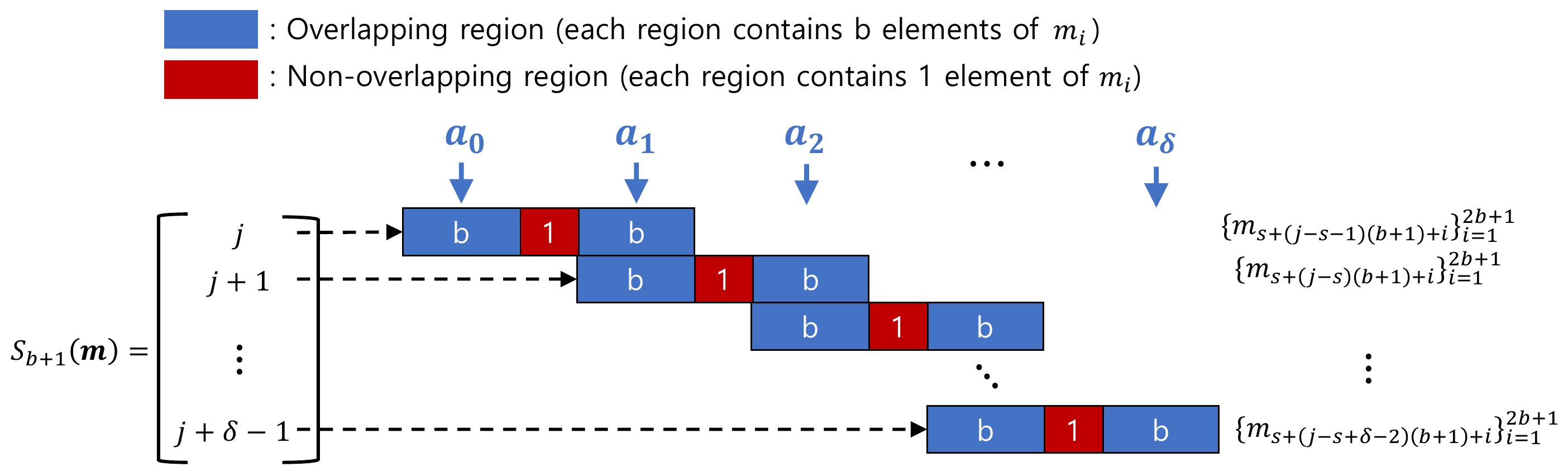}
		\caption{The illustration of overlapping regions and $\{a_l\}_{l=0}^\delta$ in \eqref{Eqn:a_l_seq}}
		\label{Fig:overlapping_region}
	\end{figure}

	According to Lemma \ref{Lemma:gather_overlap}, in order to find $v^*_q$, all that remains is to find the optimal $\vct{m} \in M_{gather,overlap}^{(q)}$ which has the minimum $v(\vct{m})$.
	Consider arbitrary $\vct{m} \in M_{gather,overlap}^{(q)}$ and denote 
	\begin{align}\label{Eqn:J_b+1_gather_overlap}
	S_{b+1}(\vct{m}) = \{j, j+1, \cdots, j+\delta-1 \}. 
	\end{align}
	Define the corresponding assignment vector $\{a_l\}_{l=0}^{\delta}$ as
	\begin{align}\label{Eqn:a_l_seq}
	a_l = \sum_{i=1}^b m_{s+(j-s-1+l)(b+1)+i},
	\end{align}
	which represents the number of indices $i$ satisfying $m_i=1$ within $l^{\text{th}}$ overlapping region, as illustrated in Fig. \ref{Fig:overlapping_region}.
	Then, we have
	\begin{align}
	a_{j_0 - j} + a_{j_0-j+1} &= \sum_{i=1}^b m_{s+(j_0-s-1)(b+1)+i} + m_{s+(j_0-s)(b+1)+i}\nonumber\\
	&\overset{\mathrm{\eqref{Eqn:M_gather_overlap}}}{=} \sum_{i=1}^{2b+1} m_{s+(j_0-s-1)(b+1)+i} \overset{\mathrm{\eqref{Eqn:A_one_sequence}, \eqref{Eqn:J_b+1_gather_overlap}}}{\geq} b+1 \label{Eqn:sum_a_seq}
	\end{align}
	for $j_0 \in \{ j, j+1, \cdots, j+\delta-1 \}$. Since $a_t$ is the sum of $b$ binary elements, we have 
	\begin{align}\label{Eqn:a_seq_constraint}
	1 \leq a_t \leq b , \quad \forall t \in \{0, 1, \cdots, \delta\}
	\end{align}
	from \eqref{Eqn:sum_a_seq}.
	Now define a message vector $\vct{m}' \in M_{gather, overlap}^{(q)}$ satisfying the followings: the corresponding assignment vector is 
	\begin{align}\label{Eqn:a_prime_seq}
	(a_0', a_1', \cdots, a_{\delta}') &= \begin{cases}
	(1,b,1,b,\cdots, 1,b) , & \text{ if } \delta \text{ is odd} \\
	(1,b,1,b,\cdots, 1) , & \text{ otherwise},
	\end{cases}
	\end{align}
	for $a_l' = \sum_{i=1}^b m_{s+(j-s-1+l)(b+1)+i}'$,
	and the elements $m_i'$ for $i \in [s]$ is
	$m_i' = \mathds{1}_{ \{i \leq i_{max} \}} $
	where 
	$i_{max} = \frac{n-1}{2} - \sum_{l=0}^{\delta} a_l' \leq s.$
	Then, we have
	\begin{align}\label{Eqn:optimum_message_vector}
	v(\vct{m}) \geq \sum_{l=0}^{\delta} a_l \overset{\mathrm{\eqref{Eqn:sum_a_seq}, \eqref{Eqn:a_seq_constraint}}}{\geq} \sum_{l=0}^{\delta} a_l' \overset{\mathrm{\eqref{Eqn:a_prime_seq}
	}}{=} v(\vct{m}')
	\end{align}
	for arbitrary $\vct{m} \in M_{gather, overlap}^{(q)}$.
	Moreover, among $\delta \geq q$, setting $\delta=q$ minimizes $v(\vct{m}')$, having the optimum value of
	\begin{align*}
	v^*_q \overset{\mathrm{(a)}}{=} v(\vct{m}') &\overset{\mathrm{\eqref{Eqn:a_prime_seq}}}{=} \begin{cases}
	\sum_{i=1}^{\frac{q+1}{2}} (1+b), & \text{ if } q \text{ is odd} \\
	1 + \sum_{i=1}^{\frac{q}{2}} (1+b), & \text{ otherwise }\end{cases}  \overset{\mathrm{(b)}}{\geq} \begin{cases}
	1+b+ 2 (\frac{q+1}{2}-1)  = b+q, & \text{ if } q \text{ is odd} \\
	1 + (1+b) + 2 (\frac{q}{2}-1) = b+q , & \text{ otherwise }\end{cases}   
	\end{align*}
	where $(a)$ is from \eqref{Eqn:optimum_message_vector} and Lemma \ref{Lemma:gather_overlap}, and $(b)$ is from $b\geq 1$. Combining this with the definition of $v^*_q$ in \eqref{Eqn:optimal_v_def} proves \eqref{Eqn:J_b+1_alternative}. This completes the proofs for \eqref{Eqn:Byz_Lemma_supp_2} and \eqref{Eqn:Byz_Lemma_supp}.
	Thus, the data allocation matrix $\mtx{G}$ in Algorithm \ref{Alg:DetCode} perfectly tolerates $b$ Byzantines. From Fig. \ref{Fig:Generator_matrix}, the required redundancy of this code is 
	\begin{align*}
	r &= \frac{s + (2b+1)L + n (n-s-L)}{n} \overset{\mathrm{(a)}}{=} \frac{n-(2b+1)}{2n} + \frac{2b+1}{n}L + \left(\frac{n+(2b+1)}{2} - L \right) \\
	&= \frac{n+(2b+1)}{2}- \left(L-\frac{1}{2}\right) \frac{n-(2b+1)}{n}, 
	\end{align*}
	where Eq.$(a)$ is from the definition of $s$ in Algorithm \ref{Alg:DetCode}.	
	\fi
\end{proof}

\section{Proof of Lemmas and Propositions}

\subsection{Proof of Lemma~\ref{Thm:q_k_updated}}\label{Sec:Proof_of_Thm:q_k_updated}

We start from finding the estimation error $q_{k \vert n_i}$ of an arbitrary node $i$ having $n_i$ data partitions.

\begin{lemma}[conditional local error]\label{Lemma:q_i_bound_conditional}
	Suppose $n_i$ data partitions are assigned to a Byzantine-free node $i$. Then, the probability of this node transmitting a wrong sign to PS for coordinate $k$ is bounded as
	\begin{align}\label{Eqn:node_fail_prob_bound}
	q_{k \vert n_i}= \PP{c_{i,k} \neq \sign(g_k) | n_i} \leq 
	\exp{(-n_i S_k^2/\{ 2(S_k^2 + 4)\})}.   
	\end{align}
\end{lemma}

This lemma is proven by applying Hoeffding's inequality for binomial random variable, as shown in Section~\ref{Section:Proof_of_Lemma:q_i_bound_conditional} of the Supplementary Materials. 
The remark below summaries the behavior of the local estimation error as the number of data partitions assigned to a node increases.

\begin{remark}
	The error bound in Lemma~\ref{Lemma:q_i_bound_conditional} is an exponentially decreasing function of $n_i$.
	This implies that as the number of data partitions assigned to a node increases, it is getting highly probable that the node correctly estimates the sign of true gradient.
	This supports the experimental results in Section 4 showing that random Bernoulli codes with a small amount of redundancy (e.g. $\EE{r}=2, 3$) are enough to enjoy a significant gap compared to the conventional SignSGD-MV~\cite{bernstein2018signsgd_ICLR} with $r=1$.
\end{remark}

Note that $n_i \sim \mathcal{B}(n,p)$ is a binomial random variable. Based on the result of Lemma~\ref{Lemma:q_i_bound_conditional}, we obtain the local estimation error $q_k$ by averaging out $q_{k \vert n_i}$ over all realizations of $n_i$.
Define $\epsilon = p^{\star}/2 = \sqrt{C \log(n)/n}$.
Then, we calculate the failure probability of node $i$ as
\begin{align*}
q_{k}  &= \PP{c_{i,k} \neq \sign(g_k)} = \sum_{n_i = 0}^n \PP{n_i} \PP{c_{i,k} \neq \sign(g_k) \ \big\vert \  n_i }\\ 
&= \sum_{n_i: \lvert n_i - np \rvert \geq n\epsilon }  \PP{n_i} \PP{c_{i,k} \neq \sign(g_k) \ \big\vert \ n_i } + \sum_{n_i: \lvert n_i - np \rvert < n\epsilon }  \PP{n_i} \PP{c_{i,k} \neq \sign(g_k) \ \big\vert \ n_i } \\
& \stackrel{(a)}{\leq} \sum_{n_i: \lvert n_i - np \rvert \geq n\epsilon }  \PP{n_i} + 
\PP{c_{i,k} \neq \sign(g_k) \ \big\vert \ n_i = n (p -\epsilon) } \\
& \stackrel{(b)}{\leq} 2 {\rm e}^{-2 \epsilon^2 n }  +  {\rm e}^{-n (p-\epsilon) \frac{S_k^2}{2(S_k^2+4)}} \leq  \frac{2}{n^{2C}} +  {\rm e}^{- \sqrt{ C n \log(n)} \frac{S_k^2}{2(S_k^2+4)}} \leq q_k^{\star}  , 
\end{align*}
where (a) is from the fact that the upper bound in \eqref{Eqn:node_fail_prob_bound} is a decreasing function of $n_i$, and (b) holds from Lemma~\ref{Lemma:q_i_bound_conditional} and Lemma~\ref{Lemma:Hoeffding_binomial}, the Hoeffding's inequality on the binomial distribution.

\subsection{Proof of Lemma~\ref{Lemma:NSCondition}}\label{Sec:Proof_of_Lemma:NSCondition}
Let an attack vector $\vct{\beta}$ and an attack function $f_{\vct{\beta} } (\cdot)$ given. Consider an arbitrary $ \vct{m} \in M^+ $. From the definitions of $\mu$ and $\hat{\mu}$, we have $\mu=\hat{\mu}$ iff 
$ f_{\vct{\beta}}(\phi(\vct{m}) ) \in Y^+$. Similarly, for an arbitrary $ \vct{m} \in M^- $, we have $\mu=\hat{\mu}$ iff $ f_{ \vct{\beta}}(\phi(\vct{m}) ) \in Y^-$. 
Thus, from the definitions of $Y^+$ and $Y^-$, the sufficient and necessary condition for $b-$Byzantine tolerance
can be expressed as follows.

\begin{prop}\label{Prop:Byz_intermediate1}
	The \textit{perfect} $b-$Byzantine tolerance condition is equivalent to the following:
	$\forall \vct{\beta} \in B_b, \forall f_{\vct{\beta} } \in \mathcal{F}_{\vct{\beta} }$,
	\begin{align}\label{Eqn:Byz_intermediate1}
	\begin{cases}
	\ \ \ZNorm{f_{\vct{\beta} }(\phi(\vct{m}) )} > \floor{\frac{n}{2}}, & \forall \vct{m} \in M^+ \\
	\ \ \ZNorm{f_{\vct{\beta} }(\phi(\vct{m}) )} \leq  \floor{\frac{n}{2}}, & \forall \vct{m} \in M^- 
	\end{cases}
	\end{align}
\end{prop}

The condition stated in Proposition \ref{Prop:Byz_intermediate1} can be further simplified as follows.

\begin{prop}\label{Prop:Byz_intermediate2}
	The \textit{perfect} $b-$Byzantine tolerance condition in Proposition \ref{Prop:Byz_intermediate1} is equivalent to 
	\begin{align}\label{Eqn:Byz_intermediate2}
	\begin{cases}
	\ \ \ZNorm{\phi(\vct{m})} > \floor{\frac{n}{2}} + b, & \forall \vct{m} \in M^+ \\
	\ \ \ZNorm{\phi(\vct{m})} \leq \floor{\frac{n}{2}} - b, & \forall \vct{m} \in M^-
	\end{cases}
	\end{align}
\end{prop}
\begin{proof}
	Consider arbitrary $\vct{m} \in M^-$. We want to prove that
	\begin{align}\label{Eqn:Equivalence_1}
	\forall \vct{\beta} \in B_b, \forall f_{\vct{\beta} } \in \mathcal{F}_{\vct{\beta} },  \quad \ZNorm{f_{\vct{\beta} }(\phi(\vct{m}) )} \leq  \floor{\frac{n}{2}}
	\end{align}
	is equivalent to 
	\begin{align}\label{Eqn:Equivalence_2}
	\ZNorm{\phi(\vct{m})} \leq \floor{\frac{n}{2}} - b.
	\end{align}
	First, we show that \eqref{Eqn:Equivalence_2} implies \eqref{Eqn:Equivalence_1}. 
	According to Lemma~\ref{Prop:ByzantineAttack}
	$\ZNorm{	f_{\vct{\beta} }(\phi(\vct{m}) )  \oplus \phi(\vct{m}) } \leq b$
	holds for  arbitrary $\vct{\beta} \in B_b$ and arbitrary $f_{\vct{\beta} } \in \mathcal{F}_{\vct{\beta} }$.
	Thus, 
	\begin{align*}
	\ZNorm{f_{\vct{\beta} }(\phi(\vct{m}) )} &\leq \ZNorm{	f_{\vct{\beta} }(\phi(\vct{m}) )  \oplus \phi(\vct{m}) } + \ZNorm{\phi(\vct{m}) }
	\leq b + \left( \floor{\frac{n}{2}} - b \right) = \floor{\frac{n}{2}}
	\end{align*}
	holds for $\forall \vct{\beta} \in B_b, \forall f_{\vct{\beta} } \in \mathcal{F}_{\vct{\beta} }$, which completes the proof.	
	Now, we prove that \eqref{Eqn:Equivalence_1} implies \eqref{Eqn:Equivalence_2}, by contra-position. 
	Suppose $\ZNorm{\phi(\vct{m})} > \floor{\frac{n}{2}} - b$. We divide the proof into two cases. The first case is when $\ZNorm{\phi(\vct{m})} > n-b$. In this case, we arbitrary choose $\vct{\beta}^* \in B_b$ and select the identity mapping $f_{\vct{\beta}^*}^*: \vct{c} \mapsto \vct{y}$ such that $y_j = c_j$ for all $j \in [n]$. Then,  
	$\ZNorm{f_{\vct{\beta}^* }^*(\phi(\vct{m}) )}=\ZNorm{\phi(\vct{m}) } > n-b \geq n- \floor{n/2} \geq \floor{n/2}$. Thus, we can state that
	\begin{equation*}
	\exists \vct{\beta}^*\in B_b, \exists f_{\vct{\beta}^*}^* \in \mathcal{F}_{\vct{\beta}^*} \text{ such that } \ZNorm{f_{\vct{\beta}^* }^*(\phi(\vct{m}) )} \leq  \floor{\frac{n}{2}}
	\end{equation*}
	when $\ZNorm{\phi(\vct{m})} > n - b$, which completes the proof for the first case.
	Now consider the second case where $ \floor{n/2} - b < \ZNorm{\phi(\vct{m})} \leq n - b$. To begin, denote $\phi(\vct{m}) =  \vct{c} = [c_1, c_2, \cdots, c_n]$. Let $S = \{i \in [n]: c_i = 0 \}$, and select $\vct{\beta}^* \in B_b$ which satisfies\footnote{We can always find such $\vct{\beta}^*$ since $\lvert S \rvert \geq b$ due to the setting of $\ZNorm{\phi(\vct{m})} \leq n - b$. }
	$\{ i\in [n]: \beta_i^* = 1 \} \subseteq S$. Now define 
	$f_{\vct{\beta}^*}^* (\cdot)$ as $f_{\vct{\beta}^*}^* (\phi(\vct{m})) = \phi(\vct{m}) \oplus \vct{\beta}^*$.
	Then, we have
	\begin{align*}
	\ZNorm{f_{\vct{\beta}^*}^* (\phi(\vct{m}))} =  \ZNorm{\phi(\vct{m})} + \ZNorm{\vct{\beta}^*} > \floor{\frac{n}{2}} - b + b = \floor{\frac{n}{2}}.
	\end{align*}
	Thus, the proof for the second case is completed, and this completes the statement of \eqref{Eqn:Byz_intermediate2} for arbitrary $\vct{m} \in M^-$. Similarly, we can show that 
	\begin{align*}
	\forall \vct{\beta} \in B_b, \forall f_{\vct{\beta} } \in \mathcal{F}_{\vct{\beta} },  \quad \ZNorm{f_{\vct{\beta} }(\phi(\vct{m}) )} >  \floor{\frac{n}{2}}
	\end{align*}
	is equivalent to 
	$\ZNorm{\phi(\vct{m})} > \floor{\frac{n}{2}} + b$
	for arbitrary $\vct{m} \in M^+$. This completes the proof. 
\end{proof}

Now, we further reduce the condition in Proposition \ref{Prop:Byz_intermediate2} as follows.
\begin{prop}\label{Prop:Byz_intermediate3}
	The \textit{perfect} $b-$Byzantine tolerance condition in Proposition \ref{Prop:Byz_intermediate2} is equivalent to 
	\begin{align}\label{Eqn:Byz_intermediate3}
	\ \ \ZNorm{\phi(\vct{m})} \leq \floor{\frac{n}{2}} - b, \quad \forall \vct{m} \in M^-
	\end{align}
\end{prop}
\begin{proof}
	All we need to prove is that \eqref{Eqn:Byz_intermediate3} implies \eqref{Eqn:Byz_intermediate2}.
	Assume that the mapping $\phi$ satisfies \eqref{Eqn:Byz_intermediate3}. 
	Consider an arbitrary $\vct{m}' \in M^+$ and denote $\vct{m}' =[m_1', m_2', \cdots, m_n']$. Define $\vct{m}=[m_1, m_2, \cdots, m_n]$ such that $m_i' \oplus m_i = 1$ for all $i \in [n]$. Then, we have $\vct{m} \in M^-$ from the definitions of $M^+$ and $M^-$. 
	Now we denote $ \phi(\vct{m}) =\vct{c} = [c_1, c_2, \cdots, c_n]$ and  $\phi(\vct{m}') =\vct{c}' =  [c_1', c_2', \cdots, c_n']$. Then,  $c_j \oplus c_j' = 1$ holds for all $j \in [n]$ since $E_j (\cdot)$ is a majority vote function\footnote{Recall that $c_j = E_j (\{m_i\}_{i \in P_j}) = \maj(\{m_i\}_{i \in P_j})$ and $c_j' = \maj(\{m_i'\}_{i \in P_j})$. Thus, $m_i' \oplus m_i = 1$ for all $i \in [n]$ implies that $c_j \oplus c_j' = 1$ holds for all $j \in [n]$.}. 
	In other words, 
	$\ZNorm{\phi(\vct{m})} + \ZNorm{\phi(\vct{m}')} = n$
	holds. Thus, if a given mapping $\phi$ satisfies
	$\ZNorm{\phi(\vct{m})} \leq \floor{n/2} - b$ for all $\vct{m} \in M^-,$
	then  
	$\ZNorm{\phi(\vct{m}')} \geq n - (\floor{n/2} - b) = \ceil{n/2} + b > \floor{n/2} + b$ holds for all $\vct{m} \in M^+,$
	which completes the proof.
\end{proof}

In order to prove Lemma~\ref{Lemma:NSCondition}, all that remains is to prove that \eqref{Eqn:Byz_intermediate3} reduces to 
\begin{align}\label{Eqn:Thm_alg_revisited}
\sum_{v=1}^{\floor{n/2}} \lvert S_v (\vct{m}) \rvert  \leq \floor{n/2} - b \quad \quad \forall \vct{m} \in M_{\floor{n/2}}.
\end{align}
Recall that $\phi(\vct{m}) = \vct{c} = [c_1, c_2, \cdots, c_n]$ where $c_j = \maj(\{m_i\}_{i \in P_j} )$ and $P_j = \{ i \in [n]: G_{ji} = 1 \}$. Moreover, we assumed that $\lvert P_j \rvert = \ZNorm{\mtx{G}(j,:)}$ is an odd number.
Thus, 
$c_j = 
\mathds{1}_{ \{ \ZNorm{\mtx{G}(j,:)}+1 \ \leq \ 2  \vct{m}^T \mtx{G}(j,:)  \}  }  $,
and the set $[n]=\{1,2,\cdots, n\}$ can be partitioned as $[n] = S_1 \cup S_2 \cup \cdots \cup S_{\floor{n/2} + 1}$ where
$S_v \coloneqq \{j \in [n]: \ZNorm{\mtx{G}(j,:)} = 2v-1 \}.$
Therefore, for a given $\vct{m} \in M^-$, we have
\begin{align*}
\ZNorm{\phi(\vct{m})} &= \sum_{j=1}^n c_j = \sum_{v=1}^{\floor{n/2} + 1} \left\lvert \{ j\in S_v: c_j = 1 \} \right\rvert \nonumber \\
&= \sum_{v=1}^{\floor{n/2} + 1} \left\lvert  \left\{ j \in S_v: \vct{m}^T \mtx{G}(j,:) \geq \frac{\ZNorm{\mtx{G}(j,:)} + 1}{2} +1 =v \right\}   \right\rvert = \sum_{v=1}^{\floor{n/2} + 1} \left\lvert S_v  (\vct{m}) \right\rvert. %
\end{align*}
Note that $S_{v}  (\vct{m})$ for $v= \floor{n/2}+1$ reduces to 
\begin{align*}%
S_{ \floor{n/2}+1}   (\vct{m}) = \{ j \in [n]: \ZNorm{\mtx{G}(j,:)} = 2 ( \floor{n/2} - 1) +1, \vct{m}^T \mtx{G}(j,:) \geq \floor{n/2} + 1 \} = \varnothing
\end{align*}
since $\vct{m} \in M^-$.
Thus, combining the two equations above, we obtain the following.

\begin{prop}\label{Prop:Byz_intermediate4}
	The perfect $b-$Byzantine tolerance condition in Proposition \ref{Prop:Byz_intermediate3} is equivalent to
	\begin{align*}%
	\sum_{v=1}^{\floor{n/2}} \left\lvert S_v  (\vct{m}) \right\rvert \leq \floor{\frac{n}{2}}-b \quad \quad \forall \vct{m} \in M^-,
	\end{align*}
	or equivalently,
	\begin{align}\label{Eqn:Byz_intermediate4_equiv}
	\sum_{v=1}^{\floor{n/2}} \left\lvert S_v  (\vct{m}) \right\rvert \leq \floor{\frac{n}{2}}-b \quad \quad \forall \vct{m} \in M_t, \quad \forall t = 0,1,\cdots, \floor{n/2}.
	\end{align}
\end{prop}

Now, we show that \eqref{Eqn:Byz_intermediate4_equiv} is equivalent to \eqref{Eqn:Thm_alg_revisited}. We can easily check that the former implies the latter, which is directly proven from the statements. Thus, all we need to prove is that \eqref{Eqn:Thm_alg_revisited} implies \eqref{Eqn:Byz_intermediate4_equiv}.
First, when $t=0$, note that $\lvert S_v (\vct{m}) \rvert = 0 $ for $\forall \vct{m} \in M_0, \forall v \in \{ 1,2, \cdots, \floor{n/2} \}$, which implies that \eqref{Eqn:Byz_intermediate4_equiv} holds trivially. Thus, in the rest of the proof, we assume that $t > 0$.

Consider an arbitrary $t \in \{1,2, \cdots \floor{n/2} \}$ and an arbitrary $\vct{m} \in M_t$. Denote $\vct{m} = \vct{e}_{i_1} +  \vct{e}_{i_2} + \cdots +  \vct{e}_{i_t}$ where $\vct{e}_1 = [1, 0, \cdots, 0]$, $\vct{e}_2 = [0, 1, 0, \cdots, 0]$, and $\vct{e}_n = [0, \cdots, 0, 1]$.
Moreover, consider an arbitrary $\vct{m}' \in M_{\floor{n/2}}$ which satisfies $m_i' = 1$ for $i=i_1, i_2, \cdots, i_t$. Denote $\vct{m}' = \vct{e}_{i_1} + \cdots + \vct{e}_{i_t} + \vct{e}_{j_1} + \cdots + \vct{e}_{j_{\floor{n/2} - t}}$. Then, 
$(\vct{m}' - \vct{m})^T \mtx{G}(j,:) \geq 0$ holds for all $j \in [n]$,
which implies $S_v (\vct{m}) \subseteq S_v (\vct{m}')$ for all $v = 1, 2, \cdots, \floor{n/2}$. Thus, we have
$\lvert S_v (\vct{m}) \rvert  \leq \lvert S_v (\vct{m}') \rvert$ for all $v \in \{ 1,2,\cdots, \floor{n/2} \},$
which implies 
$\sum_{v=1}^{\floor{n/2}} \lvert S_v (\vct{m}) \rvert  \leq  \sum_{v=1}^{\floor{n/2}} \lvert S_v (\vct{m}') \rvert .$
Since this holds for arbitrary $\vct{m}' \in M_{\floor{n/2}}$,  $\vct{m} \in M_t$, and  $t\ \in \{ 1,2, \cdots, \floor{n/2} \}$,  
we can conclude that \eqref{Eqn:Thm_alg_revisited} implies \eqref{Eqn:Byz_intermediate4_equiv}.
All in all, \eqref{Eqn:Byz_intermediate4_equiv} is equivalent to \eqref{Eqn:Thm_alg_revisited}. Combining this with Propositions \ref{Prop:Byz_intermediate1},\ref{Prop:Byz_intermediate2}, \ref{Prop:Byz_intermediate3} and \ref{Prop:Byz_intermediate4}
completes the proof of Lemma~\ref{Lemma:NSCondition}.

\subsection{Proof of Lemma~\ref{Lemma:q_i_bound_conditional}}\label{Section:Proof_of_Lemma:q_i_bound_conditional}
From Lemmas~\ref{Lemma:symm_unimodal_partition} and~\ref{Lemma:node_failure_ICML18}, we have
\begin{align}\label{Eqn:q_k^idv}
q_k^{(idv)} &= \PP{\sign(\tilde{g}_k^{(j)}) \neq \sign(g_k)} \leq 
\left\{\begin{array}{ll}
\frac{2}{9} \frac{1}{S_k^{2}} & \text { if } S_k >\frac{2}{\sqrt{3}} \\
\frac{1}{2}-\frac{S_k}{2 \sqrt{3}} & \text { otherwise }
\end{array}\right.
\end{align}
for arbitrary $j\in [n], k \in [d]$ where $S_k = \frac{\left|g_{k}\right|}{\bar{\sigma}_{k}}$ is defined in Definition~\ref{Def:S_k}.
Denote the set of data partitions assigned to node $i$ by $P_i = \{j_1, j_2, \cdots, j_{n_k} \}$. Define a random variable $X_s$ as 
\begin{align*}
X_s = \mathds{1}_{\sign(\tilde{g}_k^{(j_s)}) = \sign(g_k)}.
\end{align*}
Then, from the definition of $q_k^{(idv)}$ in \eqref{Eqn:q_k^idv}, we have $\PP{X_s = 1}=  p_k^{(idv)} \coloneqq 1-q_k^{(idv)}$, and $\PP{X_s = 0} = q_k^{(idv)}$.
Recall that $c_{i,k} = \maj \{ \sign  ( \tilde{g}_k^{(j)} ) \}_{j \in P_i}$. By using a new random variable defined as $X \coloneqq \sum_{s=1}^{n_i} X_s$, the failure probability of node $i$ estimating the sign of $g_k$ is represented as
\begin{align*}
\PP{c_{i,k} \neq \sign(g_k) | n_i} &= \PP{X \leq \frac{n_i}{2}} = \PP{X - n_i p_k^{(idv)} \leq -n_i (-\frac{1}{2} + p_k^{(idv)})} \\
&\stackrel{(a)}{\leq} e^{-2 (-\frac{1}{2} + p_k^{(idv)})^2 n_i} \stackrel{(b)}{\leq} e^{-n_i \frac{S_k^2}{2(S_k^2 + 4)}}
\end{align*}
where (a) is from Lemma~\ref{Lemma:Hoeffding_binomial} and (b) is from the fact that $\frac{1}{4 \left(-\frac{1}{2} + p_k^{(idv)} \right)^2} -1 \leq \frac{4}{S_k^2}$, which is shown as below.
Note that
\begin{align*}
-\frac{1}{2} + p_k^{(idv)} &= \frac{1}{2} - q_k^{(idv)} \geq 
\left\{\begin{array}{ll}
\frac{1}{2} - \frac{2}{9} \frac{1}{S_k^{2}} & \text { if } S_k >\frac{2}{\sqrt{3}} \\
\frac{S_k}{2 \sqrt{3}} & \text { otherwise }
\end{array}\right.
\end{align*}
When $S_k \leq \frac{2}{\sqrt{3}}$, we have 
$\frac{1}{4 \left(-\frac{1}{2} + p_k^{(idv)} \right)^2} -1 \leq \frac{3}{S_k^2} - 1 < \frac{4}{S_k^2}$. For the case of $S_k > \frac{2}{\sqrt{3}}$, we have
\begin{align*}
\frac{1}{4 \left(-\frac{1}{2} + p_k^{(idv)} \right)^2} -1 \leq  \frac{1}{S_k^{2}} \frac{\frac{8}{9}-\frac{16}{81} \frac{1}{S_k^{2}}}{1-\frac{8}{9} \frac{1}{S_k^{2}}+\frac{16}{81} \frac{1}{S_k^{4}}}<\frac{1}{S_k^{2}} \frac{\frac{8}{9}}{1-\frac{8}{9} \frac{1}{S_k^{2}}}<\frac{4}{S_k^{2}}
\end{align*}
where the last inequality is from the condition on $S_k$. This completes the proof of Lemma~\ref{Lemma:q_i_bound_conditional}.

\ifarxiv

\subsection{Proof of Proposition~\ref{Lemma:symm_unimodal_partition}}\label{Section:Proof_of_Lemma:symm_unimodal_partition}

Let $u_k^{(j)} = \tilde{g}_k^{(j)} - g_k$.
From the definition of $\tilde{g}_k^{(j)}$, we have 
$Y = B u_k^{(j)} = B (\tilde{g}_k^{(j)} - g_k) = \sum_{x \in B_j} u_k(x)$. 
From Lemma~\ref{Lemma:conv_independent}, 
$f_Y = \textrm{conv } \{ f_{u_k(x)}  \}_{x \in B_j}$.
Since $u_k(x)$ are zero-mean, symmetric, and unimodal from Assumption~\ref{Assumption:symmtric_unimodal}, Lemma~\ref{Lemma:conv_unimodal} implies that  $Y$ (and thus $u_k^{(j)}$) is also zero-mean, symmetric, and unimodal. Therefore, $\tilde{g}_k^{(j)} = g_k + u_k^{(j)}$ is unimodal and symmetric around the mean $g_k$. The result on the variance of $\tilde{g}_k^{(j)}$ is directly obtained from the independence of $\tilde{g}_k(x)$ for different $x \in B_j$.

\newpage 

\small
\bibliographystyle{plain}
\bibliography{NeurIPS_20_Election_Coding}

\end{document}